\documentclass[runningheads]{llncs}

\usepackage{graphicx}
\usepackage{stmaryrd,mathtools}
\usepackage{todonotes}
\usepackage{mathpartir,amssymb}
\usepackage{graphicx}
\usepackage{tikz}
\usetikzlibrary{arrows,automata,positioning,decorations,shapes}

\newcommand{\N}{\mathbb{N}}

\definecolor{mygray}{gray}{0.8}
\definecolor{light-gray}{gray}{0.9}
\newcommand{\hilight}[1]{\colorbox{light-gray}{#1}}
\newcommand{\mhilight}[1]{\mathrm{\hilight{$#1$}}}

\newcommand{\gjadd}[1]{{#1}}

\newcommand\ex[2]{#1^\mathit{#2}}
\newcommand\ind[1]{\mathcal{I}_{#1}} 
\newcommand\ass[2]{\mathit{assert}(#1\sim #2)}
\newcommand\tick{\mathit{time}}
\newcommand\pav[1]{\mathsf{Vis}_{{P}}(#1)}
\newcommand\oav[1]{\mathsf{Vis}_{{O}}(#1)}

\newcommand\krbs[3]{\mathcal{K}^{#3}\llbracket{#1,#2}\rrbracket}
\newcommand\vrbs[2]{\mathcal{V}^{#2}\llbracket{#1}\rrbracket}

\newcommand{\mapstoi}[1]{\mapsto}

\newcommand\rarr\rightarrow

\newcommand{\amin}[1]{\todo[inline,color=green!10]{AM: #1}}

\newcommand{\boldemph}[1]{\emph{\textbf{#1}}}
\newcommand{\defeq}{\triangleq}
\newcommand{\sep}{~|~}
\newcommand{\dom}[1]{\mathrm{dom}(#1)}

\newcommand\img[1]{\mathrm{img}(#1)}

\newcommand\ctxter[1]{\lesssim_\mathit{ter}^{#1}}
\newcommand\ctxerr[1]{\lesssim_\mathit{err}^{#1}}

\newcommand\ciupre[2]{\lesssim_\mathit{#2}^{#1(\mathit{ciu})}}
\newcommand\ctxpre[2]{\lesssim_\mathit{#2}^{#1}}

\newcommand\ctxequ[2]{\cong_\mathit{#2}^{#1}}

\newcommand{\red}{\rightarrow}
\newcommand{\ered}{\rightarrow}
\newcommand{\subst}[2]{\{#2 / #1 \}}
\newcommand{\substF}[1]{\{#1\}}

\newcommand{\ciuapperr}{\lesssim_\mathit{ciu,err}}

\newcommand{\ciuapprox}{\lesssim_\mathit{ciu}}

\newcommand{\Int}{\mathrm{Int}}
\newcommand{\Unit}{\mathrm{Unit}}
\newcommand{\Bool}{\mathrm{Bool}}
\newcommand{\reftype}[1]{\mathrm{ref} #1}
\newcommand{\conttype}[1]{\mathrm{cont}\ #1}

\newcommand{\unit}{()}
\newcommand{\nb}[1]{\widehat{#1}}
\newcommand{\trueML}{\mathbf{tt}}
\newcommand{\falseML}{\mathbf{ff}}
\newcommand{\fun} [2] {\lambda #1. #2}
\newcommand{\fix} [4] {\mathbf{rec} \ #1^{#2}(#3) .#4}

\newcommand \typingTerm[3] {#1 \vdash #2 : #3 }

\newcommand \seq[2] {#1 \vdash #2}
\newcommand \cseq[3] {#1 \vdash #2 \div #3}
\newcommand{\err}{\mathit{err}}
\newcommand{\ter}{\mathit{ter}}

\newcommand{\opredobs}[1]{\Downarrow_{\mathit{#1}}}
\newcommand{\ltsobs}[2]{\Downarrow_{#1}^{#2}}
\newcommand{\semobs}[2]{\downarrow_{#1}^{#2}}
\newcommand{\opredtererr}{\opredobs{\err}}
\newcommand{\ltstererr}[1]{\ltsobs{\err}{#1}}

\newcommand{\opredter}{\opredobs{\ter}}
\newcommand{\ltster}[1]{\ltsobs{\ter}{#1}}
\newcommand{\semter}[1]{\semobs{\ter}{#1}}
\newcommand\herr{\hat{\err}}

\newcommand \pair [2]{\langle #1,#2\rangle}
\newcommand \tuples [1]{\langle #1\rangle}
\newcommand \ifte [3] {\mathrm{if}\ #1\ #2\ #3}
\newcommand \nuref[1]{\mathrm{ref}\,#1}
\newcommand \nureft[2]{\mathrm{ref}_{#1}\,{#2}}
\newcommand \cont[1]{\mathrm{cont}\,#1}
\newcommand \contt[2]{\mathrm{cont}_{#1}\,{#2}}
\newcommand \throw {\mathrm{throw}}
\newcommand \throwto[2] {\throw \ \mathrm{#1} \ \mathrm{to} \ #2}
\newcommand \throwtot[3]{\throw_{#1} \ \mathrm{#2} \ \mathrm{to} \ #3}
\newcommand \callcc {\mathrm{call/cc}}
\newcommand \callcct[1]{\mathrm{call/cc}_{#1}}

\newcommand\letin[2]{\mathrm{let}\,#1\,\mathrm{in}\,#2}

\newcommand \Heap {\mathrm{Heap}}
\newcommand \BLoc {\mathbf{Loc}}
\newcommand \Var {\mathbf{Var}}

\newcommand{\HOSC}{{\mathrm{HOSC}}}
\newcommand{\HOS}{\mathrm{HOS}}
\newcommand{\GOS}{\mathrm{GOS}}
\newcommand{\GOSC}{\mathrm{GOSC}}
\newcommand{\FOSC}{\mathrm{FOSC}}

\newcommand{\emptymap}{\cdot}
\newcommand{\pmap}{\rightharpoonup}

\newcommand{\emptyheap}{\epsilon}

\newcommand{\Values}{\mathrm{Vals}}

\newcommand{\EContexts}{\mathrm{ECtxs}}

\newcommand{\AVal}[2]{\mathbf{AVal}_{#2}(#1)}

\newcommand\num[1]{\mathsf{Num}(#1)}
\newcommand\fnam[1]{\mathsf{Fun}(#1)}

\newcommand{\FNames}{\mathrm{FNames}}
\newcommand{\CNames}{\mathrm{CNames}}
\newcommand{\Names}{\mathrm{Names}}

\newcommand{\act}{\mathbf{a}}

\newcommand{\ansP}[2]{\bar{#1}(#2)}
\newcommand{\questP}[3]{\bar{#1}(#2,#3)}
\newcommand{\ansO}[2]{#1(#2)}
\newcommand{\questO}[3]{#1(#2,#3)}

\newcommand\perm[3]{#1\sim_{#3} #2}

\newcommand{\emptytrace}{\epsilon}
\newcommand{\tr}{{\tt t}}
\newcommand\errn{\diamond}
\newcommand\tern{\circ}
\newcommand{\finalName}{\tern}
\newcommand\cutout[1]{}
\newcommand\cconf[2]{\mathsf{C}_{#1}^{#2}}
\newcommand{\CC}{\mathbf{C}}
\newcommand{\CCc}{\mathbf{D}}

\newcommand{\conf}[1]{\langle #1\rangle}
\newcommand{\ired}[1]{\xrightarrow{#1}}
\newcommand{\iRed}[1] {\xRightarrow{#1}}

\newcommand{\Tr}[1]{\mathbf{Tr}(#1)}
\newcommand{\TrR}[2]{\mathbf{Tr}_{#1}(#2)}
\newcommand{\Treven}[1]{\mathbf{Tr}^{\mathrm{even}}(#1)}
\newcommand{\TrReven}[2]{\mathbf{Tr}^{\mathrm{even}}_{#1}(#2)}
\newcommand\trsem[2]{{\TrR{#1}{#2}}}
\newcommand\rest[2]{#1\upharpoonright #2}
\newcommand{\orderCont}[1]{\prec_{#1}}

\newcommand{\View}{\mathcal{V}}
\newcommand{\ViewF}{\mathcal{F}}

\def\doublewedge{\mathrel{\wedge\!\!\!\wedge}}
\newcommand{\mergeConf} [2] {#1 \doublewedge #2}

\newcommand\nfp{n_{\mathit{FP}}}
\newcommand\ncp{n_{\mathit{CP}}}
\newcommand\nfo{n_{\mathit{FO}}}
\newcommand\nco{n_{\mathit{CO}}}

\newcommand{\support}{\nu}

\newcommand{\xx}{\mathbf{x}}
\newcommand{\yy}{\mathbf{y}}

\newcommand{\inccwl}{\mathrm{inc}}
\newcommand{\getcwl}{\mathrm{get}}

\newcommand\fpr[1]{\mathit{fpr}_{#1}}
\newcommand\cpr[1]{\mathit{cpr}_{#1}}
\newcommand\foor[1]{\mathit{for}_{#1}}
\newcommand\cor[1]{\mathit{cor}_{#1}}
\newcommand\pr[1]{\mathit{pr}_{#1}}
\newcommand\oor[1]{\mathit{for}_{#1}}
\newcommand\kkr[1]{\mathit{kr}_{#1}}
\newcommand\ccr[1]{\mathit{cr}_{#1}}

\newcommand\topo[1]{\mathit{top}_O(#1)}
\newcommand\topp[1]{\mathit{top}_P(#1)}

\begin{document}
\title{Complete trace models of state and control}
%
%
\author{Guilhem Jaber\inst{1} \and
Andrzej S. Murawski\inst{2}}
\authorrunning{G.~Jaber\and A.~S.~Murawski}
%
\institute{Universit{\'e} de Nantes, LS2N CNRS, Inria, France\and University of Oxford, UK}
\maketitle              
\begin{abstract}
We consider a hierarchy of four typed call-by-value languages with either higher-order or ground-type references and with either $\callcc$ or no control operator.

Our first result is a fully abstract trace model for the most expressive setting, featuring both higher-order references and $\callcc$, constructed in the spirit of operational game semantics.
Next we examine the impact of suppressing higher-order references and callcc in contexts and provide an operational explanation for the game-semantic conditions known as 
visibility and bracketing respectively.
This allows us to refine the original model to provide fully abstract trace models of interaction with contexts that need not use higher-order references or $\callcc$.
Along the way, we discuss the relationship between error- and termination-based contextual testing in each case, and relate the two to trace and complete trace equivalence respectively.

Overall, the paper provides a systematic development of operational game semantics for all four cases, which represent the state-based face of the so-called semantic cube.

\end{abstract}
%
%
%

\section{Introduction}

Research into contextual equivalence has a long tradition in programming language theory,
due to its fundamental nature and applicability to numerous verification tasks, such as the correctness of compiler optimisations.
Capturing contextual equivalence mathematically, i.e. the \emph{full abstraction} problem~\cite{Mil77},
has been an important driving force in denotational semantics, which led, among others,  to the development of game semantics~\cite{AJM00,HO00}.
Game semantics models computation through sequences of question- and answer-moves by two players, traditionally called O and P, 
who play the role of the context and the program respectively.
Because of its interactive nature, it has often been referred to as a middle ground between denotational and operational semantics.

Over the last three decades the game-semantic approach has led to numerous fully abstract models
for a whole spectrum of programming paradigms.
Most papers in this strand follow a rather abstract pattern when presenting the models,
emphasing structure and compositionality, often developing a correspondence with a categorical framework along the way to facilitate proofs.
The operational intuitions behind the games are somewhat obscured in this presentation,
and left to be discovered through a deeper exploration of proofs.

In contrast, \emph{operational game semantics} aims to define models in which the interaction between
the term and the environment is described through a carefully instrumented labelled transition system (LTS), 
built using the syntax and operational semantics of the relevant language. Here, the derived trace semantics
can be shown to be fully abstract. In this line of work, the dynamics is described more directly and provides
operational intuitions about the meaning of moves, while not immediately giving structural insights about 
the structure of the traces.

In this paper, we follow the operational approach and present a whole hierarchy of trace models for 
higher-order languages with varying access to higher-order state and control.
As a vehicle for our study, we use $\HOSC$, a call-by-value higher-order language  equipped 
with general references and continuations.
We also consider its sublanguages $\GOSC$, $\HOS$ and $\GOS$, obtained respectively by restricting storage
to ground values, by removing continuations, and by imposing both restrictions.
We study contextual testing of a class of $\HOSC$ terms using
contexts from each of the languages $\xx\in\{\HOSC,\GOSC,\HOS,\GOS\}$; we write $\xx$ to refer to each case.
Our working notion of convergence will be error reachability, where an error is represented by a free variable.
Accordingly, at the technical level,  we will study a family of equivalence relations $\ctxequ{\xx}{err}$,
each corresponding to contextual testing with contexts from $\xx$, where contexts have the extra power
to abort the computation.

Our main results are trace models $\trsem{\xx}{\seq{\Gamma}{M}}$ 
for each $\xx\in\{\HOSC,\GOSC,$ $\HOS,\GOS\}$,  which  capture $\ctxequ{\xx}{err}$ through trace equivalence:
\[
\seq{\Gamma}{M_1\ctxequ{\xx}{err} M_2} \textrm{ if and only if }
\trsem{\xx}{\seq{\Gamma}{M_1}} = \trsem{\xx}{\seq{\Gamma}{M_2}}.
\]
It turns out that, for contexts with control (i.e. $\xx\in\{\HOSC,\GOSC\}$),
$\ctxequ{\xx}{err}$ coincides with the standard notion of contextual equivalence
based on termination,  written $\ctxequ{\xx}{ter}$.
However, in the other two cases, the former is strictly more discriminating than the latter.
We explain how to account for this difference in the trace-based setting,
using \emph{complete} traces.

A common theme that has emerged in game semantics is
the comparative study of the power of contexts, as it  turned out possible to identify combinatorial conditions, namely \emph{visibility}~\cite{AM97b}
and $\emph{bracketing}$~\cite{Lai97}, that correspond to contextual testing in the absence
of general references and control constructs respectively. 
In brief, visibility states that not all moves can be played,
but only those that are enabled by a ``visible part'' of the interaction, which could be thought of as
functions currently in scope.
Bracketing in turn imposes a discipline on answers, requiring that the topmost question be answered first.
In the paper, we provide an operational reconstruction of both conditions.

Overall, we propose a unifying framework for studying higher-order languages with state and control,
which we hope will make the techniques of (operational) game semantics clearer to the wider community.
The construction of  the fully abstract LTSs  is by no means automatic, as  there is no general methodology for
extracting trace semantics from game models. Some attempts in that direction have been reported  in~\cite{LS14},
but the type discipline discussed there is far too weak to be applied to the languages we study. As the most immediate
precursor to our work, we see the trace model of contextual interactions between $\HOS$ contexts and
$\HOS$ terms from~\cite{Lai07}. In comparison, the models developed in this paper are more general,
as they consider the interaction between $\HOSC$ terms and contexts drawn from any of the four languages ranged over by $\xx$.

In the 1990s, Abramsky proposed a research programme, originally called the \emph{semantic cube}~\cite{A97}, 
which concerned investigating extensions of the purely functional programming language PCF along various axes.  
From this angle, the present paper is an operational study of a \emph{semantic diamond} of languages with state,
with $\GOS$ at the bottom, extending towards $\HOSC$ at the top, either via $\GOSC$ or $\HOS$.


\section{$\HOSC$}\label{sec:hosc}

The main objects of our study will be the language  $\HOSC$ along with its fragments $\GOSC$, $\HOS$ and $\GOS$.
$\HOSC$ is a higher-order programming language equipped with general references and continuations.

\paragraph{Syntax}
 \begin{figure}[t]
\[\begin{array}{l}
\begin{array}{@{}l@{}ll}
   \sigma, \tau  & \defeq & 
   \Unit \sep \Int \sep \Bool \sep \reftype{\tau} \sep \tau \times \sigma \sep \tau \rarr \sigma  \sep \conttype{\tau} \\
   U,V & \defeq & 
   \unit \sep \trueML \sep \falseML \sep \nb{n} \sep x \sep \ell  
   \sep \pair{U}{V} \sep \lambda x^\tau.{M}  \sep \fix{y}{}{x^\tau}{M} \sep \contt{\tau}{K}  \\
    M,N & \defeq &
    V\sep  \pair{M}{N}  \sep \pi_i M  \sep M N \sep \nureft{\tau}{M} \sep !M \sep M := N \sep  \ifte{M_1}{M_2}{M_3} 
    \sep M\oplus N \sep  M \boxdot N\\
    && \sep M=N  \sep \callcct{\tau}(x.M) \sep \throwtot{\tau}{M}{N} \\
    K & \defeq & \bullet \sep  \pair{V}{K} \sep \pair{K}{M} \sep \pi_i K 
    \sep V K \sep K M \sep  \nureft{\tau} K \sep !K \sep V := K \sep K := M \sep  \ifte{K}{M}{N}\\
    && \sep K\oplus M \sep V \oplus K\sep  K\boxdot M \sep V \boxdot K \sep K = M \sep V = K \sep
    \throwtot{\tau}{V}{K} \sep  \throwtot{\tau}{K}{M}\\
    C & \defeq & \bullet \sep \pair{M}{C} \sep \pair{C}{M}\sep \pi_i C \sep  \fun{x^\tau}{C} \sep \fix{y}{}{x^\tau}{C} 
    \sep M C \sep C M \sep \nureft{\tau}{C} \sep !C\\
    && \sep C := M \sep M := C\sep  \ifte{C}{M}{N}\sep \ifte{M}{C}{N}\sep \ifte{M}{N}{C} \sep C\oplus M\sep M\oplus C\\
    && \sep C\boxdot M\sep M\boxdot C \sep C=M \sep M=C \sep \callcct{\tau}(x.C) \sep  \throwtot{\tau}{C}{M}\sep \throwtot{\tau}{M}{C}
   \end{array} \\
\end{array}\]

\begin{flushleft}
Notational conventions: \,\, $x,y \in \Var$, \,\,$\ell \in \BLoc$, \,\,$n \in \mathbb{Z}$, \,\,$i\in\{1,2\}$, \,\,$\oplus \in \{+,-,*\}$, \,\,$\boxdot \in \{=,<\}$

Syntactic sugar: $\letin{x=M}{N}$ stands for $(\lambda x.N)M$ (if $x$ does not occur in $N$ we also write $M;N$)
\end{flushleft}
\caption{$\HOSC$ syntax}
\label{fig:def-hosc}
\end{figure}
$\HOSC$ syntax is given in Figure~\ref{fig:def-hosc}. 
Assuming countably infinite sets $\BLoc$ (locations) and $\Var$ (variables),
$\HOSC$ typing judgments take  the form $\seq{\Sigma;\Gamma}{M:\tau}$, where 
$\Sigma$ and $\Gamma$ are finite partial functions that assign types to locations and variables respectively.
We list all the typing rules in the Appendix.
In typing judgements, we often write $\Sigma$ as shorthand for $\Sigma;\emptyset$ (closed)
and $\Gamma$ as shorthand for $\emptyset;\Gamma$ (location-free).
Similarly, $\seq{}{M:\tau}$ means $\seq{\emptyset;\emptyset}{M:\tau}$.

\paragraph{Operational semantics} A heap $h$ is a finite  type-respecting map from $\BLoc$ to values. 
We write $h:(\Sigma;\Gamma)$,  if  $\dom{\Sigma}\subseteq \dom{h}$ and $\seq{\Sigma;\Gamma}{h(\ell):\sigma}$ for $(\ell,\sigma)\in \Sigma$, 
The operational semantics of $\HOSC$ reduces pairs $(M,h)$, where $\seq{\Sigma;\Gamma}{M:\tau}$
and $h:(\Sigma;\Gamma)$. The rules are  given in Figure~\ref{fig:opred},
where $\{\cdot\}$ denotes (capture-avoiding) substitution.
We write $(M,h)\opredter$ if there exist $V,h'$ such that $(M,h)\red^\ast (V,h')$ and $V$ is a value.

\begin{figure}[t]
$\begin{array}{@{}c}
\begin{array}{l@{}|l}
\begin{array}{@{}l@{}l@{}l}
     (K[(\lambda x^\sigma.{M}) V],h) &  \red  &  (K[M\subst{x}{V}],h)
\\
     (K[\pi_i \pair{V_1}{V_2}],h) &  \red  & (K[V_i],h) 
\\
     (K[\ifte{\trueML}{M_1}{M_2}],h) &  \red  & (K[M_1],h)  
\\
     (K[\ifte{\falseML}{M_1}{M_2}],h) &  \red  & (K[M_2],h)      
\\  
     (K[\widehat{n} \oplus \widehat{m}],h) & \red & (K[\widehat{n \oplus m}],h)
\\     
     (K[\hat{n} \boxdot \hat{m}],h)  &  \red  &  (K[b],h) 
\\
 \multicolumn{3}{l}{\text{with } b = \trueML \text{ if } n \boxdot m, \text{otherwise } b = \falseML}
 \\
   (K[\callcc{(x^\tau.M)}],h) &\red& (K[M\subst{x}{\contt{\tau}{K}}],h)
\end{array} &
\begin{array}{@{}l@{}l@{}l}
     (K[!\ell],h)  &  \red  &  (K[h(\ell)],h)
\\
     (K[\nuref{V}],h)  &  \red  &   (K[\ell],h \cdot [\ell \mapsto V]) 
\\
     (K[\ell := V],h)  &  \red  &  (K[\unit],h[\ell \mapsto V])  
\\
     (K[\ell = \ell'],h)  &  \red  &  (K[b],h) 
\\
 \multicolumn{3}{l}{\text{with } b = \trueML \text{ if } \ell = \ell', \text{otherwise } b = \falseML}
 \\
 \multicolumn{3}{l}{(K[(\underbrace{\fix{y}{}{x^\sigma}{M}}_{U}) V],h) \red (K[M\{V/x,U/y\}],h)}
 \\
\multicolumn{3}{l}{ (K[\throwtot{\tau}{V}{\contt{\tau}{K'}}],h) \red (K'[V],h)}
\end{array}\\
\end{array}\\
\end{array}$
\caption{Operational reduction for $\HOSC$}
\label{fig:opred}
\end{figure}
We distinguish the following fragments of $\HOSC$.
\begin{definition}
\begin{itemize}
\item $\GOSC$ types are $\HOSC$ types except that reference types
are restricted to $\reftype{\iota}$, where $\iota$ is given by the grammar
$\iota \defeq  \Unit \sep \Int \sep \Bool \sep \reftype{\iota}$.
$\GOSC$ terms are  $\HOSC$ terms whose typing derivations (i.e. not only the final typing judgments) rely on $\GOSC$ types only.
$\GOSC$ is a superset of $\FOSC$~\cite{DNB12}, which also includes references to references (the $\reftype{\iota}$ case above).
\item $\HOS$ types are $\HOSC$ types that do not feature the $\mathrm{cont}$ constructor. 
$\HOS$ terms are $\HOSC$ terms whose typing derivations rely on $\HOS$ types only.
Consequently, $\HOS$ terms never have subterms of the form
 $\callcct{\tau}(x.M)$, $\throwtot{\tau}{M}{N}$ or $\contt{\tau}{K}$.
 \item $\GOS$  is the intersection of $\HOS$ and $\GOSC$, both for types and terms, i.e.
there are no continuations and storage is restricted to values of type $\iota$, defined above.
\end{itemize}
\end{definition}
\begin{definition}
Given a $\HOSC$ term $\seq{\Gamma}{M:\tau}$, we refer to  types in $\Gamma$ and $\tau$ as \boldemph{boundary types}.
Let $\xx\in\{\HOSC,\GOSC,\HOS,$ $\GOS\}$.  We say that a $\HOSC$ term $\seq{\Gamma}{M:\tau}$ has an $\xx$ boundary
if all of its boundary types are from $\xx$.
\end{definition}
\begin{remark}
Note that typing derivations of $\HOSC$ terms with an $\xx$ boundary may 
contain arbitrary $\HOSC$ types
as long as the final typing judgment uses types from $\xx$ only.
Consequently, if $\xx\neq \HOSC$, $\HOSC$ terms with an $\xx$ boundary form a strict superset of $\xx$.
\end{remark}
Next we introduce several notions of contextual testing for $\HOSC$-terms, using various kinds of contexts.
For a start, we introduce the classic notion of contextual approximation based on observing termination.
The notions are parameterized by $\xx$, indicating which language is used to build the testing contexts.
We write $\seq{\Gamma}{C:\tau\rarr\tau'}$ if $\seq{\Gamma,x:\tau}{C[x]:\tau'}$, and $\cseq{\Gamma}{C}{\tau}$
if $\seq{\Gamma}{C:\tau\rarr\tau'}$ for some $\tau'$.
\begin{definition}[Contextual Approximation]
Let $\xx\in\{\HOSC,\GOSC,\HOS$, $\GOS\}$.
Given $\HOSC$ terms $\Gamma \vdash M_1,M_2:\tau$ with an $\xx$ boundary,
we define $\Gamma\vdash M_1\ctxter{\xx} M_2$ to hold,
when for all contexts $\cseq{}{C}{\tau}$ built from the syntax of $\xx$,
if $(C[M_1],\emptyheap) \opredter$ then $(C[M_2],\emptyheap) \opredter$.
\end{definition}
We also consider another way of testing, based on observing whether a program can reach a breakpoint (error point) inside a context.
Technically, the breakpoints are represented as occurrences of a special free error variable $\err:\Unit\rarr\Unit$.
Reaching a breakpoint then corresponds to convergence to a stuck configuration of the form $(K[\err()],h)$:
we write $(M,h)\opredtererr$ if there exist $K,h'$ such that $(M,h)\rarr^\ast (K[\err()],h')$. 
\begin{definition}[Contextual Approximation through Error]
Let $\xx \in \{ \HOSC,$ $\FOSC,$ $\HOS,$ $\GOS\}$.
Given $\HOSC$ terms $\Gamma \vdash M_1,M_2:\tau$  with an $\xx$ boundary and $\err\not\in\dom{\Gamma}$,
we define $\Gamma\vdash M_1\ctxerr{\xx} M_2$ to hold,
when for all contexts $\cseq{\err:\Unit\rarr\Unit}{C}{\tau}$ {built from $\xx$-syntax},
if $(C[M_1],\emptyheap) \opredtererr$ then $(C[M_2],\emptyheap) \opredtererr$.
\end{definition}
For the languages in question, it will turn out that 
$\ctxerr{\xx}$ is at least as discriminating as $\ctxter{\xx}$ for each  $\xx\in\{\HOSC,\GOSC,\HOS,\GOS\}$,
and that they coincide for $\xx\in\{\HOSC,\GOSC\}$.
We will write $\ctxequ{\xx}{\err}$ and $\ctxequ{\xx}{\ter}$ for the associated equivalence relations.

For higher-order languages with state and control, it is well known that
contextual testing can be restricted to evaluation contexts
after instantiating the free variables of terms to closed values (the so-called \emph{closed instances of use}, CIU).
Let us write $\seq{\Sigma,\Gamma'}{\gamma:\Gamma}$ 
for substitutions $\gamma$ such that, for any $(x,\sigma_x)\in\Gamma$,
the term $\gamma(x)$ is a value satisfying $\seq{\Sigma;\Gamma'}{\gamma(x):\sigma_x}$.
Then $M\substF{\gamma}$ stands for the outcome of applying $\gamma$ to $M$.
\begin{definition}[CIU Approximation]
Let $\xx\in\{\HOSC,\GOSC,\HOS,\GOS\}$ and
let $\seq{\Gamma}{M_1,M_2:\tau}$ be $\HOSC$ terms with an $\xx$ boundary.
\begin{itemize}
\item $\Gamma \vdash M_1 \ciupre{\xx}{\ter} M_2:\tau$, when for all $\Sigma,h, K,\gamma$, all built from $\xx$ syntax,
 such that  $h:\Sigma$, $\cseq{\Sigma}{K}{\tau}$,
 and  $\seq{\Sigma}{\gamma:\Gamma}$,
 we have $(K[M_1\substF{\gamma}],h) \opredter$ implies $(K[M_2\substF{\gamma}],h) \opredter$.
\item We write $\Gamma \vdash M_1 \ciupre{\xx}{\err} M_2:\tau$,
when for all $\Sigma,h, K,\gamma$, all built from $\xx$ syntax, 
such that  $h:\Sigma;\hat{\err}$,\, $\cseq{\Sigma;{\hat{\err}}}{K}{\tau}$,\,
 and  $\seq{\Sigma;\hat{\err}}{\gamma:\Gamma}$,
 we have $(K[M_1\substF{\gamma}],h) \opredtererr$ implies $(K[M_2\substF{\gamma}],h) \opredtererr$,
 where $\err\not\in\dom{\Gamma}$ and $\hat{\err}$ stands for $\err:\Unit\rarr\Unit$.
 \end{itemize}
\end{definition}
Results stating that ``CIU tests suffice'' are referred to as CIU lemmas.
A general framework for obtaining such results for higher-order languages  with effects was developed in~\cite{HMST95,Tal98}.
The results stated therein are for termination-based testing, i.e. $\opredter$, but adapting them to $\opredtererr$ is not problematic.
\begin{lemma}[CIU Lemma]\label{lem:ciu}
Let $\xx\in\{\HOSC,\GOSC,\HOS,\GOS\}$ and $\yy\in \{\ter,\err\}$. Then we have
$\seq{\Gamma}{M_1\ctxpre{\xx}{\yy} M_2}$ iff $\seq{\Gamma}{M_1 \ciupre{\xx}{\yy} M_2}$.
\end{lemma}
The preorders $\ctxpre{\xx}{\err}$ will be the central object of study in the paper.
Among others, we shall provide their alternative characterizations 
using trace semantics.
The characterizations will apply to a class of terms that we call \emph{cr-free}.
\begin{definition}
A $\HOSC$ term $\seq{\Gamma}{M:\tau}$ is \boldemph{cr-free} if it does not contain occurrences of
$\contt{\sigma}{K}$ and locations, and its boundary types  are $\mathrm{cont}$- and $\mathrm{ref}$-free.
\end{definition}
We stress that the boundary restriction applies to $\Gamma$ and $\tau$ only, and
subterms of $M$ may well contain arbitrary $\HOSC$ types and occurrences  of 
$\mathrm{ref}_{\sigma}$, $\callcct{\sigma}$, $\mathrm{throw}_\sigma$ for any $\sigma$.
The majority of $\HOSC/\GOSC/\HOS/\GOS$ examples studied in the literature, e.g.~\cite{PS98,ADR09,DNB12}, are actually cr-free.
We will revisit some of them as Examples~\ref{ex:callback}, \ref{ex:counter}, \ref{ex:bracket}.
The fact that cr-free terms may not contain subterms $\contt{\tau}{K}$ or $\ell$ is not really a restriction, as
$\contt{\tau}{K}$ and $\ell$ being more of a run-time construct than a feature meant to be used directly by programmers.
Finally, we note that the boundary of a cr-free term is an $\xx$ boundary for any $\xx\in\{\HOSC,\GOSC,\HOS,\GOS\}$.
Thus, we can consider approximation between cr-terms for 
any $\xx$ from the range, i.e. the notions $\ctxpre{\xx}{\err}$, $\ctxpre{\xx}{\ter}$ are all applicable.
Consequently, cr-free terms provide a common setting in which the discriminating power of $\HOSC,\GOSC,\HOS$ and $\GOS$ contexts
can be compared. We discuss the scope for extending our results outside of the cr-free fragment,
and for richer type systems, in Section~\ref{sec:extensions}.


\section{HOSC[HOSC]}\label{sec:hoschosc}

Recall that $\ctxpre{\HOSC}{err}$ concerns testing $\HOSC$ terms with $\HOSC$ contexts.
Accordingly, we call this  case $\HOSC[\HOSC]$. 
For $\cont_{\sigma}(K)$-free terms, we show that
$\ctxerr{\HOSC}$ and $\ctxter{\HOSC}$ coincide, which follows from the lemma below.
\begin{lemma}\label{lem:ctxerr}
Suppose $\seq{\Gamma}{M_1, M_2}$ be $\HOSC$ terms not containing any occurrences of $\cont_{\tau}(K)$.
\begin{enumerate}
\item $\seq{\Gamma}{M_1\ctxerr{\xx} M_2}$ implies $\seq{\Gamma}{M_1\ctxter{\xx} M_2}$, for $\xx\in\{\HOSC,\GOSC, \HOS,\GOS\}$.
\item $\seq{\Gamma}{M_1\ctxter{\xx} M_2}$ implies $\seq{\Gamma}{M_1\ctxerr{\xx} M_2}$, for $\xx\in\{\HOSC,\GOSC\}$.
\end{enumerate}
\end{lemma}
In what follows, after introducing several preliminary notions,
we shall design a labelled transition system (LTS) whose traces will turn out 
to capture contextual interactions involved  in testing cr-free terms according to $\ctxpre{\HOSC}{err}$.
This will enable us to capture $\ctxpre{\HOSC}{err}$ via trace inclusion.
Actions of the LTS will  refer to functions and continuations in a symbolic way, using typed names.

\subsection{Names and abstract values}
\begin{definition}
Let $\FNames=\biguplus_{\sigma,\sigma'}\FNames_{\sigma\rarr\sigma'}$ 
be the set of \boldemph{function names}, partitioned into mutually disjoint countably infinite
sets  $\FNames_{\sigma\rarr\sigma'}$.
We will use $f,g$ to range over $\FNames$ , and write $f:\sigma\rarr\sigma'$ for $f\in\FNames_{\sigma\rarr\sigma'}$. 

Analogously, let $\CNames =\biguplus_{\sigma} \CNames_{\sigma}$  be the set of \boldemph{continuation names}.
We will use $c,d$ to range over $\CNames$, and write $c:\sigma$ for $c\in \CNames_\sigma$.
Note that the constants represent continuations, so the ``real'' type of $c$ is $\conttype{\sigma}$, but we write
$c:\sigma$ for the sake of brevity.
We assume that $\CNames,\FNames$ are disjoint and let  $\Names=\FNames \uplus \CNames$.
Elements of $\Names$ will be weaved into various constructions in the paper, e.g. terms, heaps, etc.
We will then write $\nu(X)$ to refer to the set of names used in some entity $X$.
 \end{definition}
 Because of the shape of boundary types in cr-free terms and, in particular, the presence of product types, 
 the  values that will be exchanged between the context and the program take the form of 
 tuples consisting of $()$, integers, booleans and functions.
 To describe such scenarios, we introduce the notion of \boldemph{abstract values}, which are patterns 
 that match such values. Abstract values are generated by the grammar
 \[
 A,B \defeq \unit \sep \trueML \sep \falseML \sep \nb{n} \sep f \sep \pair{A}{B} 
 \]
 with the proviso that,  in any abstract value,  a name may occur at most once.
 As function names are intrinsically typed, we can assign types to abstract values in the obvious way, writing $A:\tau$.

\subsection{Actions and traces}

Our LTS will be based on  four kinds of actions, listed below.
Each action will be equipped with a \boldemph{polarity}, which is either Player (P) or Opponent (O).
P-actions describing interaction steps made by a tested term, while O-actions involve the context.
\begin{itemize}
 \item \boldemph{Player Answer} (PA) $\ansP{c}{A}$, where $c: {\sigma}$ and $A :{\sigma}$.
 This action corresponds to the term sending  an abstract value $A$ through a continuation name $c$. 
 \item \boldemph{Player Question} (PQ) $\questP{f}{A}{c}$, where 
 $f: {\sigma \rightarrow \sigma'}$, $A :{\sigma}$ and $c:\sigma'$.
Here, an abstract value $A$  and a continuation name $c$ are sent by the term through a function name $f$. 
 \item \boldemph{Opponent Answer} (OA) $\ansO{c}{A}$, $c: {\sigma}$ then $A:{\sigma}$.
 In this case, an abstract value $A$ is received from the environment via 
 the continuation name $c$.
 
 \item \boldemph{Opponent Question} (OQ) $\questO{f}{A}{c}$, where
  $f : {\sigma \rightarrow \sigma'}$, $A:{\sigma}$ and $c: {\sigma'}$.
 Finally, this action corresponds to receiving an abstract value $A$ and a continuation name $c$ from the environment   through a function name $f$.
\end{itemize}
In what follows, $\act$ is used to range over actions.
 We will say that a name is \boldemph{introduced} by an action $\act$ if it is sent or received in $\act$.
 If $\act$ is an O-action (resp. P-action), we say that the name was introduced by O (resp. P). 
 An action $\act$ is \boldemph{justified} by another action $\act'$ if the name that $\act$ uses to 
 communicate, i.e. $f$ in questions ($\questP{f}{A}{c}$, $\questO{f}{A}{c}$) 
 and $c$ in answers ($\ansP{c}{A}$, $\ansO{c}{A}$),  has been introduced by $\act'$.
 
We will work with sequences of actions of a very special shape, specified below. 
The definition assumes two given sets of names, $N_P$ and $N_O$,
which represent names that have already been introduced by P and O respectively.
 \begin{definition}
Let $N_O, N_P\subseteq \Names$.  An  $(N_O,N_P)$-\boldemph{trace} is a sequence $t$ of actions such that:
 \begin{itemize}
  \item the actions alternate between Player and Opponent actions;
  \item no name is introduced twice; 
  \item names from $N_O, N_P$ need no introduction;
    \item if an action  $\act$ uses a name to communicate then 
    \begin{itemize}
    \item $\act=\questP{f}{A}{c}$ ($f\in N_O$) or $\act=\ansP{c}{A}$ ($c\in N_O$) or 
    $\act=\questO{f}{A}{c}$ ($f\in N_P$) or $\act=\ansO{c}{A}$ ($c\in N_P$) or
    \item the name has been introduced by an earlier action $\act'$ of opposite polarity.
 \end{itemize}
 \end{itemize}
\end{definition}
Note that, due to the shape of actions, a continuation name can only be introduced/justified by a question.
Moreover, because names are never introduced twice, if $\act'$ justifies $\act$ then $\act'$ is uniquely determined in a given trace.
Readers familiar with game semantics will recognize that traces are very similar
to alternating justified sequences except that traces need not be started by O.
\begin{example}\label{ex:trace}
Let $(N_O,N_P)=(\{c\},\emptyset)$ where $c:\tau = ((\Unit\rarr\Unit)\rarr\Unit)\times(\Unit\rarr\Int)$. Then the following sequence is an
$(N_O,N_P)$-trace:
\[
\tr_1 =\ansP{c}{\pair{g_1}{g_2}}\,\,\,
\questO{g_1}{f_1}{c_1}\,\,\,
\questP{f_1}{()}{c_2}\,\,\,
\ansO{c_2}{()}\,\,\,
\ansP{c_1}{()}\,\,\,
\ansO{c_2}{()}\,\,\,
\ansP{c_1}{()}\,\,\,
\questO{g_2}{()}{c_3}\,\,\,
\ansP{c_3}{2}
\]
where $g_1:  (\Unit\rarr\Unit)\rarr\Unit$, $g_2:\Unit\rarr\Int$, $f_1:\Unit\rarr\Unit$, $c_1,c_2:\Unit$, $c_3:\Int$.
\end{example}


\subsection{Extended syntax and reduction}

We extend the definition of $\HOSC$ presented in Figure~\ref{fig:opred}
to take into account these names.
We refine the operational reduction using continuation names to keep track
of the toplevel continuation.
We list all the changes below.
\begin{itemize}
 \item Function names are added to the syntax as \emph{constants}. 
 Since they are meant to represent values, they are also considered to be syntactic values in the extended language.
{ \[
\frac{f\in\FNames_{\sigma\rarr\sigma'}}{\seq{\Sigma;\Gamma}{f:\sigma\rarr\sigma'}}
\]}
  \item Continuation names are \emph{not}  terms on their own. Instead,
  they are built into  the syntax via a new construct $\contt{\sigma}{(K,c)}$, 
 { subject to the following typing rule.
 \[
 \frac{\seq{\Sigma;\Gamma}{K:\sigma\rarr\sigma'}\quad c\in \CNames_{\sigma'}}{
 \seq{\Sigma;\Gamma}{\contt{\sigma}{(K,c)}:\conttype{\sigma}}}
 \]}
$\contt{\sigma}{(K,c)}$ is a staged continuation that first evaluates terms inside $K$
and, if this produces a value, the value is passed to $c$. 
This operational meaning will be implemented  through a suitable reduction rule,
to be discussed next. $\contt{\sigma}{(K,c)}$ is also regarded as a value.
Note that we remove the old construct $\contt{\sigma}{K}$ from the extended syntax.
\item The operational semantics $\ered$ underpinning the LTS is based on triples 
$(M,c,h)$ such that $\Sigma;\Gamma \vdash M:\sigma$, $c\in\CNames_\sigma$ and $h:\Sigma$.
The continuation name $c$ is used to represent the surrounding context, which is left abstract.
The previous operational rules $\red$ are embedded into the new reduction $\ered$ using the rule below.
 \[\begin{array}{l}
 \inferrule*{(M,h) \red (M',h')}{(M,c,h) \ered (M',c,h')}
  \end{array}\]
The two reduction rules related to continuations, previously used to define $\red$,
are \emph{not} included. Instead we use the following rules, which take advantage of the extended syntax.
\[\begin{array}{rcl}
 (K[\callcct{\tau}{(x.M)}],c,h) &\ered& (K[M\subst{x}{\contt{\tau}{(K,c)}}],c, h) \\
  (K[\throwtot{\tau}{V}{\contt{\tau}{(K',c')}}],c,h) &\ered& (K'[V],c',h)
\end{array}\]
\end{itemize}


\subsection{Configurations}

We write $\Values$ for the extended set of syntactic values, i.e. $\FNames\subseteq \Values$.
Let $\EContexts$ stand for the set of extended evaluation contexts, defined as $K$ in Figure~\ref{fig:def-hosc} taking the extended definition of values into account.
Before defining the transition relation of our LTS, we discuss the shape of configurations, providing intuitions behind each component.

\boldemph{Passive configurations} take the form $\conf{\gamma,\xi,\phi,h}$ and are meant to represent stages at which the environment is to make a move.
\begin{itemize}
\item $\gamma:(\FNames \pmap \Values) \uplus (\CNames \pmap \EContexts)$ is a finite map.  It will play the role of an environment that
relates function names communicated to the environment (i.e. those introduced by P) to syntactic values,
and continuation names introduced by P to evaluation contexts.
\item $\xi:(\CNames \pmap \CNames)$ is a finite map. It complements the role of $\gamma$ for continuation names and indicates the continuation to which the outcome
of applying $\gamma(c)$ should be passed.
\item $\phi\subseteq\Names$. The set $\phi$ will be used to collect all the names used in the interaction, regardless of which participant introduced them. 
Following our description above, those introduced by O will correspond to $\phi\setminus\dom{\gamma}$.
\end{itemize}
The components satisfy healthiness conditions, implied by their role in the system. Let $\Sigma=\dom{h}$.
\begin{itemize}
\item If $f:\dom{\gamma}\cap \FNames_{\sigma\rarr\sigma'}$ then $\gamma(f)$ is a value such that $\seq{\Sigma}{\gamma(f):\sigma\rarr\sigma'}$.
\item $\dom{\xi}=\dom{\gamma}\cap \CNames$.
\item If $c:\dom{\gamma}\cap \CNames_{\sigma}$ and $\seq{\Sigma}{\gamma(c):\sigma\rarr\sigma'}$ then
$\xi(c)\in\CNames_{\sigma'}$.
\item Finally, names introduced by the environment and communicated to the program may end up in the environments and the heap:
$\nu(\img{\gamma}), \nu(\img{\xi})$, $\nu(\img{h})\subseteq \phi\setminus \dom{\gamma}$.
\end{itemize}
\boldemph{Active configurations} take the form $\conf{M,c,\gamma,\xi,\phi,h}$ and represent interaction steps of the term.
The $\gamma,\xi,\phi,h$ components have already been described above. For $M$ and $c$, given $\Sigma=\dom{h}$, we will have
$\seq{\Sigma;\emptyset}{M:\sigma}$, $c\in \CNames_\sigma$ and $\nu(M)\cup\{c\}\subseteq \phi\setminus\dom{\gamma}$.

\subsection{Transitions}

\cutout{\begin{itemize}
\item $\sem{\tau}$ is the set of 
the pairs $(A,\phi)$, where $A$ is an abstract value of type $\tau$ and $\phi$
consists of all function names occurring in $A$.
  \[\begin{array}{l}
   \sem{\Unit} \defeq  \{(\unit,\varnothing)\} \qquad
   \sem{\Bool} \defeq \{(\trueML,\varnothing),(\falseML,\varnothing)\} \qquad
   \sem{\Int}  \defeq  \{(\nb{n},\varnothing) \sep n \in \mathbb{Z}\} \\
   \sem{\tau \rarr \sigma}  \defeq  \{(f,\{f\}) \sep f \in \FNames_{\tau \rarr \sigma}\} \\
   \sem{\tau \times \sigma}  \defeq  \{(\pair{A_1}{A_2},\phi_1 \cup \phi_1) \sep 
   (A_1,\phi_1) \in \sem{\tau},\, (A_2,\phi_2) \in \sem{\sigma},\, \phi_1 \cap \phi_2 = \varnothing\} \\
  \end{array}
 \]
 
\item  $\sem{\Gamma}$ contains  triples of the form $(A,\gamma,\phi)$, 
 where $A$ is an abstract value,
 $\gamma$ is a substitution of abstract values for each variable in $\dom{\Gamma}$,
 and $\phi$ consists of function names in $A$. 
 \[
  \begin{array}{lll}
 \sem{\emptymap} & \defeq & \{((),\emptymap,\varnothing)\} \\
 \sem{(x:\tau),\Gamma} & \defeq & \{\pair{A}{\vec{A}},\gamma \cdot [x \mapsto A],\phi \cup \phi' 
  \sep (A,\phi) \in \sem{\tau}, (\vec{A},\gamma,\phi') \in \sem{\Gamma}, \phi \cap \phi' = \varnothing\}
  \end{array}
\]
\end{itemize}}
Observe that any closed value $V$ of a $\mathrm{cont}$- and $\mathrm{ref}$-free  type $\sigma$ can be decomposed into an abstract value $A$ (pattern)
 and the corresponding substitution $\gamma$ (matching). The set of all such decompositions,
 written  $\AVal{V}{\sigma}$, is defined below.
 Given a value $V$ of  a (cr-free) type $\sigma$, $\AVal{V}{\sigma}$ contains all pairs $(A,\gamma)$ such that
  $A$ is an abstract value and $\gamma:\nu(A)\rarr\Values$ is a substitution such that $A\{\gamma\} = V$.
More concretely,
\[
\begin{array}{lll}
 \AVal{V}{\sigma} & \defeq & \{(V,\emptyset)\} \quad 
  \text{ for } \sigma\in\{\Unit,\Bool,\Int\} \\
 \AVal{V}{\sigma \rarr \sigma'} & \defeq & 
  \{(f,[f \mapsto V]) \sep f \in \FNames_{\sigma \rarr \sigma'}\} \\
 \AVal{\pair{U}{V}}{\sigma \times \sigma'} & \defeq & 
 \{(\pair{A_1}{A_2},\gamma_1 \cdot \gamma_2) \sep 
 (A,\gamma_1) \in \AVal{U}{\sigma},\,(A_2,\gamma_2) \in \AVal{V}{\sigma'}\}
\end{array}
\]
Note that, by writing $\cdot$, we mean to implicitly require that the function domains be disjoint.
Similarly, when writing $\uplus$, we stipulate that the argument sets be disjoint.
\begin{example}
Let $\sigma = {(\Int\rarr\Bool)\times(\Int\times(\Unit\rarr\Int))}$ and
$V\equiv\pair{\lambda x^\Int.x \neq 1}{\pair{2}{{\lambda x^\Unit.3}}}$. Then
$\AVal{V}{\sigma}$ equals
\[
\{  (\pair{f}{\pair{2}{g}}, [ f\mapsto (\lambda x^\Int.x \neq 1)]\cdot [g\mapsto (\lambda x^\Unit.3)])\,\,\,|\,\,\,  
f\in\FNames_{\Int\rarr\Unit},  \, g\in \FNames_{\Unit\rarr\Int}\}.
\]
\end{example}
Finally, we  present the transitions of, what we call the $\HOSC[\HOSC]$ LTS, in  Figure~\ref{fig:lts-hosc}.
\begin{example}
We analyze the (PQ) rule below in more detail. 
\[\begin{array}{lllll}
   (PQ) & \conf{K[fV],c,\gamma,\xi,\phi,h} & \ired{\questP{f}{A}{c'}} & 
    \conf{\gamma \cdot\gamma'\cdot[c' \mapsto K],\xi\cdot [c' \mapsto c],
    \phi\uplus\nu(A)\uplus \{c'\},h} \\
    & \multicolumn{3}{l}{\text{ when } f:\sigma \rarr \sigma',\, (A,\gamma') \in \AVal{V}{\sigma} \text{ and } c':\sigma'}
    \end{array}\]
    The use of $\uplus$ in $\phi\uplus\nu(A)\uplus \{c'\}$ is meant to
highlight the requirement that the names introduced in $\questP{f}{A}{c'}$, i.e. $\nu(A)\cup\{c'\}$, should be fresh and disjoint from $\phi$.
Moreover, note how  $\gamma$ and $\xi$ are updated.
In general, $\gamma,\xi,h$ are updated during P-actions.
\end{example}

\begin{figure}[t]
  \[ \begin{array}{l|lll}
   (P\tau) & \conf{M,c,\gamma,\xi,\phi,h} & \ired{\ \tau \ } & 
     \conf{N,c',\gamma,\xi,\phi,h'} \\
   & \multicolumn{3}{l}{\text{ when } (M,c,h) \ered (N,c',h')}\\
   (PA) & \conf{V,c,\gamma,\xi,\phi,h} & \ired{\ansP{c}{A}} & 
     \conf{\gamma \cdot \gamma',\xi,\phi\uplus\nu(A),h} \\
   & \multicolumn{3}{l}{\text{ when } c:\sigma,\,  (A,\gamma') \in \AVal{V}{\sigma}}\\   
   (PQ) & \conf{K[fV],c,\gamma,\xi,\phi,h} & \ired{\questP{f}{A}{c'}} & 
    \conf{\gamma \cdot\gamma'\cdot[c' \mapsto K],\xi\cdot [c' \mapsto c],
    \phi\uplus\nu(A)\uplus \{c'\},h} \\
    & \multicolumn{3}{l}{\text{ when } f:\sigma \rarr \sigma',\, (A,\gamma') \in \AVal{V}{\sigma},\, c':\sigma'}\\
   (OA) & \conf{\gamma,\xi,\phi,h} & \ired{\ansO{c}{A}} & 
    \conf{K[A],c',\gamma,\xi,\phi \uplus \nu(A),h} \\
    & \multicolumn{3}{l}{\text{ when }  c:\sigma,\,A:\sigma,\,\gamma(c) = K,\,   \xi(c)=c'}\\
   (OQ) & \conf{\gamma,\xi,\phi,h} & \ired{\questO{f}{A}{c}} & 
     \conf{V A,c,\gamma,\xi,\phi \uplus \nu(A)  \uplus \{c\},h} \\
     & \multicolumn{3}{l}{\text{ when } f:\sigma \rarr \sigma',\,  A:\sigma,\, c:\sigma',\,\gamma(f) = V }\\[3mm]
%
%
\multicolumn{4}{l}{\text{NB $c:\sigma$ stands for $c\in\CNames_\sigma$.}}
  \end{array}\]
  \caption{HOSC[HOSC] LTS}
  \label{fig:lts-hosc}
  \end{figure}
\begin{definition}
Given two configurations  $\CC,\CC'$,
 we write $\CC \iRed{\act} \CC'$  if $\CC {\ired{\tau}}^\ast \CC'' \ired{\act} \CC'$,
 with $\ired{\tau}^\ast$ representing multiple (possibly none) $\tau$-actions.
 This notation is extended to sequences of actions: given 
 $\tr = \act_1  \ldots \act_n$, we write $\CC \iRed{\tr} \CC'$,
 if there exist $\CC_1,\ldots,\CC_{n-1}$ such that
 $\CC \iRed{\act_1} \CC_1 \cdots \CC_{n-1}\iRed{\act_n} \CC'$. 
 We define $\TrR{\HOSC}{\CC}=\{ \tr \,|\, \textrm{there exists $\CC'$ such that } \CC\iRed{\tr}\CC'\}$.
 \end{definition}

 \begin{lemma}
 Suppose $\CC=\conf{\gamma,\xi,\phi,h}$ or $\CC=\conf{M,c,\gamma,\xi,\phi,h}$ are configurations.
 Then  elements of $\TrR{\HOSC}{\CC}$ are $(\phi\setminus\dom{\gamma},\dom{\gamma})$-traces.
  \end{lemma}

\begin{figure}
  {    \[
  \begin{array}{l|l}
  \begin{array}{rl}
   \ex{M_1}{cwl}: &  \mathrm{let \ x \ = ref \ 0 \ in} \\
 & \mathrm{let \ b \ = ref \ \falseML \ in} \\   
  &\mathrm{\langle \lambda f. } \
  \mathrm{if \ \lnot(!b) \ then} \\
  &\quad \mathrm{b := \trueML;} \, \mhilight{\mathrm{f();x:=!x+1}}; \\
  &\quad \mathrm{b := \falseML;}\\
  &\quad\mathrm{else \ () ,\,\lambda \_:\Unit.!x\rangle} \\
  \end{array} &
  \begin{array}{rl}
\ex{M_2}{cwl}:  &  \mathrm{let \ x \ = ref \ 0 \ in} \\ 
&  \mathrm{let \ b \ = ref \ \falseML \ in} \\  
&  \mathrm{\langle \lambda f. }  
 \ \mathrm{if \ \lnot(!b) \ then} \\
&  \quad \mathrm{b := \trueML;} \,
  \mhilight{\mathrm{let \ n = !x \ in \ f(); x:=n+1}};\\
&  \quad \mathrm{b := \falseML;} \\  
 & \quad \mathrm{else \ ()\,,\lambda \_:\Unit.!x\rangle} \\
  \end{array}
  \end{array}
  \]}\hspace{-1em}
  \caption{Callback-with-lock Example~\cite{ADR09}}
 \label{fig:ex-cwl} 
\end{figure}

  \begin{example}\label{ex:deriv}
  In Figure~\ref{fig:trace},
  we show that the trace from Example~\ref{ex:trace} is generated by 
  the configuration  $\CC\defeq\conf{\ex{M_1}{cwl},c,\emptyset,\emptyset,\{c\},\emptyset}$, where
  $\ex{M_1}{cwl}$ is given in Figure~\ref{fig:ex-cwl}.
We write 
  $\inccwl \defeq \lambda f. \ifte{\lnot(!\ell_b)}{(\ell_b:=\trueML; f(); \ell_x:=!\ell_x+1;\ell_b :=\falseML)}{()}$, 
  $\getcwl \defeq\lambda \_. !\ell_x$
  and $c:((\Unit\rarr\Unit)\rarr\Unit)\times(\Unit\rarr\Int)$.
  \gjadd{It is interesting to notice that in this interaction, Opponent uses  the continuation $N$ twice,
  incrementing the counter $x$ by two. The second time, it does it without having to
  call $\inccwl$ again, but rather by using the continuation name $c_2$.}
\begin{figure}
  \[\renewcommand\arraystretch{0.7}
  \begin{array}{rll}
  \CC = & \conf{\ex{M_1}{cwl},c,\emptyset,\emptyset,\{c\},\emptyset}\\
  \ired{\tau^\ast} & \conf{\pair{\inccwl}{\getcwl}, c, \emptyset, \emptyset, \{c\}, [\ell_b \mapsto \falseML, \ell_x \mapsto 0]}\\
  \ired{\ansP{c}{\pair{g_1}{g_2}}} & \conf{\gamma_1, \emptyset, \{c,g_1,g_2\}, [\ell_b \mapsto \falseML, \ell_x \mapsto 0]}
&  \text{ with } \gamma_1 = [g_1\mapsto \inccwl, g_2\mapsto \getcwl],\\
\ired{\questO{g_1}{f_1}{c_1}} &\conf{\inccwl f_1, c_1, \gamma_1, \emptyset, \phi_2,  [\ell_b \mapsto \falseML, \ell_x \mapsto 0]}
& \text{ with } \phi_2 = \{c,g_1,g_2,f_1,c_1\}\\
\ired{\tau^\ast} & \conf{f_1();N, c_1, \gamma_1, \emptyset, \phi_2,  [\ell_b \mapsto \trueML, \ell_x \mapsto 0]}
& \text{ with } N = \ell_x:=!\ell_x+1;\ell_b:=\falseML\\
\ired{\questP{f_1}{()}{c_2}} & \conf{\gamma_2, \xi, \phi_3,  [\ell_b \mapsto \trueML, \ell_x \mapsto 0]}
&  \text{ with } \gamma_2 = \gamma_1\cdot [c_2\mapsto \bullet;N],  \\
\ired{\ansO{c_2}{()}} &  \conf{(); N, c_1, \gamma_2, \xi, \phi_3,  [\ell_b \mapsto \trueML, \ell_x \mapsto 0]}
& \quad \xi = [c_2\mapsto c_1]  \text{ and } \phi_3 = \phi_2 \uplus \{c_2\}\\
\ired{\tau^\ast} &  \conf{(), c_1, \gamma_2, \xi, \phi_3,  [\ell_b \mapsto \falseML, \ell_x \mapsto 1]}\\
\ired{\ansP{c_1}{()}}& \conf{\gamma_2, \xi, \phi_3,  [\ell_b \mapsto \falseML, \ell_x \mapsto 1]}\\
\ired{\ansO{c_2}{()}} &  \conf{(); N, c_1, \gamma_2, \xi, \phi_3,  [\ell_b \mapsto \falseML, \ell_x \mapsto 1]}\\
\ired{\tau^\ast} &  \conf{(), c_1, \gamma_2, \xi, \phi_3,  [\ell_b \mapsto \falseML, \ell_x \mapsto 2]}\\
\ired{\ansP{c_1}{()}}& \conf{\gamma_2, \xi, \phi_3,  [\ell_b \mapsto \falseML, \ell_x \mapsto 2]}\\
\ired{\questO{g_2}{()}{c_3}}&\conf{\getcwl (), c_3, \gamma_2, \xi,\phi_4,  [\ell_b \mapsto \falseML, \ell_x \mapsto 2]}
& \text{ with }\phi_4= \phi_3\uplus\{c_3\}\\
\ired{\tau^\ast}&\conf{2, c_3, \gamma_2, \xi, \phi_4,  [\ell_b \mapsto \falseML, \ell_x \mapsto 2]}\\
\ired{\ansP{c_3}{2}} &\conf{\gamma_2, \xi, \phi_4,  [\ell_b \mapsto \falseML, \ell_x \mapsto 2]}\\
  \end{array}
  \]
  \caption{Trace derivation in the HOSC[HOSC] LTS}\label{fig:trace}
  \end{figure}
 \end{example}
 \begin{remark}\label{rem:invar}
 Due to the freedom of name choice, note that
 $\TrR{\HOSC}{\CC}$ is closed under type-preserving renamings that preserve names from $\CC$.
 \end{remark}

 \subsection{Correctness and full abstraction}
 
 We define two kinds of special configurations that will play an important role in spelling out correctness results  for the HOSC[HOSC] LTS.
Let $\Gamma=\{x_1:\sigma_1,\cdots, x_k:\sigma_k\}$.
A map $\rho$ from $\{x_1,\cdots, x_k\}$ to the set of abstract values will be called a \boldemph{$\Gamma$-assignment}
provided, for all $1\le i\neq j\le k$, we have $\rho(x_i):\sigma_i$ and  $\nu(\rho(x_i))\cap \nu(\rho(x_j))=\emptyset$.

\begin{definition}[Program configuration]
Given a $\Gamma$-assignment $\rho$, a cr-free HOSC term $\seq{\Gamma}{M:\tau}$ and  $c:\tau$, 
we define the active configuration  $\cconf{M}{\rho,c}$ by
 $\cconf{M}{\rho,c} = \conf{M\{\rho\}, c, \emptyset, \emptyset, \nu(\rho)\cup\{c\}, \emptyset}$.
 \end{definition}
 Note that traces from $\TrR{\HOSC}{\cconf{M}{\rho,c}}$ will be $( \nu(\rho)\cup\{c\},\emptyset)$-traces.
 \begin{definition}
 The \boldemph{$\HOSC[\HOSC]$ trace semantics} of a cr-free HOSC term $\seq{\Gamma}{M:\tau}$ is defined to be
 \[
 \trsem{\HOSC}{\seq{\Gamma}{M:\tau}} = \{  ((\rho,c) ,t)\,|\, \textrm{$\rho$ is a $\Gamma$-assignment},\,c:\tau,\,t\in \TrR{\HOSC}{\cconf{M}{\rho,c}} \}.
 \]
 \end{definition}
 \begin{example}
  Recall  the term $\seq{}{\ex{M_1}{cwl}:\tau}$ from Example~\ref{ex:deriv}, 
  the trace $\tr_1$ and the configuration $\CC$ such that $\tr_1\in\TrR{\HOSC}{\CC}$.
 Because $\ex{M_1}{cwl}$ is closed ($\Gamma=\emptyset$), the only $\Gamma$-assignment is the empty map $\emptyset$.
 Thus, $\CC=\cconf{\ex{M_1}{cwl}}{\emptyset, c}$, so $((\emptyset,c), \tr_1)\in \trsem{\HOSC}{\seq{}{\ex{M_1}{cwl}:\tau}}$.
 \end{example}

Having defined active configurations associated to terms, we now turn to defining passive configurations 
associated to contexts.
Let us fix ${\errn\in\FNames_{\Unit\rarr\Unit}}$ and, for each $\sigma$, a continuation name $\tern_\sigma \in \CNames_\tau$.
Let ${\Large\circ}=\bigcup_\sigma \{\tern_\sigma  \}$.
Intuitively, the names $\errn$ will correspond to $\opredtererr$ and $\tern_\sigma$ to $\opredter$.

Recall that $\hat{\err}$ stands for $\err:\Unit\rarr\Unit$.
Given a heap $h:\Sigma;\herr$,
an evaluation context $\seq{\Sigma;\herr}{K:\tau\rarr\tau'}$ 
and a substitution $\seq{\Sigma;\herr}{\gamma:\Gamma}$ (as in the definition of $\ciupre{\HOSC}{err}$), 
let us replace every occurrence of $\contt{\sigma}{K'}$ inside $h,K,\gamma$
with $\contt{\sigma}{(K',\circ_{\sigma'})}$, if $K'$ has type $\sigma\rarr\sigma'$.
Moreover, let us replace every occurrence of the variable $\err$ with the function name $\errn$.
This is done to adjust $h,K,\gamma$ to the extended syntax of the LTS:
the upgraded versions are called $h_\circ, \gamma_\circ, K_\circ$.

Next we define the set $\AVal{\gamma}{\Gamma}$ of  all disjoint decompositions of values from $\gamma_\circ$ into abstract values and 
the corresponding matchings. Recall that $\Gamma=\{x_1:\sigma_1,\cdots, x_k:\sigma_k\}$. Below
$\vec{A_i}$ stands for $(A_1,\cdots, A_k)$, and $\vec{\gamma_i}$ for $(\gamma_1,\cdots,\gamma_k)$.
\[\begin{array}{rcl}
\AVal{\gamma}{\Gamma} = \{ \quad
(\vec{A_i},\vec{\gamma_i}) &| &(A_i,\gamma_i)\in\AVal{\gamma_\circ(x_i)}{\sigma_i},\,\,i=1,\cdots,k;\\
&&  \nu(A_1),\cdots, \nu(A_k) \textrm{ mutually disjoint and without $\errn$}\quad  \} 
\end{array}\]
\begin{definition}[Context configuration]
Given $\Sigma,\, h:\Sigma;\herr$,\,
$\seq{\Sigma;\herr}{K:\tau\rarr\tau'}$,\,
$\seq{\Sigma;\herr}{\gamma:\Gamma}$,\,
$(\vec{A_i},\vec{\gamma_i})\in \AVal{\gamma}{\Gamma}$
and $c:\tau$ ($c\not\in\circ$), the corresponding configuration $\cconf{h, K,\gamma}{\vec{\gamma_i}, c}$ is defined by
\[
\cconf{h, K,\gamma}{\vec{\gamma_i}, c} = \conf{
\biguplus_{i=1}^k \gamma_i \uplus \{ c \mapsto K_\circ\},
\{ c\mapsto \tern_{\tau'}\},\\
\biguplus_{i=1}^k \nu(A_i) \uplus \{c\} \uplus \circ \uplus \{\errn\}, h_\circ
}.
\]
\end{definition}
Intuitively, the names $\nu(A_i)$ correspond to calling function values extracted from $\gamma$,
whereas $c$ corresponds to $K$.
Note that traces in $\TrR{\HOSC}{\cconf{h, K,\gamma}{\vec{\gamma_i}, c}}$ 
will be $( \circ \uplus \{\errn\}, \biguplus_{i=1}^k \nu(A_i) \uplus \{c\})$-traces.

In preparation for the next result, we introduce the following shorthands.
\begin{itemize}
\item Given a $(N_O,N_P)$-trace $t$, we write $t^\bot$ for the $(N_P,N_O)$-trace obtained by changing the polarity of each name:
$\questO{f}{A}{c'}$ becomes $\questP{f}{A}{c'}$ (and vice versa) and $\ansO{c}{A}$ becomes $\ansP{c}{A}$ (and vice versa).
\item Given $(\vec{A_i},\vec{\gamma_i})\in \AVal{\gamma}{\Gamma}$,
we define a $\Gamma$-assignment $\rho_{\vec{A_i}}$  by $\rho_{\vec{A_i}}(x_i)= A_i$.
Note that $\support{(\rho_{\vec{A_i}})}=\biguplus_{i=1}^k \dom{\gamma_i}$.
\end{itemize}

\begin{lemma}[Correctness]\label{lem:cor}
Let $\seq{\Gamma}{M:\tau}$ be a cr-free $\HOSC$ term,
let $\Sigma,h,K,\gamma$ be as above,
$(\vec{A_i},\vec{\gamma_i})\in \AVal{\gamma}{\Gamma}$,
and $c:\tau$ ($c\not\in\tern$).
Then
\begin{itemize}
\item 
$(K[M\substF{\gamma}],h)\opredtererr$ iff there exist $t,c'$ such that 
$t\in \TrR{\HOSC}{\cconf{M}{\rho_{\vec{A_i}},c}}$ and $t^\bot\, \questP{\errn}{()}{c'} \in \TrR{\HOSC}{\cconf{h, K,\gamma}{\vec{\gamma_i}, c} }$.
\item 
$(K[M\substF{\gamma}],h)\opredter$ iff there exist $t,A,\sigma$ such that 
$t\in \TrR{\HOSC}{\cconf{M}{\rho_{\vec{A_i}},c}}$ and $t^\bot\, \ansP{\tern_{\sigma}}{A} \in \TrR{\HOSC}{\cconf{h, K,\gamma}{\vec{\gamma_i}, c}}$.
\end{itemize}
Moreover, $t$ satisfies $\nu(t)\cap (\circ\cup\{\errn\})=\emptyset$.
\end{lemma}
Intuitively, the lemma above confirms that the potential of a term to converge is determined by its traces.
Accordingly, we have:
 \begin{theorem}[Soundness]\label{hoscsound}
 For any cr-free $\HOSC$ terms $\seq{\Gamma}{M_1,M_2}$, if \\
  $\trsem{\HOSC}{\seq{\Gamma}{M_1}}\subseteq$ $\trsem{\HOSC}{\seq{\Gamma}{M_2}}$ then 
$\seq{\Gamma}{M_1 \ciupre{\HOSC}{err} M_2}$.
  \end{theorem}
  \cutout{
Note that in Lemma~\ref{lem:cor}, because of alternation between O and P,  $t^\bot\, \questP{\errn}{()}{c'}$ will be of even length
and $t$ of odd length.
Consequently, to establish the theorem above, it  would suffice to assume inclusion for odd-length traces.
However, due to the construction of the LTS, odd-length trace inclusion 
for  $\trsem{\HOSC}{\seq{\Gamma}{M_1}}$, $\trsem{\HOSC}{\seq{\Gamma}{M_2}}$ 
implies even-length trace inclusion.
This is because, after an odd-length trace from an active configuration,
the next action will be by O. Such  transitions can always fire
and the label  depends only on the set of names encountered earlier. }

 To prove the converse, we need to know that every odd-length trace generated by a term  actually participates
 in a contextual interaction. This will follow from the lemma below. Note that $\opredtererr$ relies on even-length
 traces from the context (Lemma~\ref{lem:cor}).
 \begin{lemma}[Definability]\label{hoscdef}
Suppose $\phi\uplus\{\errn\}\subseteq\FNames$
and $t$ is an even-length $(\circ\uplus \{\errn\},\phi\uplus\{c\})$-trace starting with an O-action.
There exists a passive configuration $\CC$ such that the even-length traces
 $\TrR{\HOSC}{\CC}$ are exactly the even-length prefixes of $t$
(along with all renamings that preserve types and $\phi\uplus \{c\}\uplus \circ\uplus\{\errn\}$, cf. Remark~\ref{rem:invar}).
Moreover, $\CC=\conf{\gamma_\circ\cdot [c\mapsto K_\circ], \{ c\mapsto \tern_{\tau'} \}, \phi\uplus \{c\}\uplus \circ \uplus \{\errn\},h_\circ}$,
where 
$h,K,\gamma$ are built from $\HOSC$ syntax.
\end{lemma}
\begin{proof}[Sketch]
The basic idea is to use references in order to record all continuation and function names introduced by the environment.
For continuations, the use of $\callcct{\tau}$ is essential.
Once stored in the heap, the names can be accessed by terms when needed in P-actions.
The availability of $\mathrm{throw}$ and references to all O-continuations means that arbitrary answer actions can be scheduled when needed.
\end{proof}

 \begin{theorem}[Completeness]\label{hosccomplete}
 For any cr-free $\HOSC$ terms $\seq{\Gamma}{M_1,M_2}$,
 $\seq{\Gamma}{M_1 \ciupre{\HOSC}{err} M_2}$  implies
 $\trsem{\HOSC}{\seq{\Gamma}{M_1}}\subseteq \trsem{\HOSC}{\seq{\Gamma}{M_2}}$.
  \end{theorem}
Theorems~\ref{hoscsound},~\ref{hosccomplete} (along with Lemmas~\ref{lem:ciu},~\ref{lem:ctxerr}) imply the following full abstraction results.
\begin{corollary}[HOSC Full Abstraction]
Suppose $\seq{\Gamma}{M_1,M_2}$ are cr-free $\HOSC$ terms. Then
$\trsem{\HOSC}{\seq{\Gamma}{M_1}}\subseteq \trsem{\HOSC}{\seq{\Gamma}{M_2}}$ iff
$\seq{\Gamma}{M_1 \ctxpre{\HOSC}{err} M_2}$ iff
$\seq{\Gamma}{M_1 \ctxpre{\HOSC}{ter} M_2}$.
\end{corollary}

\begin{example}[Callback with lock~\cite{ADR09}]\label{ex:callback}
{Recall the term $\seq{}{\ex{M_1}{cwl}: ((\Unit\rarr\Unit)\rarr\Unit)\times(\Unit\rarr\Int)}$ from 
Example~\ref{ex:deriv}, given in Figure~\ref{fig:ex-cwl}.}
%
We had $\tr_1=\ansP{c}{\pair{g_1}{g_2}}$
$\questO{g_1}{f_1}{c_1}$
$\questP{f_1}{()}{c_2}$
$\ansO{c_2}{()}$
$\ansP{c_1}{()}$
$\ansO{c_2}{()}$
$\ansP{c_1}{()}$
$\questO{g_2}{()}{c_3}$
$\ansP{c_3}{2}$ $\in \TrR{\HOSC}{\cconf{\ex{M_1}{cwl}}{{\emptyset,c}}}$.

Define $\tr_2$ to be $\tr_1$ except that its last  action $\ansP{c_3}{2}$  is replaced with $\ansP{c_3}{1}$.
Observe that  
$\tr_1 \in \TrR{\HOSC}{\cconf{\ex{M_1}{cwl}}{{\emptyset,c}}}\setminus\TrR{\HOSC}{\cconf{\ex{M_2}{cwl}}{{\emptyset,c}}}$
and
$\tr_2 \in \TrR{\HOSC}{\cconf{\ex{M_2}{cwl}}{{\emptyset,c}}} \setminus \TrR{\HOSC}{\cconf{\ex{M_1}{cwl}}{{\emptyset,c}}}$,
i.e. by the Corollary above the terms are incomparable wrt $\ctxpre{\HOSC}{err}$.
However, they are  equivalent wrt $\ctxpre{\xx}{err}$ for $\xx\in\{\GOSC,\HOS,\GOS\}$~\cite{DNB12}.
\end{example}
The above Corollary also provides a handle to reason about equivalence via trace equivalence. 
Sometimes this can be done directly on the LTS, especially {when $\gamma$ 
can be kept bounded}.


\begin{example}[Counter~\cite{PS98}]\label{ex:counter}
For $i\in\{1,2\}$,
consider the terms $\seq{}{M_i: (\Unit\rarr\Unit)\times (\Unit\rarr\Int)}$ given by
 $M_i\equiv \letin{x=\nuref{0}}{\pair{\mathsf{inc}_i}{\mathsf{get}_i}}$,
 where $\mathsf{inc}_1\equiv (\lambda y. x:=!x+1)$, $\mathsf{inc}_2\equiv (\lambda y. x:=!x-1)$, 
 $\mathsf{get}_1\equiv \lambda z. !x$, $\mathsf{get}_2\equiv\lambda z. -!x$.
 In this case, $\TrR{\HOSC}{\cconf{M_i}{{\emptyset,c}}}$ contains (prefixes of)
traces of the form $\ansP{c}{\pair{g}{h}}\, t$, where $t$ is built from 
segments of two kinds: either $\questO{g}{()}{c_i}$ $\ansP{c_i}{()}$
 or $\questO{h}{()}{c_i'}$ $\ansP{c_i'}{n}$, where the $c_i$s and $c_i'$s are pairwise different.
  Moreover,  in the latter case,  $n$ must be equal to the number of preceding actions of the form $\questO{g}{()}{c_i}$.
 For this example, trace equality could be established by induction on the length of trace.
 Consequently, $M_1\ctxequ{\HOSC}{err} M_2$.
 \end{example}


\section{GOSC[HOSC]}\label{sec:goschosc}

Recall that $\GOSC$ is the fragment of $\HOSC$ in which general storage is restricted to values of \emph{ground} type, 
i.e. arithmetic/boolean constants, the associated reference names, references to those names and so on.
In what follows, we are going to provide characterizations of  $\ctxpre{\GOSC}{err}$ via  trace inclusion.
Recall that, by Lemma~\ref{lem:ctxerr}, $\ctxpre{\GOSC}{err} = \ctxpre{\GOSC}{ter}$.
Note that we work in an asymmetric setting with terms belonging to $\HOSC$ being more powerful than contexts.

We start off by identifying several technical consequences of the restriction to $\GOSC$ syntax.
First we observe that $\GOSC$ internal reductions never contribute extra names.
\begin{lemma}\label{lem:goscinv}
Suppose $(M,c,h)\ered (M',c',h')$, where $M$ is a $\GOSC$ term and $h$ is a $\GOSC$ heap.
Then $\nu(M)\cup\{c\}\supseteq \nu(M')\cup\{c'\}$.
\end{lemma}
\begin{proof}
By case analysis. All defining rules for $\ered$, with the  exception of the $  (K[!\ell],h)   \red  (K[h(\ell)],h)$ rule, 
are easily seen to satisfy the Lemma (no function or continuation names are added).
However, if the heap is restricted to storing elements of type $\iota$ (as in $\GOSC$)
then  $h(\ell)$ will never contain a name, so the Lemma follows.
\end{proof}
The lemma has interesting consequences for the shape of traces generated by  the context
configurations $\cconf{h, K,\gamma}{\vec{\gamma_i}, c}$ if they are built from $\GOSC$ syntax.
Recall that P-actions have the form $\questP{f}{A}{c'}$ or $\ansP{c}{A}$,
where $f,c$ are names introduced by O.
It turns out that when $h,K,\gamma$ are restricted to $\GOSC$, more can be said about the origin of the names
in traces generated by $\cconf{h, K,\gamma}{\vec{\gamma_i}, c}$:
they will turn out to come from a restricted set of names introduced by O, which we identify below.
The definition below is based on following the justification structure of a trace -- recall that one action 
is said to justify another if the former introduces a name that is used for communication in the latter.
\begin{definition}
Suppose $\phi\uplus\{\errn\}\subseteq\FNames$ and $c\in\CNames$.
Let $t$ be an odd-length $(\circ\uplus \{\errn\}, \phi\uplus\{c\})$-trace starting with an O-action.
The set   $\pav{t}$ of {P-visible names} of $t$ is defined as follows.
\[\begin{array}{rclcl}
\pav{t \,\,\, \ansO{c'}{A'}} & = & \{\errn\}\cup \circ \cup \nu(A') &\qquad& c'=c\\
\pav{t \,\,\,\questP{f''}{A''}{c'}\,\,\, t'\,\,\, \ansO{c'}{A'}} & = & \pav{t} \cup \nu(A') && c'\neq c\\
\pav{t \,\,\, \questO{f'}{A'}{c'}} & = & \{\errn \}\cup\circ\cup\nu(A')\cup\{c'\}  &&f'\in\phi\\
\pav{t \,\,\,\questP{f''}{A''}{c''}\,\,\, t'\,\,\, \questO{f'}{A'}{c'}} & = & \pav{t}\cup\nu(A')\cup\{c'\}  &&f'\in\nu(A'')\\
\pav{t \,\,\,\ansP{c''}{A''}\,\,\, t'\,\,\, \questO{f'}{A'}{c'}} & = & \pav{t}\cup\nu(A')\cup\{c'\}  &&f'\in\nu(A'')
\end{array}\]
\end{definition}


Note that, in the inductive cases, the definition follows links between names introduced by P 
and the point of their introduction, names introduced in-between are ignored.
Here readers familiar with game semantics will notice similarity to the notion of P-view~\cite{HO00}.

Next we specify a property of traces that will turn out to be satisfied by configurations corresponding to $\GOSC$ contexts.
\begin{definition}
Suppose $\phi\uplus\{\errn\}\subseteq\FNames$ and $c\in\CNames$.
Let $t$ be a $(\circ\uplus \{\errn\},\phi\uplus\{c\})$-trace starting with an O-action.
$t$ is called \boldemph{P-visible} if
\begin{itemize}
\item for any even-length prefix $t' \,\questP{f}{A}{c}$ of $t$, we have $f \in\pav{t'}$,
\item for any even-length prefix $t'\,\ansP{c}{A}$ of $t$, we have $c\in\pav{t'}$.
\end{itemize}
\end{definition}

\begin{lemma}\label{lem:pvis}
Consider $\CC=\cconf{h, K,\gamma}{\vec{\gamma_i}, c}$, where $h,K,\gamma$ are from $\GOSC$ 
and $(\vec{A_i},\vec{\gamma_i})\in \AVal{\gamma}{\Gamma}$.
Then all traces in  $\TrR{\HOSC}{\CC}$ are P-visible.
\end{lemma}
\cutout{
\begin{proof}[Proof sketch]
We reason by induction on the length of the trace  but, to enable an inductive proof, the statement above needs to be
generalized to components $M',c',\gamma',\xi'$  of reachable configurations. For example, 
 if an active configuration with $M',c'$ is reached from $\CC_O$ via $t$,
then $\nu{(M')}\cup\{c'\}\subseteq \pav{t}$. Similarly, if names $f',c'$ are introduced in a P-action after $t$, one shows
that $\nu(\gamma'(f')),\nu(\gamma'(c')),\xi(c')$ must come from $\pav{t}$. 
\end{proof}}
The Lemma above shows that contextual interactions with $\GOSC$ contexts rely on restricted traces.
We shall now modify the  $\HOSC[\HOSC]$ LTS to capture the restriction. Note that, from the perspective of the term,
the above constraint is a constraint on the use of names by O (context), so we need to talk about O-available names instead.
This dual notion is defined below.
\begin{definition}
Suppose $\phi\subseteq\FNames$ and $c\in \CNames$.
Let $t$ be a $(\phi\uplus\{c\},\emptyset)$-trace of odd length.
The set   $\oav{t}$ of \boldemph{O-visible names} of $t$ is defined as follows.
\[\begin{array}{rclcl}
\oav{t \,\,\, \ansP{c'}{A'}} & = & \nu(A') && c'=c\\
\oav{t \,\,\,\questO{f''}{A''}{c'}\,\,\, t'\,\,\, \ansP{c'}{A'}} & = & \oav{t} \cup \nu(A') && c'\neq c\\
\oav{t \,\,\, \questP{f'}{A'}{c'}} & = & \nu(A')\cup\{c'\}  &&f'\in\phi\\
\oav{t \,\,\,\questO{f''}{A''}{c''}\,\,\, t'\,\,\, \questP{f'}{A'}{c'}} & = & \oav{t}\cup\nu(A')\cup\{c'\}  &&f'\in\nu(A'')\\
\oav{t \,\,\,\ansO{c''}{A''}\,\,\, t'\,\,\, \questP{f'}{A'}{c'}} & = & \oav{t}\cup\nu(A')\cup\{c'\}  &\qquad &f'\in\nu(A'')\\
\end{array}\]
Analogously, a $(\phi\uplus\{c\},\emptyset)$-trace $t$ is \boldemph{O-visible} if, for any 
even-length prefix $t' \,\questO{f}{A}{c}$ of $t$, we have $f \in\oav{t'}$ and, 
for any even-length prefix $t'\,\ansO{c}{A}$ of $t$, we have $c\in\oav{t'}$.
\end{definition}
\begin{example}\label{ex:ovis}
Recall the trace 
\[
\tr_1=\ansP{c}{\pair{g_1}{g_2}}\,\,\,
\questO{g_1}{f_1}{c_1}\,\,\,
\questP{f_1}{()}{c_2}\,\,\,
\ansO{c_2}{()}\,\,\,
\ansP{c_1}{()}\,\,\,
\ansO{c_2}{()}\,\,\,
\ansP{c_1}{()}\,\,\,
\questO{g_2}{()}{c_3}\,\,\,
\ansP{c_3}{2}
\]
from previous examples. 
Observe that 
\[\begin{array}{rcl}
\oav{
\ansP{c}{\pair{g_1}{g_2}}\,\,
\questO{g_1}{f_1}{c_1}\,\,
\questP{f_1}{()}{c_2}} &= & \{ g_1,g_2, c_2\}\\
\oav{
\ansP{c}{\pair{g_1}{g_2}}\,\,
\questO{g_1}{f_1}{c_1}\,\,
\questP{f_1}{()}{c_2}\,\,
\ansO{c_2}{()}\,\,
\ansP{c_1}{()}} &= & \{ g_1,g_2\}
\end{array}
\]
Consequently, the first use of $\ansO{c_2}{()}$ in $\tr_1$ does not violate O-visibility, but the second one does.
\end{example}
\begin{figure}[t]
\[
\begin{array}{l|l@{}ll}
 (P\tau) & \conf{M,c,\gamma,\xi,\phi,h,\ViewF}  &\quad  \ired{\ \tau \ }
 & \conf{N,c',\gamma,\xi,\phi,h',\ViewF} \\
 & \multicolumn{3}{l}{\text{ when } (M,c,h) \ered (N,c',h')}\\
 (PA) & \conf{V,c,\gamma,\xi,\phi,h,\ViewF} 
 & \ired{\ansP{c}{A}} 
 & \conf{\gamma \cdot \gamma', \xi,
     \phi\uplus\nu(A),h, \ViewF, \ViewF(c) \uplus \nu(A)} \\
 & \multicolumn{3}{l}{\text{ when } c:\sigma
 \text{ and } (A,\gamma') \in \AVal{V}{\sigma}}\\
 (PQ) & \conf{K[fV],c,\gamma,\xi,\phi,h,\ViewF} 
 & \ired{\questP{f}{A}{c'}} 
 & \conf{  \gamma \cdot\gamma'\cdot [c' \mapsto K],
    \xi\cdot[c' \mapsto c],\phi\uplus\phi',h, \ViewF, \ViewF(f) \uplus \phi' } \\
 & \multicolumn{3}{l}{\text{ when } f:\sigma \rarr \sigma', \,(A,\gamma') \in \AVal{V}{\sigma}, \, c':\sigma' \text{ and } \phi'=\nu(A)\uplus\{c'\}}\\
 (OA) & \conf{\gamma,\xi,\phi,h,\ViewF,\View} 
 & \ired{\ansO{c}{A}} 
 & \conf{K[A],c',\gamma,\xi,
    \phi \uplus \nu(A),h, \ViewF\cdot[\nu(A)\mapsto \View]} \\
 & \multicolumn{3}{l}{\text{ when } c \in \View,\, c:\sigma,\, A:\sigma,\, \gamma(c) = K,\, \xi(c)=c'}\\
 (OQ) & \conf{\gamma,\xi,\phi,h, \ViewF,\View} 
 & \ired{\questO{f}{A}{c}}
 & \conf{V A,c,\gamma,\xi,   \phi \uplus \phi',h,\ViewF \cdot [\phi' \mapsto \View]} \\
 & \multicolumn{3}{l}{\text{ when } f \in \View,\, f:\sigma \rarr \sigma',\, A:\sigma,\, c:\sigma',\,
  \gamma(f) = V   \text{ and } \phi' = \nu(A)\uplus \{c\}}\\
\cutout{
 (IOQ) & \conf{\Gamma \vdash M:\tau} 
 & \ired{\questOinit{\vec{A}}{c}} 
 & \conf{M\{\delta\},c,\emptymap,\emptymap             ,[c \mapsto \varnothing],
    \phi \uplus \{c\},\emptymap} \\
 & \multicolumn{3}{l}{\text{ when } (\vec{A},\delta,\phi) \in \sem{\Gamma}\}}}\\
  \multicolumn{4}{l}{\text{Given $N\subseteq\Names$, $[N\mapsto\View]$ stands for the map $[n\mapsto\View\,|\, n\in N]$.}}
  \end{array}
 \]

\caption{$\GOSC[\HOSC]$ LTS}\label{fig:gosc}
\end{figure}
In Figure~\ref{fig:gosc}, we present a new LTS, called the  $\GOSC[\HOSC]$ LTS, 
which will turn out to capture $\ctxpre{\GOSC}{err}$ through trace inclusion.
It is obtained from the $\HOSC[\HOSC]$ LTS by restricting O-actions to those that rely on O-visible names.
Technically, this is done by enriching configurations with an additional component $\ViewF$, which maintains
historical information about O-available names immediately before each O-action.
After each P-action, $\ViewF$ is accessed to calculate the current set $\View$ of O-available names 
according to the definition of O-availability and only O-actions compatible with O-availability are allowed to proceed
(due to the $f\in\View$, $c\in\View$ side conditions). We write $\TrR{\GOSC}{\CC}$ for the set of traces generated 
from $\CC$ in the $\GOSC[\HOSC]$ LTS.

Recall that, given a $\Gamma$-assignment $\rho$, term $\seq{\Gamma}{M:\tau}$ and  $c\in\CNames_\tau$, 
the active configuration  $\cconf{M}{\rho,c}$ was defined by
 $\cconf{M}{\rho,c} = \conf{M\{\rho\}, c, \emptyset, \emptyset, \nu(\rho)\cup\{c\}, \emptyset}$.
 We need to upgrade it to the LTS by initializing the new component to the empty map:
$\cconf{M,\mathit{vis}}{\rho,c}= \conf{M\{\rho\}, c, \emptyset, \emptyset, \nu(\rho)\cup\{c\}, \emptyset,\emptyset}$.
 \begin{definition}
 The \boldemph{$\GOSC[\HOSC]$ trace semantics} of a cr-free HOSC term $\seq{\Gamma}{M:\tau}$ is defined to be
 \[
 \trsem{\GOSC}{\seq{\Gamma}{M:\tau}} = \{  ((\rho,c) ,t)\,|\, \textrm{$\rho$ is a $\Gamma$-assignment},\,c:\tau,\,t\in 
 \TrR{\GOSC}{\cconf{M,\mathit{vis}}{\rho,c}} \}.
 \]
   \end{definition}
 By construction, it follows that
\begin{lemma}\label{lem:ovis}
$t\in \TrR{\GOSC}{\cconf{M,\mathit{vis}}{\rho,c}}$ iff $t\in  \TrR{\HOSC}{\cconf{M}{\rho,c}}$ and $t$ is O-visible.
\end{lemma}
Noting that the witness trace $t$ from Lemma~\ref{lem:cor} is O-visible iff $t^\bot\, \questP{\errn}{()}{c'}$ is P-visible,
we can conclude that, for $\GOSC$, the traces relevant to $\opredtererr$ are O-visible, which yields:
 \begin{theorem}[Soundness]\label{goscsound}
 For any cr-free $\HOSC$ terms $\seq{\Gamma}{M_1,\,\, M_2}$,
 if\\ $\trsem{\GOSC}{\seq{\Gamma}{M_1}} \subseteq \trsem{\GOSC}{\seq{\Gamma}{M_2}}$ then
  $\seq{\Gamma}{M_1 \ciupre{\GOSC}{err} M_2}$.
  \end{theorem}
  
  \cutout{
    \begin{proof}
Suppose $\trsem{\GOSC}{\seq{\Gamma}{M_1}}\subseteq \trsem{\GOSC}{\seq{\Gamma}{M_2}}$.
Consider  $\Sigma, h,K,\gamma$ (as in the definition of $\ciuapperr^\GOSC$) such that $(K[M_1\substF{\gamma}],h)\opredtererr$.
In particular, $h,K,\gamma$ consist of $\GOSC$ syntax.
Suppose $(\vec{A_i},\vec{\gamma_i})\in\AVal{\gamma}{\Gamma}$ and $c:\sigma$ ($c\neq\tern$).
By Lemma~\ref{lem:cor} (left-to-right), there exist $t,c'$ such that
$t\in \TrR{\HOSC}{\cconf{M_1}{\rho_{\vec{A_i}},c}}$ and $t^\bot\, \questP{\errn}{()}{c'} \in \TrR{\HOSC}{\cconf{h, K,\gamma}{\vec{\gamma_i}, c} }$.
By  Lemma~\ref{lem:pvis}, $t^\bot\, \questP{\errn}{()}{c'}$ is P-visible. Thus, $t$ is O-visible and, by Lemma~\ref{lem:ovis} (right-to-left),  
$t\in\trsem{\GOSC}{\cconf{M_1}{\rho_{\vec{A_i}},c}}$.
From $\trsem{\GOSC}{\seq{\Gamma}{M_1}}\subseteq \trsem{\GOSC}{\seq{\Gamma}{M_2}}$, we get $t\in \TrR{\GOSC}{\cconf{M_2}{\rho_{\vec{A_i}},c}}$.
By Lemma~\ref{lem:ovis} (left-to-right), $t\in \TrR{\HOSC}{\cconf{M_2}{\rho_{\vec{A_i}},c}}$.
Because $t\in \TrR{\HOSC}{\cconf{M_2}{\rho_{\vec{A_i}},c}}$ and 
$t^\bot\, \questP{\errn}{()}{c'} \in \TrR{\HOSC}{\cconf{h, K,\gamma}{\vec{\gamma_i}, c} }$, by Lemma~\ref{lem:cor} (right-to-left), we
can conclude $(K[M_2\substF{\gamma}],h)\opredtererr$. Thus, $\seq{\Gamma}{M_1 \ciuapperr^\GOSC M_2}$.
  \end{proof}}
To prove the converse, we need a new definability result. This time we are only allowed to use $\GOSC$ syntax, 
but the target is also more modest: we are only aiming to capture P-visible traces.
  \begin{lemma}[Definability]\label{goscdef}
Suppose $\phi\uplus\{\errn\}\subseteq\FNames$
and $t$ is an even-length \emph{P-visible} $(\circ\uplus \{\errn\},\phi\uplus\{c\})$-trace starting with an O-action.
There exists a passive configuration $\CC$ such that the even-length traces
 in $\TrR{\HOSC}{\CC}$ are exactly the even-length prefixes of $t$
(along with all renamings that preserve types and $\phi\uplus \{c\}\uplus \circ\uplus\{\errn\}$).
Moreover, $\CC=\conf{\gamma_\circ\cdot [c\mapsto K_\circ], \{ c\mapsto \tern_{\tau'} \}, \phi\uplus \{c\}\uplus \circ \uplus \{\errn\},h_\circ}$,
where 
$h,K,\gamma$ are built from $\GOSC$ syntax.
\end{lemma}
\cutout{
\begin{lemma}[Definability]\label{goscdef}
Suppose $\phi\uplus\{\errn\}\subseteq\FNames$,
and $t=o_1 p_1\cdots o_n p_n$ is a P-visible $(\{\errn,\tern\},\phi\uplus\{c\})$-trace such that $\tern\not\in\nu(t)$.
There exists a passive configuration $\CC_O$ built from $\GOSC$ syntax such that
 $\TrReven{\HOSC}{\CC_O}$ consists of all even-length prefixes of $t$
(along with their renamings via type-preserving permutations that fix $\phi\uplus\{c,\errn,\tern\}$, cf. Lemma~\ref{lem:inv}).
Moreover, $\CC_O=\conf{\gamma_O, \{ c\mapsto \tern \}, \phi\uplus \{c,\errn,\tern\},h}$,
where $\dom{\gamma_O}=\phi\uplus\{c\}$, 
$\nu({\gamma})=\{\errn\}$ and $\nu({h}) =\emptyset$.
\end{lemma}}
\begin{proof}[Sketch]
This time we cannot rely on references to recall on demand all continuation and function names introduced by the environment.
However, because $t$ is P-visible, it turns the uses of the names can be captured 
through variable bindings ($\lambda x.\cdots$ for function and $\callcct{\tau}{(x.\dots)}$ for continuation names).
Using $\mathrm{throw}$, we can then force an arbitrary answer action, as long as it uses a P-available name.
To select the right action at each step, we branch on the value of a single global reference of type $\reftype{\,\Int}$ 
that keeps track of the number of steps simulated so far.
\end{proof}
Completeness now follows because, for a potential  O-visible witness $t$ from Lemma~\ref{lem:cor},
one can create a corresponding context by invoking the Definability result for $t^\bot\, \questP{\errn}{()}{c'}$.
It is crucial that the addition of $\questP{\errn}{()}{c'}$ does not break P-visibility ($\errn$ is P-visible).
 \begin{theorem}[Completeness]\label{gosccomplete}
 For any cr-free $\HOSC$ terms $\seq{\Gamma}{M_1,M_2}$,
 if  $\seq{\Gamma}{M_1 \ciupre{\GOSC}{err} M_2}$ then 
 $\trsem{\GOSC}{\seq{\Gamma}{M_1}}\subseteq \trsem{\GOSC}{\seq{\Gamma}{M_2}}$.
  \end{theorem}
\cutout{\begin{proof}
We follow the same path as in the proof of Theorem~\ref{hosccomplete} except that, in this case, we will have 
$t,t_1\in \TrR{\GOSC}{\cconf{M_1}{\rho_{\vec{A_i}},c}}$.
Consequently, we can conclude that $t_2 =t_1^{\bot}\,\questP{\errn}{()}{c'}$ is P-visible and invoke Lemma~\ref{goscdef} (instead of Lemma~\ref{hoscdef})
to obtain $C_O$ that corresponds to $h,K,\gamma$ from $\GOSC$. 
Because $k,K,\gamma$ are in $\GOSC$, we can then appeal to the assumption $\seq{\Gamma}{M_1 \ciuapperr^\GOSC M_2}$
and complete the proof like for Theorem~\ref{hosccomplete}.
\end{proof}}
Altogether, Theorems~\ref{goscsound},~\ref{gosccomplete} (along with Lemma~\ref{lem:ciu}) imply 
the following result.
\begin{corollary}[$\GOSC$ Full Abstraction]
Suppose $\seq{\Gamma}{M_1,M_2}$ are cr-free $\HOSC$ terms. Then
 $\trsem{\GOSC}{\seq{\Gamma}{M_1}}\subseteq \trsem{\GOSC}{\seq{\Gamma}{M_2}}$ iff
$\seq{\Gamma}{M_1 \ciupre{\GOSC}{err} M_2}$ iff
$\seq{\Gamma}{M_1 \ctxpre{\GOSC}{err} M_2}$.
\end{corollary}
\begin{example}
In the \emph{Callback with lock} example (Example~\ref{ex:callback}), 
we exhibited traces $\tr_1, \tr_2$ that separated $\ex{M_1}{cwl},\ex{M_2}{cwl}$ 
wrt $\ctxpre{\HOSC}{err}$.
Example~\ref{ex:ovis} shows that neither trace is O-visible, i.e. they cannot be found in $\trsem{\GOSC}{\seq{\Gamma}{M_1}}$
or $\trsem{\GOSC}{\seq{\Gamma}{M_2}}$. Thus, the two traces cannot be used
to separate $\ex{M_1}{cwl}, \ex{M_2}{cwl}$ wrt $\ctxpre{\GOSC}{err}$. As already mentioned, this is in fact impossible: 
we have $\seq{}{\ex{M_1}{cwl}\ctxequ{\GOSC}{err} \ex{M_2}{cwl}}$.
\end{example}
\begin{example}[Well-bracketed state change~\cite{ADR09}]\label{ex:bracket}
Consider the following two terms
\[\begin{array}{rcl}
\ex{M_1}{wbsc}&\defeq&\letin{x=\nuref{0}}{\lambda f. (x:=0; f(); x:=1; f(); !x)}\\
\ex{M_2}{wbsc}&\defeq &\lambda f. (f();f(); 1).
\end{array}\]
of type $\tau=(\Unit\rarr\Unit)\rarr\Int$,  let 
\[
\tr_3=\ansP{c}{g}\,\,\,\,
\questO{g}{f_1}{c_1}\,\,\,\,
\questP{f_1}{()}{c_2}\,\,\,\,
\ansO{c_2}{()}\,\,\,\,
\questP{f_1}{()}{c_3}\,\,\,\,
\questO{g}{f_2}{c_4}\,\,\,\,
\questP{f_2}{()}{c_5}\,\,\,\,
\ansO{c_3}{()}\,\,\,\,
\ansP{c_1}{0}\]
and  let $\tr_4$ be obtained from $\tr_3$ by changing $0$ in the last action to $1$.
One can check that both traces are O-visible: in particular, the action $\ansO{c_3}{()}$ is not a violation
because 
\[\oav{ \ansP{c}{g}\,\,
\questO{g}{f_1}{c_1}\,\,
\questP{f_1}{()}{c_2}\,\,
\ansO{c_2}{()}\,\,
\questP{f_1}{()}{c_3}\,\,
\questO{g}{f_2}{c_4}\,\,
\questP{f_2}{()}{c_5}} = \{ g, c_3,c_5 \}.
\]
Moreover, we have 
$\tr_3 \in \TrR{\GOSC}{\cconf{\ex{M_1}{wbsc}}{{\emptyset,c}}}\setminus\TrR{\GOSC}{\cconf{\ex{M_2}{wbsc}}{{\emptyset,c}}}$
and
$\tr_4 \in \TrR{\GOSC}{\cconf{\ex{M_2}{wbsc}}{{\emptyset,c}}} \setminus \TrR{\GOSC}{\cconf{\ex{M_1}{wbsc}}{{\emptyset,c}}}$.
By the Corollary above, we can conclude that $\ex{M_1}{wbsc},\ex{M_2}{wbsc}$ are incomparable wrt $\ctxpre{\GOSC}{err}$.
However, they turn out to be $\ctxequ{\HOS}{err}$- and $\ctxequ{\GOS}{err}$-equivalent.
\end{example}


\section{HOS[HOSC]}\label{sec:hoshosc}

Recall that $\HOS$ is the fragment of $\HOSC$ that does not feature continuation types and the associated syntax.
In what follows we are going to provide alternative characterisations of $\ctxpre{\HOS}{err}$ and $\ctxpre{\HOS}{ter}$ 
in terms of trace inclusion and complete trace inclusion respectively. 

We start off by identifying several technical consequences of the restriction to $\HOS$ syntax.
First we observe that $\HOS$ internal reductions never change the associated continuation name.
\begin{lemma}\label{lem:hosinv}
If $(M,c,h)\ered (M',c',h')$,  $M$ is a $\HOS$ term and $h$ is a $\HOS$ heap then $c=c'$.
\end{lemma}
\begin{proof}
The only rule that could change $c$ is the rule for $\mathrm{throw}$, but it is not part of $\HOS$.
\end{proof}
The lemma has a bearing on the shape of traces generated by  the (passive) 
configurations $\cconf{h, K,\gamma}{\vec{\gamma_i}, c}$ corresponding to $\HOS$ contexts.
In the presence of $\mathrm{throw}$ and storage for continuations, it was possible for P to play answers involving 
arbitrary continuation names introduced by O. By Lemma~\ref{lem:hosinv}, in $\HOS$ this will be restricted to the continuation name
of the current configuration, which will restrict the shape of possible traces.
Below we identify the continuation name $\topp{t}$ that becomes the relevant name after trace $t$.
If the last move was an O-question then the continuation name introduced by that move will become that name.
Otherwise, we track a chain of answers and questions, similarly to the definition of P-visibility.

Observe that, because $h,K,\gamma$ are from $\HOS$,
$\cconf{h, K,\gamma}{\vec{\gamma_i}, c}$ will generate
$( \{\circ_{\tau'},\errn\}, \phi\uplus \{c\})$-traces, where $\tau'$ is the result type of $K$, because $h_\circ=h,K_\circ=K, \gamma_\circ=\gamma$.
\begin{definition}\label{def:pbra}
Suppose $\phi\uplus\{\errn\}\subseteq\FNames$ and $c\in\CNames$.
Let $t$ be a $(\{\circ_{\tau'},\errn\}, \phi\uplus\{c\})$-trace of odd length starting with an O-action.
The continuation name $\topp{t}$ is defined as follows.
\[\begin{array}{rcll}
\topp{t \,\, \ansO{c}{A}} & = & \tern_{\tau'}\\
\topp{t_1 \,\,\questP{f}{A''}{c'}\,\, t_2\,\, \ansO{c'}{A'}} & = & \topp{t_1}\\
\topp{t\,\, \questO{f}{A'}{c'}} &=& c'
\end{array}\]
We say that a  $(\{\circ_{\tau'}\cup \{\errn\}, \phi\uplus\{c\})$-trace $t$  starting with an O-action is \boldemph{P-bracketed} if, 
for any prefix $t' \,\,\ansP{c'}{A}$ of $t$  (i.e. any prefix ending 
with a P-answer), we have $c' =\topp{t'}$.
\end{definition}

\begin{lemma}\label{lem:pbracket}
Consider $\CC=\cconf{h, K,\gamma}{\vec{\gamma_i}, c}$, where $h,K,\gamma$ are from $\HOS$ 
and $(\vec{A_i},\vec{\gamma_i})\in \AVal{\gamma}{\Gamma}$.
Then all  traces in $\TrR{\HOSC}{\CC}$ are P-bracketed.
\end{lemma}

The Lemma above characterizes the restrictive nature of contextual interactions with $\HOS$ contexts.
Next we shall constrain the  $\HOSC[\HOSC]$ LTS accordingly to capture the restriction. Note that, from the point of view of the term,
the above-mentioned constraint concerns the use of continuation names by O (the context), so we need to talk about O-bracketing instead.
This dual notion of ``a top name for O'' is specified below.
\begin{definition}
Suppose $\phi\subseteq\FNames$ and $c\in \CNames$.
Let $t$ be a $(\phi\uplus\{c\},\emptyset)$-trace of odd length.
The continuation name $\topo{t}$ is defined as follows.
In the first case, the value is $\bot$ (representing ``none''), 
because $c$ is the top continuation passed by the environment to the term (if it gets answered
there is nothing left to answer).
\[\begin{array}{rcll}
\topo{t \,\, \ansP{c}{A}} & = & \bot\\
\topo{t_1 \,\,\questO{f}{A''}{c'}\,\, t_2\,\, \ansP{c'}{A'}} & = & \topo{t_1}\\
\topo{t\,\, \questP{f}{A'}{c'}} &=& c'
\end{array}\]
We say that a  $(\phi\uplus\{c\},\emptyset)$-trace $t$ is \boldemph{O-bracketed} if, for any prefix $t' \,\,\ansP{c'}{A}$ of $t$  (i.e. any prefix ending 
with a P-answer), we have $c' = \topo{t'}$.
\end{definition}
\begin{figure}[t]
\[  \begin{array}{l|lll}
   (P\tau) & \conf{M,c,\gamma,\xi, \phi,h} & \ired{\ \tau \ } & 
     \conf{N,c',\gamma,\xi, \phi,h'} \\
   & \multicolumn{3}{l}{\text{ when } (M,c,h) \ered (N,c',h')}\\
   (PA) & \conf{V,c,\gamma,\xi,\phi,h} & \ired{\ansP{c}{A}} & 
     \conf{\gamma \cdot \gamma',\xi,\phi\uplus\nu(A),h,c'} \\
      & \multicolumn{3}{l}{\text{ when } c:\sigma,\,  (A,\gamma') \in \AVal{V}{\sigma},\,\xi(c)=c' } \\  
   (PQ) & \conf{K[fV],c,\gamma,\xi,\phi,h} & \ired{\questP{f}{A}{c'}} & 
    \conf{\gamma \cdot\gamma'\cdot[c' \mapsto K],\xi\cdot [c' \mapsto c],
    \phi\uplus\nu(A)\uplus \{c'\},h, c'} \\
    & \multicolumn{3}{l}{\text{ when } f:\sigma \rarr \sigma',\, (A,\gamma') \in \AVal{V}{\sigma},\, c':\sigma'}\\
   (OA) & \conf{\gamma,\xi,\phi,h,c''} & \ired{\ansO{c}{A}} & 
    \conf{K[A],c',\gamma,\xi,\phi \uplus \nu(A),h} \\
    & \multicolumn{3}{l}{\text{ when }  c=c'',\, c:\sigma,\,A:\sigma,\,\gamma(c) = K,\,   \xi(c)=c'}\\
   (OQ) & \conf{\gamma,\xi,\phi,h,c''} & \ired{\questO{f}{A}{c}} & 
     \conf{V A,c,\gamma,\xi\cdot [c\mapsto c''],\phi \uplus \nu(A)  \uplus \{c\},h} \\
     & \multicolumn{3}{l}{\text{ when } f:\sigma \rarr \sigma',\,  A:\sigma,\, c:\sigma',\,\gamma(f) = V }\\[3mm]
  \end{array}
 \]
\caption{$\HOS[\HOSC]$ LTS}\label{fig:hos}
\end{figure}

In Figure~\ref{fig:hos}, we present a new LTS, called the  $\HOS[\HOSC]$ LTS, 
which will turn out to capture $\ctxpre{\HOS}{err}$.
It is obtained from the $\HOSC[\HOSC]$ LTS by restricting O-actions to those that satisfy O-bracketing.
Technically, this is done by enriching passive configurations with a component for storing  the current 
value of $\topo{t}$.  In order to maintain this information, we need to know which continuation will
become the top one if P plays an answer. This can be done with a map that maps continuations 
introduced by O to other continuations. Because its flavour is similar to $\xi$ (which is a map from continuations
introduced by P) we integrate this information into $\xi$. The $c=c''$ side condition then enforces O-bracketing.
We shall write $\TrR{\HOS}{\CC}$ for the set of traces generated 
from $\CC$ in the $\HOS[\HOSC]$ LTS.

Recall that, given a $\Gamma$-assignment $\rho$, term $\seq{\Gamma}{M:\tau}$ and  $c:\tau$, 
the active configuration  $\cconf{M}{\rho,c}$ was defined by
 $\cconf{M}{\rho,c} = \conf{M\{\rho\}, c, \emptyset, \emptyset, \nu(\rho)\cup\{c\}, \emptyset}$.
 We upgrade it to the new LTS by setting
$\cconf{M,\mathit{bra}}{\rho,c}= \conf{M\{\rho\}, c, \emptyset,[c\mapsto \bot], \nu(\rho)\cup\{c\}, \emptyset,\emptyset}$.
This initializes $\xi$ in such a way that, after $\ansP{c}{A}$ is played, the extra component will be set to $\bot$,
where $\bot$ is a special element not in $\CNames$.


{
 \begin{definition}
 The \boldemph{$\HOS[\HOSC]$ trace semantics} of a cr-free HOSC term $\seq{\Gamma}{M:\tau}$ is defined to be
 \[
 \trsem{\HOS}{\seq{\Gamma}{M:\tau}} = \{  ((\rho,c) ,t)\,|\, \textrm{$\rho$ is a $\Gamma$-assignment},\,c:\tau,\,t\in 
 \TrR{\HOS}{\cconf{M,\mathit{bra}}{\rho,c}} \}.
 \]
   \end{definition}}
 By construction, it follows that
\begin{lemma}\label{lem:obracket}
$t\in \TrR{\HOS}{\cconf{M,\mathit{bra}}{\rho,c}}$ iff $t\in  \TrR{\HOSC}{\cconf{M}{\rho,c}}$ and $t$ is O-bracketed.
\end{lemma}
Noting that the witness trace $t$ from Lemma~\ref{lem:cor} is O-bracketed iff $t^\bot\, \questP{\errn}{()}{c'}$ is P-bracketed,
we can conclude that, for $\HOS$, the traces relevant to $\opredtererr$ are O-bracketed, which yields:
 \begin{theorem}[Soundness]\label{hossound}
 For any cr-free $\HOSC$ terms $\seq{\Gamma}{M_1, M_2}$,  if\\ $\trsem{\HOS}{\seq{\Gamma}{M_1}}$ $\subseteq$ $\trsem{\HOS}{\seq{\Gamma}{M_2}}$ then
  $\seq{\Gamma}{M_1 \ciupre{\HOS}{err} M_2}$.
  \end{theorem}
For the converse,  we establish another definability result, this time for a P-bracketed trace.
 \begin{lemma}[Definability]\label{lem:hosdef}
Suppose $\phi\uplus\{\errn\}\subseteq\FNames$
and $t$ is an even-length \emph{P-bracketed} $( \{\circ_{\tau'},\errn\},\phi\uplus\{c\})$-trace starting with an O-action.
There exists a passive configuration $\CC$ such that the even-length traces
 $\TrR{\HOSC}{\CC}$ are exactly the even-length prefixes of $t$
(along with all renamings that preserve types and $\phi\uplus \{c,\circ_{\tau'},\errn\}$).
Moreover, $\CC=\conf{\gamma\cdot [c\mapsto K], \{ c\mapsto \tern_{\tau'} \}, \phi\uplus \{c,\circ_{\tau'},\errn\},h}$,
where 
$h,K,\gamma$ are built from $\HOS$ syntax.
\end{lemma}
\begin{proof}[Sketch]
Our argument for $\HOSC$ is structured in such a way that, for a P-bracketed trace, there is no need for continuations 
(throwing and continuation capture are not necessary).
\end{proof}
Completeness now follows because, for a potential witness trace $t$ from Lemma~\ref{lem:cor},
one can create a corresponding context by invoking the Definability result for $t^\bot\, \questP{\errn}{()}{c'}$.
It is crucial that the addition of $\questP{\errn}{()}{c'}$ does not break P-bracketing (it does not, because the action is a question).
 \begin{theorem}[Completeness]\label{hoscomplete}
 For any cr-free $\HOSC$ terms $\seq{\Gamma}{M_1,M_2}$,
 if  $\seq{\Gamma}{M_1 \ciupre{\HOS}{err} M_2}$ then 
 $\trsem{\HOS}{\seq{\Gamma}{M_1}}$ $\subseteq$ $\trsem{\HOS}{\seq{\Gamma}{M_2}}$.
  \end{theorem}
Altogether, Theorems~\ref{hossound},~\ref{hoscomplete} (along with Lemma~\ref{lem:ciu}) imply 
the following result.
\begin{corollary}[$\HOS$ Full Abstraction]
Suppose $\seq{\Gamma}{M_1,M_2}$ are cr-free $\HOSC$ terms. Then
 $\trsem{\HOS}{\seq{\Gamma}{M_1}}\subseteq \trsem{\HOS}{\seq{\Gamma}{M_2}}$ iff
$\seq{\Gamma}{M_1 \ciupre{\HOS}{err} M_2}$ iff
$\seq{\Gamma}{M_1 \ctxpre{\HOS}{err} M_2}$.
\end{corollary}
\begin{example}[Assignment/callback commutation~\cite{MT11b}]
For $i\in\{1,2\}$, let $\seq{f:\Unit\rarr\Unit}{M_i:\Unit\rarr\Unit}$ be defined by:
\[\begin{array}{rcl}
M_1 &\defeq& \letin{n=\nuref{(0)}}{ \lambda y^\Unit.}{ \ifte{(!n>0)}{()}{(n:=1;f()) }},\\
M_2&\defeq& \letin{n=\nuref{(0)}}{ \lambda y^\Unit. \ifte{(!n>0)}{()}{(f(); n:=1) }}.
\end{array}\]
Operationally, one can see that $\seq{f}{M_1\not\ctxpre{\HOS}{err} M_2}$ due to the following
$\HOS$ context:
$\letin{r=\nuref{}(\lambda y.y)}{(\letin{f=\lambda y.(!r)()}{(r:=\bullet; (!r)())});\err}$.
 In our framework, this is confirmed by the trace
 \[
 \tr_5\quad=\quad \ansP{c}{g}\quad\questO{g}{()}{c_1}\quad\questP{f}{()}{c_2}\quad\questO{g}{()}{c_2}\quad\ansP{c_2}{()},
 \]
 which is in  $\TrR{\HOS}{\cconf{M_1}{\rho,c}}\setminus\TrR{\HOS}{\cconf{M_2}{\rho,c}}$.
 On the other hand, 
 \[
 \tr_6\quad=\quad\ansP{c}{g}\quad\questO{g}{()}{c_1}\quad\questP{f}{()}{c_2}\quad\questO{g}{()}{c_2}\quad\questP{f}{()}{c_3}
 \]
 is in $\TrR{\HOS}{\cconf{M_2}{\rho,c}}\setminus\TrR{\HOS}{\cconf{M_1}{\rho,c}}$, so the terms are incomparable.
Note, however, that both traces break O-visibility: specifically, we have
\[
\oav{\ansP{c}{g}\,\,\questO{g}{()}{c_1}\,\,\questP{f}{()}{c_2}}=\{c_2\}, 
\]
so the $\questO{g}{()}{c_2}$ action violates the condition.
Consequently, the traces do not preclude $\seq{f}{M_1\ctxequ{\xx}{err}{M_2}}$ for $\xx\in\{\GOSC,\GOS\}$.
\end{example}
For $\xx\in\{\HOSC,\GOSC\}$,  $\ctxpre{\xx}{err}$ and $\ctxpre{\xx}{ter}$ coincide.
Intuitively, this is because the presence of continuations in the context makes it possible
to make an escape at any point.
In contrast, for $\HOS$, the context must run to completion in order to terminate.

At the technical level, one can appreciate the difference when trying to transfer our results for $\ciupre{\HOS}{err}$ to $\ciupre{\HOS}{ter}$.
Recall that, according to Lemma~\ref{lem:cor}, $\opredter$ relies on a witness trace $t$ such that the context configuration
generates $t^\bot\, \ansP{\tern_{\tau'}}$.
In $\HOS$, the latter must satisfy P-bracketing, so we need $\topp{t^\bot} =\circ_{\tau'}$. Note that this is equivalent to $\topo{t}=\bot$.
Consequently, only such traces are relevant to observing $\opredter$.

Let us call an odd-length O-bracketed $(\phi\uplus\{c\},\emptyset)$-trace $t$ \boldemph{complete} if $\topo{t}=\bot$.
Let us write $\trsem{\HOS}{\seq{\Gamma}{M_1}}\subseteq_c \trsem{\HOS}{\seq{\Gamma}{M_2}}$
if we have $((\rho,c),t) \in \trsem{\HOS}{\seq{\Gamma}{M_2}}$ whenever
$((\rho,c),t) \in \trsem{\HOS}{\seq{\Gamma}{M_1}}$ and $t$ is complete.
Following our methodology, one can then show:
\begin{theorem}[$\HOS$ Full Abstraction for $\ctxpre{\HOS}{ter}$]
Suppose $\seq{\Gamma}{M_1,M_2}$ are cr-free $\HOSC$ terms. Then
 $\trsem{\HOS}{\seq{\Gamma}{M_1}}\subseteq_c \trsem{\HOS}{\seq{\Gamma}{M_2}}$ iff
$\seq{\Gamma}{M_1 \ciupre{\HOS}{ter} M_2}$ iff
$\seq{\Gamma}{M_1 \ctxpre{\HOS}{ter} M_2}$.
\end{theorem}
\begin{example}
Let  $M_1\equiv\lambda f^{\Unit\rarr\Unit}. f();\Omega_\Unit$ and $M_2\equiv \lambda f^{\Unit\rarr\Unit}. \Omega_\Unit$.
We will see that $\seq{}{M_1\not\ctxpre{\HOS}{err} M_2}$ but $\seq{}{M_1\ctxpre{\HOS}{ter} M_2}$.
To see this, note that $\TrR{\HOS}{\cconf{M_1}{\rho,c}}$ contains prefixes of $\ansP{c}{g}\,\, \questO{g}{f}{c_1}\,\,\questP{f}{()}{c_2}\,\,\ansO{c_2}{()}$,
while $\TrR{\HOS}{\cconf{M_2}{\rho,c}}$ only those of $\ansP{c}{g}\,\, \questO{g}{f}{c_1}$.
Observe that the only complete trace among them is $\ansP{c}{g}$. The trace $t=\ansP{c}{g}\,\, \questO{g}{f}{c_1}\,\,\questP{f}{()}{c_2}$ is not complete, because $\topo{t}=c_2$. Consequently,   $\trsem{\HOS}{\seq{\Gamma}{M_1}}\not\subseteq \trsem{\HOS}{\seq{\Gamma}{M_2}}$ but
$\trsem{\HOS}{\seq{\Gamma}{M_1}}\subseteq_c \trsem{\HOS}{\seq{\Gamma}{M_2}}$.
\end{example}
The theorem above generalizes the characterisation of contextual equivalence between $\HOS$ terms with respect to $\HOS$ contexts~\cite{Lai07},
where trace completeness means both O- and P-bracketing and ``all questions must be answered".  Our definition of completeness is weaker (O-bracketing + ``the top question must be answered"),
because it also covers $\HOSC$ terms. However, in the presence of both O- and P-bracketing, i.e. for $\HOS$ terms, they will coincide.


\section{GOS[HOSC]\label{sec:goshosc}}

Recall that $\GOS$ features ground state only and, technically, is the intersection of $\GOSC$ and $\HOS$.
Consequently, it follows from the previous sections that $\GOS$ contexts yield configurations that satisfy both P-visibility and P-bracketing.
For such traces, the definability result for $\GOSC$ yields a $\GOS$ context. Thus, in a similar fashion to the previous sections,
we can conclude that O-visible and O-bracketed traces underpin $\ctxpre{\GOS}{err}$.
To define the $\GOS$ LTS we simply combine the restrictions imposed in the previous sections, and define $\trsem{\GOS}{\seq{\Gamma}{M}}$ analogously.
We present the LTS in Appendix~\ref{apx:gos}. The results on $\ctxpre{\GOS}{ter}$ from the previous section  
also carry over to $\GOS$.

\begin{theorem}[$\GOS$ Full Abstraction]\label{gosfull}
Suppose $\seq{\Gamma}{M_1,M_2}$ are cr-free $\HOSC$ terms. Then:
\begin{itemize}
\item $\trsem{\GOS}{\seq{\Gamma}{M_1}}\subseteq \trsem{\GOS}{\seq{\Gamma}{M_2}}$ iff
$\seq{\Gamma}{M_1 \ciupre{\GOS}{err} M_2}$ iff
$\seq{\Gamma}{M_1 \ctxpre{\GOS}{err} M_2}$.
\item 
$\trsem{\GOS}{\seq{\Gamma}{M_1}}\subseteq_c \trsem{\GOS}{\seq{\Gamma}{M_2}}$ iff
$\seq{\Gamma}{M_1 \ciupre{\GOS}{ter} M_2}$ iff
$\seq{\Gamma}{M_1 \ctxpre{\GOS}{ter} M_2}$.
\end{itemize}
\end{theorem}


\section{Concluding remarks}
\label{sec:extensions}

\paragraph{Asymmetry} 

Our framework is able to deal with asymmetric scenarios, where
programs are taken from $\HOSC$, but are tested with contexts 
from  weaker fragments.
\cutout{This form of asymmetry would be rather hard to achieve for proof techniques
whose soundness proof relies on being a congruence together with an adequacy result,
like environmental bisimulations or 
logical relations. Indeed, $\ctxerr{\HOS}$, seen as a relation over $\HOSC$ terms,
is not a congruence over $\HOSC$ contexts, but only over $\HOS$ contexts.}
For example, we can compare the following two $\HOSC$ programs, where $f:((\Unit\rarr\Unit)\rarr \Unit)\rarr\Unit$
is a free identifier.
\[\begin{array}{l  c  l}
\mathrm{let \ b=ref\ \falseML \ in \ callcc(y.} && \mathrm{callcc(y.} \\
\qquad \mathrm{ f(\lambda g. b:=\trueML; g(); throw() \ to \ y);}
&\qquad \qquad& \qquad \mathrm{ f(\lambda g. g(); throw() \ to \ y);} \\
\qquad \mathrm{if \ !b \ then \ () \ else \ div)} && \qquad \mathrm{div)}
\end{array}\]
with $\mathrm{div}$ representing divergence.
The terms happen to be $\ctxequ{\HOS}{err}$-equivalent, but not $\ctxequ{\HOSC}{\err}$-equivalent.

To see this at the intuitive level, we make the following observations.

\begin{itemize}
\item Firstly, we observe that, to distinguish the terms, $\mathrm{f}$ should use its argument.
Otherwise, the value of $\mathrm{b}$ will  remain equal to $\falseML$, and the only subterm
that distinguishes the terms (`$\mathrm{if \ !b \ then \ () \ else \ div}$') will play the same role as  $\mathrm{div}$
in the second term.

\item Secondly, if $\mathrm{f}$ does use its argument, then $\mathrm{b}$ will be set to $\trueML$ in the first 
program, raising the possibility of distinguishing the terms.
However, if we allow $\HOS$ contexts only then, since the argument to $\mathrm{f}$ was
used, it will have to run to completion, before `$\mathrm{if \ !b \ then \ () \ else \ div}$' is reached.
Consequently, we will  encounter `$\throwto{()}{y}$' earlier and never reach  `$\mathrm{if \ !b \ then \ () \ else \ div)}$'. 
This is represented by the trace

\[ 
\questP{f}{h}{c_1} \qquad
\questO{h}{g}{c_2}\qquad
\questP{g}{()}{c_3} \qquad \ansO{c_3}{()} \qquad \ansP{c}{()}\]

This trace is O-bracketed, but not $P$-bracketed since Player uses $\throw$
to answer directly to the initial continuation $c$ rather than $c_2$.

\item Finally, if $\HOSC$ contexts are allowed, it is possible to reach  `$\mathrm{if \ !b \ then \ () \ else \ div)}$'
 $\mathrm{b}$ set to $\trueML$. This is represented by the trace
 
\[ 
\questP{f}{h}{c_1} \qquad
\questO{h}{g}{c_2}\qquad
\questP{g}{()}{c_3} \qquad \ansO{c_1}{()} \qquad \ansP{c}{()}
\]
This trace is not O-bracketed, because $c_1$ is answered rather than $c_3$, like above.
Consequently, the trace witnesses termination of the first term, but the second term would diverge during interaction with the same context.
\end{itemize}
\cutout{It is interesting to notice that to be sound, trace equivalence does not need to be 
a congruence, but rather that from a given term and an
evaluation context, there exists a trace representing their interaction.
So the trace has to be proven being in the denotation of the term, and its dual in
the denotation of the context.}

We plan to explore the opportunities presented by this setting in the future, especially with respect to
fully abstract translations, for example, from $\HOSC$ to $\GOS$.

\paragraph{Richer Types} Recall that our full abstraction results are stated for cr-free terms,
terms with cont- and ref-free types at the boundary.
Here we first discuss how to extend them to more complicated types.

To deal with reference type at the boundary, i.e. location exchange, 
one needs to generalize the notion of
traces, so that they can carry, for each action, a
heap representing the values stored in the disclosed part of the heap,
as in~\cite{Lai07,MT11b}. \cutout{KOBs
have then to check the relatedness of the disclosed part of the 
heaps at each interaction, as in~\cite{JT15}.
In fact, this is also done in the definition of 
environmental bisimulations~\cite{SKS11} and Kripke Logical Relations~\cite{PS98}.
Notice that Opponent creation of references
become apparent when he provides a fresh location to Player, either via a question or
an answer, so it has to be taken care of in the definition of
$\vrbs{\tau}{}{w}$ and $\krbs{\tau}{\sigma}{}{w}$.

In the $\HOS$ and $\GOS$ cases, we have seen that KOBs capture $\ctxequ{}{\err}$
but not $\ctxequ{}{\ter}$. There are examples, like the so-called
``deferred divergence'' one, that are contextually equivalent 
but cannot be handled by KOBs.
To deal with them, we have to reason on equivalence of complete traces rather than
any traces. This means allowing momentarily an inequivalence,
as soon as the two terms always diverge in the future. To do so,
one can use inconsistent worlds or configurations~\cite{DNB12,JT15,BLP19}.
} The extension to sum, recursive and empty types seems conceptually straightforward,
by simply extending the definition of abstract values for these types,
following the similar notion of ultimate pattern in~\cite{LL07}.
The same idea should apply to allow continuation types at the boundary.
Operational game semantics for an extension of HOS with polymorphism 
has been explored in~\cite{JT16}.

\paragraph{Innocence} On the other hand, all of the languages we considered were stateful. In the presence of state, all of the actions that are represented by labels
(and their order and frequency) can be observed, because they could generate a side-effect.
A natural question to ask whether the techniques could also be used to provide analogous theorems for purely functional computation, i.e. contexts taken from the language PCF.
Here, the situation is different. For example, the terms $\seq{f:\Int\rarr\Int}{f(0)}$ and $\seq{f:\Int\rarr\Int}{\ifte{f(0)}{f(0)}{f(0)}}$ should be equivalent, even though the sets of their 
traces are incomparable.

It is known~\cite{HO00} that PCF strategies  satisfy a uniformity condition called innocence. Unfortunately, restricting our traces to ``O-innocent ones" (like we did with O-visibility and O-bracketing) 
would not deliver the required  characterization. Technically, this is due to the fact that, in our arguments, given a single trace (with suitable properties), 
we can produce a context that induces the given trace and no other traces (except those implied by the definition of a trace).
For innocence, this would not be possible due to the uniformity requirement. It will imply that, although we can find a functional context that generates an innocent trace, 
it might also generate other traces, which then have to be taken into account when considering contextual testing. This branching property makes it difficult to capture equivalence
with respect to functional contexts explicitly, e.g. through traces, which is illustrated by the use of the so-called intrinsic quotient in game models of PCF~\cite{AJM00,HO00}.


\cutout{
\amin{This section will come at the end, before conclusions. 
We need to explain why the methodology is applicable to other types.
Perhaps we should say that the LTS would be rather complicated?}

Points to make

\begin{itemize}
\item Covering reference types requires us to complicate labels.

\item Continuation types?

\item Recursive types?

\item Polymorphism?

\end{itemize}
}


\section{Related Work}

We have presented four operational game models for $\HOSC$, which capture term interaction with
contexts built from any of the four sublanguages $\xx\in\{\HOSC,$ $\GOSC,$ $\HOS,$ $\GOS\}$ respectively.
The most direct precursor to this work  is Laird's trace model for $\HOS[\HOS]$~\cite{Lai07}.
Other frameworks in this spirit include models for objects~\cite{JR05b}, aspects~\cite{JPR07} and system-level code~\cite{GTz12}.
In~\cite{Jab15}, Laird's model has been related formally to the denotational game model from~\cite{MT11b}.
However, in general, it is not yet clear how one can move systematically between the operational and denotational 
game-based approaches,  despite some promising steps reported in~\cite{LS14}.
Below  we mention other operational techniques for reasoning about contextual equivalence.

In~\cite{SL07}, fully abstract Eager-Normal-Form (enf) Bisimulations are presented for 
an untyped $\lambda$-calculus with store and control, similar to $\HOSC$ 
(but with control represented using the $\lambda\mu$-calculus).
The bisimulations are parameterised by worlds to model the evolution of store,
and bisimulations on contexts are used to deal with control.
Like our approach, they are based on symbolic evaluation of open terms.
Typed enf-bisimulations, for a language without store and in control-passing style, have been introduced 
in~\cite{LL07}.
Fully-abstract enf-bisimulations are presented in~\cite{BLP19} for a language with state only, 
corresponding to an untyped version of $\HOS$. Earlier works in this strand include~\cite{JR99,San93}.


Environmental Bisimulations~\cite{KW06,SKS11,Sum09} have also been introduced for languages
with store. They work on closed terms, computing the arguments that contexts can provide
to terms using an environment similar to our component $\gamma$.
They have also been extended to languages with call/cc~\cite{yachi2016sound}
and delimited control operators~\cite{lmcs:3942,biernacki2013environmental}.



Kripke Logical Relations~\cite{PS98,ADR09,DNB12} have been introduced for languages with state and control.
In~\cite{DNB12}, a characterization of contextual equivalence for each case $\xx[\xx]$ ($\xx\in\{\HOSC,\GOSC,\HOS,\GOS\}$)
is given, using  techniques called backtracking and public transitions, which exploit the 
absence of higher-order store and that of control constructs respectively.
Importing these techniques in the setting of Kripke Open Bisimulations~\cite{JT15}
should allow one to build a bridge between the game-semantics characterizations
and Kripke Logical Relations.

Parametric bisimulations~\cite{HDNV12} have been introduced as an operational technique,
merging ideas from Kripke Logical Relations and Environmental Bisimulations.
They do not represent functional values coming from the environment using names, but instead 
use a notion of global and local knowledge to compute these values, reminiscent of the work on
environmental bisimulations. The notion of global knowledge depends itself on a notion of evolving world.
To our knowledge, no fully abstract Parametric Bisimulations have been presented.

A general theory of applicative~\cite{DLGL17} and normal-form bisimulations~\cite{DLG19}
has been developed, with the goal of being modular with respect to the effects considered.
While the goal is similar to our work, the papers consider monadic and algebraic presentation of effects,
trying particularly to design a general theory for proving soundness and completeness of such bisimulations.
These works complement ours, and we would like to explore possible connections.

\cutout{

\begin{itemize}
\item Talk about predecessors of DNB: \cite{PS98}, \cite{ADR09}

We do not need step-indexing in our definitions. This is because we do not 
use universal quantification over related values in negative positions,
as it is done in the relation over values and evaluations contexts in
logical relations. So we can use directly a coinductive definition.

\item To formally link KLR and KOBs, one need to develop a general compositionality (i.e. congruence)
study of KOBs, so that we can reason on them by induction on the typing derivation trees of the
considered related terms. Relate to~\cite{LS14}

\item We want to develop a general theory of up-to techniques and abstractions that 
would encompass Kripke open, environmental, and enf- bisimulations. 
Cite Madiot, Pous \& Sangiorgi.

\item Talk about Open Logical Relations by Barthe, Crubillé, Dal Lago \& Gavazzo

\item Comparison with game semantics: \cite{Mur03,MT11b}, Laird on control ?
\end{itemize}
}

\bibliographystyle{splncs04}
\bibliography{my}

\appendix

\section{Additional material for Section~\ref{sec:hosc} (HOSC)}

\subsection{Type System}
\label{apx:type}

Please see Figure~\ref{fig:typing-rules}.

\begin{figure}
 \begin{mathpar}
 \inferrule*{ }{\Sigma;\Gamma \vdash \unit :\Unit}
 
 \inferrule*{ }{\typingTerm{\Sigma;\Gamma}{\trueML}{\Bool}}
 
 \inferrule*{ }{\typingTerm{\Sigma;\Gamma}{\falseML}{\Bool}}

 \inferrule*{ }{\Sigma;\Gamma \vdash \nb{n} :\Int}
  
 \inferrule*{(x,\tau) \in \Gamma }{\Sigma;\Gamma \vdash x : \tau}
 
 \inferrule*{(\ell,\tau) \in \Sigma }
  {\Sigma;\Gamma \vdash \ell : \reftype{\tau}}
 
 \inferrule*{\typingTerm{\Sigma;\Gamma}{M}{\sigma} \\ \typingTerm{\Sigma;\Gamma}{N}{\tau}}
  {\typingTerm{\Sigma;\Gamma}{\pair{M}{N}}{\sigma \times \tau}}

 \inferrule*{\typingTerm{\Sigma;\Gamma}{M}{\tau_1 \times \tau_2}}
  {\typingTerm{\Sigma;\Gamma}{\pi_i M}{\tau_i}}
  
 \inferrule*{\Sigma;\Gamma,x:\sigma \vdash M : \tau}
  {\Sigma;\Gamma \vdash \lambda x^\sigma. M : \tau}

 \inferrule*{\Sigma;\Gamma,f:\sigma \rightarrow \tau, x:\sigma \vdash M : \tau}
  {\Sigma;\Gamma \vdash \fix{f}{}{x^\sigma}{M} : \sigma \rightarrow \tau}

 \inferrule*{\Sigma;\Gamma \vdash M : \sigma \rightarrow \tau  \\
   \Sigma;\Gamma \vdash N : \sigma} {\Sigma;\Gamma \vdash M N : \tau}  \\
 
 \inferrule*{\Sigma;\Gamma \vdash M :  \tau }
   {\Sigma;\Gamma \vdash \nureft{\tau}{M} : \reftype{\tau}} 
 
 \inferrule*{\Sigma;\Gamma \vdash M :  \reftype{\tau}}
  {\Sigma;\Gamma \vdash !M : \tau}

 \inferrule*{\Sigma;\Gamma \vdash M : \reftype{\tau} \\ 
  \Sigma;\Gamma \vdash N : \tau}{\Sigma;\Gamma \vdash M := N : \Unit} \\
 
 \inferrule*{\typingTerm{\Sigma;\Gamma}{M_1}{\Bool} \\ \typingTerm{\Sigma;\Gamma}{M_2}{\tau}
             \\ \typingTerm{\Sigma;\Gamma}{M_3}{\tau}}
    {\typingTerm{\Sigma;\Gamma}{\ifte{M_1}{M_2}{M_3}}{\tau}}

 \inferrule*{\typingTerm{\Sigma;\Gamma}{M_1}{\Int} \\
   \typingTerm{\Sigma;\Gamma}{M_2}{\Int}}{\Sigma;\Gamma \vdash M_1 \oplus M_2 :\Int}

 \inferrule*{\typingTerm{\Sigma;\Gamma}{M_1}{\Int} \\
    \typingTerm{\Sigma;\Gamma}{M_2}{\Int}}{\Sigma;\Gamma \vdash M_1 \boxdot M_2 :\Bool}

 \inferrule*{\typingTerm{\Sigma;\Gamma}{M_1}{\reftype{\tau}} \\
    \typingTerm{\Sigma;\Gamma}{M_2}{\reftype\tau}}{\Sigma;\Gamma \vdash M_1 = M_2 :\Bool} \\
     
 \inferrule*{\typingTerm{\Sigma;\Gamma,x:{\tau}}{K[x]}{\sigma}}
   {\Sigma;\Gamma \vdash \contt{\tau}{K} :\conttype{\tau}}
      
 \inferrule*{\typingTerm{\Sigma;\Gamma,x:\conttype{\tau}}{M}{\tau}}
   {\Sigma;\Gamma \vdash \callcct{\tau}{(x.M)} :\tau}
   
 \inferrule*{\typingTerm{\Sigma;\Gamma}{M}{\sigma}
 \\ \typingTerm{\Sigma;\Gamma}{N}{\conttype{\sigma}}}
   {\Sigma;\Gamma \vdash \throwtot{\tau}{M}{N} :\tau} 
 \end{mathpar}
 \caption{$\HOSC$ typing rules}
 \label{fig:typing-rules}
\end{figure}


\subsection{Proof of Lemma~\ref{lem:ciu} (CIU)}

\cutout{
\begin{quote}
{\sc Definition.} Let $\xx\in\{\HOSC,\GOSC,\HOS,\GOS\}$.
Given $\HOSC$-terms $\seq{\Gamma}{M_1,M_2:\sigma}$ with an $\xx$ boundary,
we define $\Gamma \vdash M_1 \ciupre{\xx}{\yy} M_2:\sigma$ to hold, when for all 
$\Sigma$, 
$\cseq{\Sigma}{K}{\sigma}$,
$\seq{\Sigma}{\gamma:\Gamma}$, $h$, all from $\xx$, 
 we have $(K[M_1\substF{\gamma}],h) \opredobs{\yy}$ implies $(K[M_2\substF{\gamma}],h) \opredobs{\yy}$.
\end{quote}}

In~\cite{HMST95,Tal98}, the authors propose general frameworks for establishing CIU theorems for higher-order languages with
effects and control.  The results are based on the usual contextual testing observing termination.
Below we repeat the pattern of their argument in our framework for both $\ciupre{\xx}{ter}$ and $\ciupre{\xx}{err}$.
The names of the lemmas come from Section 2.3 of~\cite{HMST95}. Their technical aim is to establish that
each relation is a precongruence.

Let $\yy\in\{ \mathit{ter},\mathit{err} \}$.

\begin{lemma}[Op CIU]
Suppose $M_1\ciupre{\xx}{\yy} M_2$. Then, whenever the terms are typable and the relevant operation
is allowable in an $\xx$-context, we have:
\begin{itemize}
\item 
$\pair{M_1}{M}\ciupre{\xx}{\yy} \pair{M_2}{M}$, \,\,
$\pi_i M_1 \ciupre{\xx}{\yy} \pi_i M_2$,\,\,
$M_1 M \ciupre{\xx}{\yy} M_2 M$, \,\,
$\nuref{M_1}\ciupre{\xx}{\yy} \nuref{M_2}$,\,\, 
$!M_1\ciupre{\xx}{\yy} !M_2$,\,\,
$M_1:= M \ciupre{\xx}{\yy} M_2:=M$,\,\,  
$\ifte{M_1}{M}{M'}\ciupre{\xx}{\yy} \ifte{M_2}{M}{M'}$, \,\,
$M_1\oplus M\ciupre{\xx}{\yy} M_2\oplus M$,\,\,
$M_1\boxdot M\ciupre{\xx}{\yy} M_2\boxdot M$,\,\, 
$M_1= M\ciupre{\xx}{\yy} M_2 = M$,\,\,
$\throwto{M_1}{M}\ciupre{\xx}{\yy}\throwto{M_2}{M}$;

\item $\pair{M}{M_1}\ciupre{\xx}{\yy} \pair{M}{M_2}$,  \,\,
 $M M_1 \ciupre{\xx}{\yy} M M_2$, \,\,
$M:= M_1 \ciupre{\xx}{\yy} M:=M_2$,   \,\, 
$\ifte{M}{M_1}{M'}\ciupre{\xx}{\yy} \ifte{M}{M_2}{M'}$, \,\,  
$\ifte{M}{M'}{M_1}\ciupre{\xx}{\yy} \ifte{M}{M'}{M_2}$, \,\,
$M\oplus M_1\ciupre{\xx}{\yy} M\oplus M_2$,  \,\,
$M\boxdot M_1\ciupre{\xx}{\yy} M\boxdot M_2$,  \,\,
$M= M_1\ciupre{\xx}{\yy} M= M_2$, \,\,
$\throwto{M}{M_1}\ciupre{\xx}{\yy}\throwto{M}{M_2}$.
\end{itemize}
\end{lemma}
\begin{proof}
We handle the first case from each category, as the rest are analogous.
\begin{itemize}
\item Suppose $K,\gamma,h$ are such that $(K[\pair{M_1}{M}\substF{\gamma}],h)\opredobs{\yy}$. 

Observe that $K[\pair{M_1}{M}\substF{\gamma}]=K[\pair{M_1\substF{\gamma}}{M\substF{\gamma}}] = K'[M_1\substF{\gamma}]$ for some $K'$.

Because $M_1\ciupre{\xx}{\yy} M_2$ and $(K'[M_1\substF{\gamma}],h)\opredobs{\yy}$, we get  $(K'[M_2\substF{\gamma}],h)\opredobs{\yy}$.

Because $K[\pair{M_2}{M}\substF{\gamma}]=K[\pair{M_2\substF{\gamma}}{M\substF{\gamma}}]=K'[M_2\substF{\gamma}]$, 
this implies $(K[\pair{M_2}{M}\substF{\gamma}],h)\opredobs{\yy}$, as needed.

\item 
Suppose $K,\gamma,h$ are such that $(K[\pair{M}{M_1}\substF{\gamma}]\,h) \opredobs{\yy}$. 
We need to show $(K[\pair{M}{M_2}\substF{\gamma}],h)\opredobs{\yy}$.

Observe that $K[\pair{M}{M_1}\substF{\gamma}]=K[\pair{M\substF{\gamma}}{M_1\substF{\gamma}}]$.

We will argue by induction on the number of transitions in 
$(K[\pair{M\substF{\gamma}}{M_1\substF{\gamma}}],h)\opredobs{\yy}$ for all $M\substF{\gamma},h$. 

Because of $(K[\pair{M\substF{\gamma}}{M_1\substF{\gamma}},h) \opredobs{\yy}$, we have 
the following cases for $M\substF{\gamma}$.

\begin{itemize}
\item ($M\substF{\gamma} = V$)

In this case, $K[\pair{M}{M_1}\substF{\gamma}]= K[\pair{V}{M_1\substF{\gamma}}] = K'[M_1\substF{\gamma}]$.

By $M_1\ciupre{\xx}{\yy} M_2$, we get $(K'[M_2\substF{\gamma}],h)\opredobs{\yy}$. 

Because $K[\pair{M_2}{M}\substF{\gamma}]=K'[M_2\substF{\gamma}]$,
we obtain $(K[\pair{M_2}{M}\substF{\gamma}],h)\opredobs{\yy}$, as needed.

\item ($M\substF{\gamma} = K'[\err()]$, only for $y=\mathit{err}$)

Here $K[\pair{M\substF{\gamma}}{M_1\substF{\gamma}}]$ in $0$ steps, and it follows that
$(K[\pair{M\substF{\gamma}}{M_2\substF{\gamma}}],h)\opredobs{\yy}$.

\item ($M\substF{\gamma} = K'[N]$ such that $(K'[N],h)\ired{} (K'[N'],h')$) 

$(K[\pair{M\substF{\gamma}}{M_1\substF{\gamma}}],h)= 
(K[\pair{K'[N]\substF{\gamma}}{M_1\substF{\gamma}}],h)\ired{}
(K[\pair{K'[N']\substF{\gamma}}{M_1\substF{\gamma}}],h') \opredobs{\yy}$.

By IH, $(K[\pair{K'[N']\substF{\gamma}}{M_2\substF{\gamma}}],h')\opredobs{\yy}$.

Hence, because 
$(K[\pair{M\substF{\gamma}}{M_2\substF{\gamma}}],h)\ired{}
(K[\pair{K'[N']\substF{\gamma}}{M_2\substF{\gamma}}],h')$,
we have\\
$(K[\pair{M\substF{\gamma}}{M_2\substF{\gamma}}],h)\opredobs{\yy}$.

Note that this case also covers the reduction rule for $\callcc$.

\item ($M\substF{\gamma} = K'[\throwto{V}{\cont{K''}}]$)

In this case, $(K[\pair{M\substF{\gamma}}{M_1\substF{\gamma}}],h)\red (K''[V],h)$ and $(K''[V],h)\opredobs{\yy}$.
Note that then $(K[\pair{M\substF{\gamma}}{M_2\substF{\gamma}}],h) \red (K''[V],h)$ too, so we are done.
\end{itemize}
\end{itemize}
\end{proof}

\begin{lemma}[Lambda CIU]
$M_1\ciupre{\xx}{\yy} M_2$ implies $\lambda x.M_1\ciupre{\xx}{\yy} \lambda x. M_2$.
\end{lemma}
\begin{proof}
Take $K,\gamma,h$ such that $(K[(\lambda x.M_1)\substF{\gamma}],h)\opredobs{\yy}$.
Let us write $M_i^{\gamma}$ for $M_i\substF{\gamma}$.
Note that $(\lambda x.M_1)\substF{\gamma}= \lambda x. M_1^{\gamma}$.
We need to show $(K[\lambda x.M_2^{\gamma}],h)\opredobs{\yy}$.

Instead we shall show that
$M\subst{z}{\lambda x.M_1^{\gamma}}\opredobs{\yy}$ implies 
$M\subst{z}{\lambda x.M_2^{\gamma}}\opredobs{\yy}$ for any $\seq{\Sigma; z}{M}$. 
The Lemma then follows by taking  $M=K[z]$.

We use induction on the number of steps $k$ in $(M\subst{z}{\lambda x.M_1^{\gamma}},h)\opredobs{\yy}$ for  all $M,h$.

Suppose $(M\subst{z}{\lambda x.M_1^{\gamma}},h)\opredobs{\yy}$.
\begin{itemize}
\item If $k=0$ and $y=\mathit{err}$ then $M=K'[\err ()]$. Thus, 
$(M\subst{z}{\lambda x.M_2^{\gamma}},h)\opredobs{err}$ too.
\item If $k=0$ and $y=\mathit{ter}$ then $M=V\subst{z}{\lambda x.M_2^{\gamma}}$ or $M=z$.
In both cases, $M\subst{z}{\lambda x.M_2^{\gamma}}$ is a value, and $M\subst{z}{\lambda x.M_2^{\gamma}}\opredobs{ter}$.
\item Suppose $k>0$. Because $(M\subst{z}{\lambda x.M_1^{\gamma}},h)\opredobs{\yy}$, the following cases arise.
\begin{itemize}
\item ($M=K'[N]$ and $(K'[N],h)\rarr (K'[N'],h')$)

Then  $(K'[N'] \subst{z}{\lambda x.M_1^{\gamma}},h')\opredobs{\yy}$ in $(k-1)$ steps. 

So, by IH, $(K'[N']\subst{z}{\lambda x.M_2^{\gamma}},h')\opredobs{\yy}$.

Because $(M\subst{z}{\lambda x.M_2^{\gamma}},h)\rarr (K'[N']\subst{z}{\lambda x.M_2^{\gamma}},h')$, we are done.

\item ($M=K'[\throwto{V}{\cont{K''}}]$)

Then  $(K''[V] \subst{z}{\lambda x.M_1^{\gamma}},h)\opredobs{\yy}$ in $(k-1)$ steps.

So, by IH, $(K''[V]\subst{z}{\lambda x.M_2^{\gamma}},h)\opredobs{\yy}$.

Because $(M\subst{z}{\lambda x.M_2^{\gamma}},h)\rarr (K''[V]\subst{z}{\lambda x.M_2^{\gamma}},h)$, we are done.

\item ($M=K'[z V]$) 

Then $(K'[M_1^{\gamma}\subst{x}{V}] \subst{z}{\lambda x.M_1^{\gamma}},h)\opredobs{\yy}$ in $(k-1)$ steps.

By IH, $(K'[M_1^{\gamma}\subst{x}{V}] \subst{z}{\lambda x.M_2^{\gamma}},h) \opredobs{\yy}$.

Because $M_1\ciupre{\xx}{\yy} M_2$, this implies $(K'[M_2^{\gamma}\subst{x}{V}] \subst{z}{\lambda x.M_2^{\gamma}},h) \opredobs{\yy}$.

Since $(M\subst{z}{\lambda x.M_2^{\gamma}},h)\rarr 
(K'[M_2^{\gamma}\subst{x}{V}] \subst{z}{\lambda x.M_2^{\gamma}},h)$, we are done.
\end{itemize}
\end{itemize}

\end{proof}

\begin{lemma}[fix CIU]
$M_1\ciupre{\xx}{\yy} M_2$ implies $ \fix{f}{}{x}{M_1} \ciupre{\xx}{\yy}  \fix{f}{}{x}{M_2}$.
\end{lemma}
\begin{proof}
Take $K,\gamma,h$ such that $(K[( \fix{f}{}{x}{M_1})\substF{\gamma}],h)\opredobs{\yy}$.

We need to show $(K[ \fix{f}{}{x}{M_2})\substF{\gamma}],h)\opredobs{\yy}$.

Let us write $M_i^{\gamma}$ for $M_i\substF{\gamma}$, and $F_i$ for $\fix{f}{}{x}{M_i^{\gamma}}$.

We follow the same pattern as in the previous case and show
that $M\subst{z}{F_1}\opredobs{\yy}$ implies 
$M\subst{z}{F_2}\opredobs{\yy}$ for any $\seq{\Sigma; z}{M}$. 
The Lemma then follows by taking  $M=K[z]$. 

We use induction on the number of steps $k$ in $(M\subst{z}{F_1},h)\opredobs{\yy}$ for  all $M,h$.
Suppose $(M\subst{z}{F_1},h)\opredobs{\yy}$.

The following cases can be argued in the same way as above.
\begin{itemize}
\item ($M=K'[\err ()]$, $y=\mathit{err}$)
\item ($M=V$ or $M=z$, $y=\mathit{ter}$)
\item ($M=K'[N]$ and $(K'[N],h)\rarr (K'[N'],h')$)
\item ($M=K'[\throwto{V}{\cont{K''}}]$)
\end{itemize}
It remains to deal with 
\begin{itemize}
\item ($M=K'[z V]$) 

Then $(K'[ M_1^{\gamma} \subst{x}{V} \subst{f}{F_1}] \subst{z}{F_1},h)\opredobs{\yy}$ in $(k-1)$ steps.

Observe that $(K'[ M_1^{\gamma} \subst{x}{V} \subst{f}{F_1}] \subst{z}{F_1},h) = (K'[ M_1^{\gamma} \subst{x}{V} \subst{f}{z}] \subst{z}{F_1},h)$.

Hence, by IH, $(K'[M_1^{\gamma}\subst{x}{V}\subst{f}{z}] \subst{z}{F_2},h) \opredobs{\yy}$.

Because $M_1\ciupre{\xx}{\yy} M_2$, this implies 
$(K'[M_2^{\gamma}\subst{x}{V}\subst{f}{z}] \subst{z}{F_2},h) \opredobs{\yy}$.

Since $(M\subst{z}{F_2},h)\rarr  (K'[M_2^{\gamma}\subst{x}{V}\subst{f}{z}] \subst{z}{F_2},h)$, we are done.
\end{itemize}

\end{proof}

\begin{lemma}[$\callcc$ CIU]
$M_1\ciupre{\xx}{\yy} M_2$ implies $\callcc{(x.M_1)}\ciupre{\xx}{\yy} \callcc{(x. M_2)}$.
\end{lemma}
\begin{proof}
Let $K,\gamma,h$ be such that $(K[\callcc{(x.M_1)}],h)\opredobs{\yy}$.
Note that 
\[
(K[\callcc{(x.M_1)}\substF{\gamma}],h)\rarr (K[M_1\subst{x}{\cont{K}}]\substF{\gamma}],h)= 
(K[M_1\substF{\gamma\cdot[x\mapsto \cont{K}]}],h).
\]
Because of $M_1\ciupre{\xx}{\yy} M_2$, we get  $(K[M_2\substF{\gamma\cdot [x\mapsto \cont{K}]}],h)\opredobs{\yy}$
Consequently, \\ $(K[\callcc{(x.M_2)}\substF{\gamma}],h)\opredobs{\yy}$, 
because $(K[\callcc{(x.M_2)}\substF{\gamma}],h)\ired{} (K[M_2\substF{\gamma\cdot [x\mapsto \cont{K}]}],h)$.
\end{proof}

\begin{lemma}[Precongruence]\label{lem:pre}
Suppose $\xx\in\{\HOSC,\GOSC,\HOS,\GOS\}$, $\seq{\Gamma}{M_1,M_2:\sigma}$ are $\HOSC$-terms with an $\xx$ boundary,
and $C$ is an $\xx$-context such that $\seq{\Gamma'}{C[M_1],C[M_2]:\sigma'}$.
Then $\seq{\Gamma}{M_1\ciupre{\xx}{\yy} M_2:\sigma}$ implies
$\seq{\Gamma'}{C[M_1]\ciupre{\xx}{\yy} C[M_2]:\sigma'}$.
\end{lemma}
\begin{proof}
By induction on the structure of contexts using the preceding lemmas.
\end{proof}

\begin{corollary}[CIU result]
Suppose $\xx\in\{\HOSC,\GOSC,\HOS,\GOS\}$ and $\seq{\Gamma}{M_1,M_2:\sigma}$ are $\HOSC$-terms with an $\xx$ boundary.
$\seq{\Gamma}{M_1\ciupre{\xx}{\yy} M_2:\sigma}$ iff $\seq{\Gamma}{M_1\ctxpre{\xx}{\yy} M_2:\sigma}$.
\end{corollary}
\begin{proof}
The left-to-right implication follows from Lemma~\ref{lem:pre}.
The right-to-left implication holds, because testing with $h,K,\gamma$ is a special case of testing with $C$.
\end{proof}
The Corollary is the same as Lemma~\ref{lem:ciu}.


\section{Additional material for Section~\ref{sec:hoschosc} (HOSC[HOSC])}

\subsection{Extended Operational Semantics}

\begin{definition}
  Taking $M$ a term, $c$ a continuation name,
 $h$ a heap
 we write $\Sigma;\Gamma \vdash (M,c,h):\tau$
 if $\Sigma;\Gamma \vdash M:\tau$, $c:\tau$
 and $h:(\Sigma;\Gamma)$.
\end{definition}

\begin{lemma}
 \label{lm:decomp-ered}
 Taking $\Sigma;\Gamma \vdash (M,c,h):\tau$,
 then:
 \begin{itemize}
  \item either $(M,c,h)$ is reducible (for $\ered$);
  \item or $M$ is a a callback $K[f \ V]$ with $f \in \dom{\Gamma}$;
  \item or $M$ is a value $V$.
 \end{itemize}
\end{lemma}

\begin{lemma}
 \label{lm:decomp-fred}
 Taking $\Sigma;\Gamma \vdash (M,c,h):\tau$,
 and $\cseq{\Sigma;\Gamma}{K}{\tau}$
 and $\gamma$ an idempotent substitution s.t. $\vdash \gamma:\Gamma$,
 writing $\tilde{M}$ for $M\substF{\gamma}$ and $\tilde{h}$ for $h \substF{\gamma}$
 then $(K[\tilde{M}],\tilde{h}) \red (N,h')$ implies that
 \begin{itemize}
  \item either $(M,c,h) \ered (M',c',h'')$ and
  $N = K'[M'\substF{\gamma}]$ with $K' = \gamma(c')$, and $h''\substF{\gamma} = h'$;
  \item or $M$ is a callback $K'[f \ V]$ with $\gamma(f)$ a $\lambda$-abstraction $\lambda x.P$
  and $N = K[\tilde{K'}[P\subst{x}{\tilde{V}}]]$,
  with $\tilde{K'} = K'\substF{\gamma}$ and $\tilde{V} = V\substF{\gamma}$;
  \item or $M$ is a value and $K$ is an evaluation context larger than $\bullet$. 
 \end{itemize}
\end{lemma}

\begin{definition}
 Taking $M$ an extended term and $\kappa$ a substitution from continuation names to evaluation contexts that
 contains the continuation names appearing in the support of $M$,
 one write $M\substF{\kappa}$ for the term where all the occurrences of $\cont{K,c}$
 are substituted by $\cont{K'[K[\bullet]]}$, with $\kappa(c) = K'$.
 One extend this definition to heaps, writing $h\substF{\kappa}$ 
 for the heap $\{(\ell,v\substF{\kappa}) \sep (\ell,v) \in h\}$.
\end{definition}

\begin{theorem}
 \label{thm:ered-mplies-fred}
 Taking $M$ a term, $h$ a heap, $\kappa$ a substitution from continuation names to evaluation contexts that
 contains the continuation names appearing in the support of $M$ and $h$,
 and $c,c'$ two continuation names s.t. $\kappa(c) = K$ and $\kappa(c') = K'$,
 then for all $M',h'$,
 if $(M,c,h) \ered (M',c',h')$ then $(K[M\substF{\kappa}],h\substF{\kappa}) 
 \red (K'[M'\substF{\kappa}],h'\substF{\kappa})$.
\end{theorem}

\begin{proof}
 We reason by case analysis: 
 \begin{itemize}
 \item if $M = K_1[\callcc{(x.M_1)}]$, then one has:
  \begin{itemize}
   \item $(K_1[\callcc{(x.M_1)}],c,h) \ered (K_1[M_1\subst{x}{\cont{K_1,c}}],c,h)$;
   \item $(K[K_1[\callcc{(x.M_1)}]\substF{\kappa}],h\substF{\kappa}) 
    \red (K[K_1[M_1\subst{x}{\cont{K[K_1]}}]\substF{\kappa}],h\substF{\kappa})$
  \end{itemize}  
  and we conclude using the fact that $(M_1\subst{x}{\cont{K_1,c}})\substF{\kappa} = 
  (M_1\subst{x}{\cont{K[K_1]}})\substF{\kappa}$ since $\kappa(c) = K$.
 \item if $M = K_1[\throwto{V}{\cont{K_2,c'}}]$, then one has:
  \begin{itemize}
   \item $(K_1[\throwto{V}{\cont{K_2,c'}}],c,h) \ered (K_2[V],c',h)$;
   \item $(K[K_1[\throwto{V}{\cont{K_2,c'}}]\substF{\kappa}],h\substF{\kappa})
    \red (K'[K_2[V]\substF{\kappa}],h\substF{\kappa})$ since $\kappa(c') = K'$. 
  \end{itemize}  
 \item If there exists a (unique) reduction $(M,h) \red (M',h')$ then:
 \begin{itemize}
   \item $(M\substF{\kappa},c,h\substF{\kappa}) \ered (M'\substF{\kappa},c,h'\substF{\kappa})$
   \item $(K[M\substF{\kappa}],h\substF{\kappa}) 
    \red (K[M'\substF{\kappa}],h'\substF{\kappa})$.    
 \end{itemize}
%

 \end{itemize}
\end{proof}


\subsection{Proof of Lemma~\ref{lem:ctxerr}}
\begin{proof}
{
We reason by contraposition.
\begin{enumerate}
\item Suppose $\seq{\Gamma}{M_1\not\ctxter{\xx} M_2:\tau}$, i.e.
 $C[M_1]\opredter$ and $C[M_2]\not\opredter$ for some $\cseq{}{C}{\tau}$.
 
 Then we can construct $\cseq{\err}{C'}{\tau}$ such that $C'[M_1]\opredtererr$ and $C'[M_2]\not\opredtererr$ as follows:
 \[
 C'[\bullet] = (C_{;\err}[\bullet];\err),
 \] 
where $C_{;\err}$ refers to $C$ in which each occurrence of $\contt{\sigma}{(-)}$ is replaced with  $\contt{\sigma}{(-;\err)}$.
In this way, the construction transforms all opportunities for $\opredter$  into ones for $\opredtererr$.
 Note that, if  $M_1$ contained $\contt{\sigma}{K}$, it would not necessarily be the case that   $C'[M_1]\opredtererr$,
 because $M_1$ is not affected by the transformation.

\item Let $\xx\in \{\HOSC,\GOSC\}$. Suppose $\seq{\Gamma}{M_1\not\ctxpre{\xx}{err} M_2}$, i.e.
$C[M_1]\opredtererr$ and $C[M_2]\not\opredtererr$ for some $C$ such that $\cseq{\err}{C}{\tau}$.
Then we can construct $\cseq{}{C'}{\tau}$ such that
$C'[M_1]\opredter$ and $C'[M_2]\not\opredter$ as follows.
\[
C'[\bullet] = \callcc{(y. 
\,C_{;\Omega}[\bullet] \subst{\err}{(\fun{z}{\throwto{()}{y}})} ;\Omega)}
\] 
where $C_{;\Omega}$ is defined analogously to $C_{;\err}$.
Note that we add $;\Omega$, because $C[M_2]\not\opredtererr$ could be due to $\opredter$ (rather than divergence),
and we want to make sure that $C'[M_2]$ diverges, which will imply $C'[M_2]\not\opredter$.

Note that, because of the use of continuations, $C'$ is an $\xx$-context only for $\xx\in\{\HOSC,\GOSC\}$.

In this case, we also rely on $\contt{\sigma}(K)$-freeness (of $M_2$).
If $C[M_2]\not\opredtererr$ was due to $\opredter$ caused by $\contt{\sigma}{K}$ in $M_2$, 
then our $;\Omega$ transformation might not imply divergence for $C'[M_2]$.
\end{enumerate}
}
\end{proof}

\subsection{Name invariance}

We say that a permutation $p$ of $\Names$ is \boldemph{type-preserving} if it is also a permutation once
 restricted to each of $\CNames_\sigma$ and $\FNames_{\sigma\rarr\sigma'}$.
 Given $X\subseteq\Names$, we say that $p$ fixes $X$ if $p(x)=x$ for all $x\in X$.
 Type-preserving permutations can be applied to traces in the obvious way.
 In particular, if $t$ is $(N_O,N_P)$-trace then $p(t)$ is a $(p(N_O),p(N_P))$-trace.
 We write $\perm{t_1}{t_2}{X}$ if there exists a type-preserving permutation $p$ that fixes $X$ 
 such that $p(t_1)=t_2$.
 \begin{lemma}\label{lem:inv}
Suppose $\CC=\conf{\cdots,\phi,h}$ is a configuration and $p$ is a type-preserving permutation. 
If $t\in \TrR{\HOSC}{\CC}$ and $p$ fixes $\phi$ then $p(t)\in \TrR{\HOSC}{\CC}$.
 \end{lemma}

Due to the arbitrariness of name choice in transitions (i.e. freedom to choose fresh names),
 $\TrR{\HOSC}{\CC}$ is closed under renamings that preserve types and the names already present in $\CC$.

\subsection{Proof of Lemma~\ref{lem:cor}}

Delegated to Section~\ref{sec:proof}.

\subsection{Proof of Theorem~\ref{hoscsound}}

 \begin{proof}
  Suppose $\trsem{\HOSC}{\seq{\Gamma}{M_1}}\subseteq \trsem{\HOSC}{\seq{\Gamma}{M_2}}$.
  We handle $\seq{\Gamma}{M_1 \ciupre{\HOSC}{err} M_2}$, as it is slightly more involved. The reasoning for  $\ciupre{\HOSC}{ter}$ is symmetric.
   
Let $\Sigma, h,K,\gamma$ be such that $(K[M_1\substF{\gamma}],h)\opredtererr$.
Suppose $(\vec{A_i},\vec{\gamma_i})\in\AVal{\gamma}{\Gamma}$ and $c:\sigma'$ ($c\not\in\circ$).
By Lemma~\ref{lem:cor} (left-to-right), there exist $t,c'$ such that
$t\in \TrR{\HOSC}{\cconf{M_1}{\rho_{\vec{A_i}},c}}$ and $t^\bot\, \questP{\errn}{()}{c'} \in \TrR{\HOSC}{\cconf{h, K,\gamma}{\vec{\gamma_i}, c} }$.
By $\trsem{\HOSC}{\seq{\Gamma}{M_1}}\subseteq \trsem{\HOSC}{\seq{\Gamma}{M_2}}$, we have $t\in \TrR{\HOSC}{\cconf{M_2}{\rho_{\vec{A_i}},c}}$.
Because $t\in \TrR{\HOSC}{\cconf{M_2}{\rho_{\vec{A_i}},c}}$ and 
$t^\bot\, \questP{\errn}{()}{c'} \in \TrR{\HOSC}{\cconf{h, K,\gamma}{\vec{\gamma_i}, c} }$, by Lemma~\ref{lem:cor} (right-to-left) we
can conclude $(K[M_2\substF{\gamma}],h)\opredtererr$. Thus, $\seq{\Gamma}{M_1\ciupre{\HOSC}{err} M_2}$.
  \end{proof}

\subsection{Proof of Lemma~\ref{hoscdef}}

Recall that abstract values are tuples consisting of boolean and integer constants, as well as function names.
We can refer to them using projections of the form $\pi_{\vec{i}}$, where $\vec{i}\in\{1,2\}^+$, on the understanding
that $\pi_{i,\vec{i}} x= \pi_i (\pi_{\vec{i}} x)$.
\begin{itemize}
\item Suppose $\num{A}=\{(\vec{i},n)\,|\, \pi_{\vec{i}} A=n: \Bool,\Int \}$.
Then $\ass{x}{A}$ will act as shorthand for the following code
$\ifte{(\bigwedge_{(\vec{i},n)\in\num{A}} \pi_{\vec{i}}\, x = n  )}{()}{\Omega}$.
which checks if the boolean/integer arguments match those of $A$.

\item Another operation, written $A[\pi x/f]$, will substitute
for each $f\in \nu{(A)}$, the corresponding projection $\pi_{\vec{i}_f} x$ (i.e. one such that $\pi_{\vec{i}_f} A = f$).
\end{itemize}
This syntax will be used in all definability arguments.

\cutout{
\begin{theorem}
Suppose $\phi\uplus\{\errn\}\subseteq\FNames$,
and $t=o_1 p_1\cdots o_n p_n$ is a $(\phi\cup\{c\},\{\errn,\tern\})$-trace such that $\tern\not\in\nu(t)$.
There exists a passive configuration $\CC_O$ such that
 $\Treven{\CC_O}$ consists of all even-length prefixes of $t$
(along with their renamings via type-preserving permutations that fix $\phi\uplus\{\errn,\tern,c\}$.
Moreover, $\CC_O=\conf{\gamma_O, \{ c\mapsto \tern \}, \phi\uplus \{c,\errn,\tern\},h}$,
where $\dom{\gamma_O}=\phi\uplus\{c\}$, 
$\nu(\img{\gamma})=\emptyset$ and $\nu(\img{h}) =\emptyset$.
\end{theorem}}

\cutout{
 \begin{lemma}[Definability]\label{hoscdef}
Suppose $\phi\uplus\{\errn\}\subseteq\FNames$
and $t$ is an even-length $(\circ\uplus \{\errn\},\phi\uplus\{c\})$-trace starting with an O-action.
There exists a passive configuration $\CC$ such that the even-length traces
 $\TrR{\HOSC}{\CC}$ are exactly the even-length prefixes of $t$
(along with all renamings that preserve types and $\phi\uplus \{c\}\uplus \circ\uplus\{\errn\}$, cf. Remark~\ref{rem:invar}).
Moreover, $\CC=\conf{\gamma_\circ\cdot [c\mapsto K_\circ], \{ c\mapsto \tern_\Unit \}, \phi\uplus \{c\}\uplus \circ \uplus \{\errn\},h_\circ}$,
where 
$h,K,\gamma$ are built from $\HOSC$ syntax.
\end{lemma}}

\bigskip

Lemma~\ref{hoscdef} follows from the lemma given below for $i=0$. 
Consider $h'=h_0$, $K'=\gamma_0(c)$, $\gamma'=\gamma_0\setminus c$.
We have $\nu(\img{\gamma_0},\img{h_0})\subseteq  \circ\uplus \{\errn\}$.
As names $\circ_\sigma$ can only
occur inside terms of the form $\cont(K',\circ_\sigma)$,
we can conclude that $(h',K',\gamma')= (h_\circ,K_\circ,\gamma_\circ)$, 
where $h,K,\gamma$ are from $\HOSC$.

\begin{lemma}\label{hoscaux}
Suppose $\phi\uplus\{\errn\}\subseteq\FNames$, $c\in \CNames$ and
$t=o_1 p_1\cdots o_n p_n$ is a 
$(\circ\uplus \{\errn\},\phi\uplus\{c\})$-trace starting with an O-action.
Given $0\le i\le n$, let $t_i = o_{i+1}p_{i+1} \cdots o_n p_n$.
There exist  passive configurations $\CC_i$ such that  $\Treven{\CC_i}$ consists 
of even-length prefixes of $o_{i+1}p_{i+1} \cdots o_n p_n$
(along with their renamings via permutations on $\Names$ that fix $\phi_i$).
Moreover, $\CC_i=\conf{\gamma_i, \xi_i, \phi_i, h_i}$ ($0\le i\le n$), where
\begin{itemize}
\item $\dom{\gamma_i}$ consists of $\phi\cup\{c\}$ and all names introduced by P in $o_1 p_1\cdots o_i p_i$;
\item $\nu(\img{\gamma_i})=\emptyset$;
\item $\dom{\xi_i}$ consists of $c$ and all continuation names introduced by P in $o_1 p_1\cdots o_i p_i$;
\item for all $d\in\dom{\xi_i}$, $\xi_i(d)=\topp{o_1\cdots o_j}$ if $d$ was introduced in $p_j$ (we regard $c$ as being introduced
in $p_0$ and define $\topp{o_1\cdots o_0}=\circ_{\tau'}$);
\item ${\phi_i}$ consists of $\circ\uplus \{\errn\}\uplus \phi\uplus\{c\}$ and all names introduced in $o_1 p_1\cdots o_i p_i$;
\item $\dom{h_i}=\dom{h_0}$;
\item $\nu(\img{h_i})$ may only contain elements of $\circ\uplus\{\errn\}$ and names introduced by O in $o_1 p_1\cdots o_i p_i$.
\end{itemize}
\end{lemma}
\begin{proof}
The main idea is to use references in order to record all continuations and functions introduced by O, 
so that they can be accessed in terms at the time when they need to be used by P.
Other references will also be used to inject the right pieces of code into the LTS.

Below we explain how the content of $\CC_i$ is meant to evolve and what invariants will be maintained by the construction
for each kind of name in $t$. 


\begin{description}
\item[$\FNames$ from P] Suppose $\nfp$ is the number of function names in $\phi$ and those introduced by P in $t$.
We shall write $f_P^j$ ($0\le j <n_{FP}$) to refer to the $j$th such name, on the understanding names from $\phi$ are
introduced first and this is followed by names in $t$ in order of appearance (from left to right).

For each $f_P^j: \sigma_j\rarr\tau_j$, we will have a dedicated reference $\fpr{j}:\reftype{(\sigma_j\rarr\tau_j)}$ in all heaps.
The content of $h_i(\fpr{j})$ will be changing at each step of the construction and it will be used to arrange for
suitable behaviour following O-actions of the form $\questO{f_P^j}{A}{c}$. For example,
if the action is not meant to generate a response at a stage, we can use $\fpr{j} := \lambda x. (!\fpr{j}) x$ to cause divergence by creating a cycle in
the heap.

If $f_P^j$ was introduced in $p_i$ (we take $i=0$ for $f_P^j\in\phi$), then
$f_P^j$ will be present in all $\phi_{i'},\gamma_{i'}$ for $i'\ge i$.
We shall maintain the invariant $\gamma_{i'}(f_P^j)=\lambda x. (!\fpr{j}) x$ for all $i'\ge i$.

Note that this is consistent with $\nu(\img{\gamma_i})=\emptyset$.

\item[$\CNames$ from P] Suppose $\ncp$ is the number of continuation names  introduced by P in $t$ plus $1$, to take $c$ into account.
Similarly to the previous case, we write $c_P^j$ ($0\le j <c_{FP}$) to refer to the $j$th such name, on the understanding 
that $c_P^0=c$ and other names are enumerated in the same order as they appear in $t$ (from left to right).

For each $c_P^j:\sigma_j$, we will have a dedicated reference $\cpr{j} : \reftype{(\sigma_j\rarr\tau_j)}$, if $c_P^j$ was introduced in $p_{j'}$
and $\topo{o_1\cdots o_{j'}}:\tau_j$, in all heaps.

Its content will be changing at each step of the construction, in order to provide suitable 
reactions to O-actions of the form $\ansO{c_P^j}{A}$.

If  $c_P^j$ was introduced in $p_i$ (we take $i=0$ for $c_P^j=c_P$) then
$c_P^j$ will be present in all $\phi_{i'},\gamma_{i'}$ for $i'\ge i$.
We shall maintain the invariant $\gamma_{i'}(c_P^j)=(\lambda x. (!\cpr{j}) x)\bullet$ and $\xi_{i'}(c_P^j)=\topo{o_1\cdots o_{j'}}$, if $c_P^j$ was introduced in $p_{j'}$.

Note that this is consistent with $\nu(\img{\gamma_i})=\emptyset$.

\item[$\FNames$ from O] We use similar notation here and suppose $\nfo$ is the number of function names introduced by O.
As in previous cases, we use $f_O^j$ ($0\le j <\nfo$) to refer to such names.

For each $f_O^j:\sigma_j\rarr\tau_j$, we will have a corresponding reference $\foor{j}:\reftype{(\sigma_j\rarr\tau_j)}$ in all heaps,
which will be used to store the name as soon as it is played, i.e.
if $f_O^j$ is introduced in $o_i$ (for $\errn$ we take $i=0$), then $h_{i'}(\foor{j})=f_O^j$ for all $i'\ge i$.
Earlier we will use a divergent value, i.e.  $h_{i'}(\foor{j})=\lambda x. (!\foor{j}) x$ for $i'<i$. 

$f_O^j$ will be part of $\phi_{i'}$ for all $i'\ge i$.

Note that this is consistent with: $\nu(\img{h_i})$ may only contain elements of $\circ\uplus\{\errn\}$ and names introduced by O in $o_1 p_1\cdots o_i p_i$.

\item[$\CNames$ from O] Suppose $\nco$ is the number of continuation names introduced by O in $t$.
As before, we use $c_O^j$ ($0\le j <\nfo$) to refer to such names.

For each $c_O^j:\sigma_j$, we will have a corresponding reference  $\cor{j}:\reftype{(\conttype{\sigma_j})}$,
which will be used to store the name as soon as it is played, i.e.
if $c_O^j$ is introduced in $o_i$, then $h_{i'}(\cor{j})=\cont{(\bullet,c_O^j)}$ for all $i'\ge i$.
Earlier we will use a divergent value, i.e.  $h_{i'}(\cor{j})=\cont{((\lambda x.\Omega)\bullet,\circ_{\tau'})}$ for $i'<i$, where 
$\Omega$ is a divergent term.

$c_O^j$ will be part of $\phi_{i'}$ for all $i'\ge i$.

Note that this is consistent with: $\nu(\img{h_i})$ may only contain elements of $\circ\uplus\{\errn\}$ and names introduced by O in $o_1 p_1\cdots o_i p_i$.
\end{description}
Overall, for each $0\le i\le n$, we shall have
\[
\dom{h_i} = \{ \fpr{j} \,|\,  0\le j<\nfp\} \cup
 \{ \cpr{j} \,|\,  0\le j<\ncp\} \cup
 \{ \foor{j} \,|\,  0\le j<\nfo\} \cup
\{ \cor{j} \,|\,  0\le j<\nco\}.
\]

The above description specifies $\phi_i,\gamma_i,\xi_i$,$\dom{h_i}$ and 
 $h_i(\foor{j})$ ($0\le j < \nfo$), $h_i(\cor{j})$ ($0\le j<\nco$), for any $0\le i\le n$.
Hence, in the forthcoming argument we will focus on defining $h_i(\fpr{j})$ ($0\le j < \nfp$) and $h_i(\cpr{j})$ ($0\le i<\ncp$).
Because the values written to these references will only contain elements from $\circ\uplus\{\errn\}$, it will follow that
 $\nu(\img{h_i})$ may only contain elements of $\circ\uplus\{\errn\}$ and names introduced by O in $o_1 p_1\cdots o_i p_i$.

\bigskip

We proceed by reverse induction, starting from $i=n$.

\paragraph{$\mathbf{i=n}$}

To complete the definition of $\CC_n$, it suffices to specify  $h_n(\fpr{j})$ ($0\le j<\nfp$) and $h_n(\cpr{j})$ ($0\le j< \ncp$).
We set 
$h_n(\fpr{j}) = (\lambda x. (!\fpr{j})x)$ and
$h_n(\cpr{j}) = (\lambda x. (!\cpr{j})x)$, i.e. deferencing will cause divergence.
Consequently, because  $\gamma_n(f_P^j)=\lambda x.(!\fpr{j})x$ and   $\gamma_n(c_P^j)=\lambda x.(!\cpr{j})x$, 
any O action from $\CC_n$ will trigger divergence.
Thus, the only  even-length trace that can be generated is the empty one, and we have  $\Treven{\CC_n}=\{\epsilon\}$, as required.

\paragraph{$\mathbf{0 \le i < n}$}
Let $0 \le i < n$. Assume validity of the Lemma for $i+1$ and suppose $\CC_{i+1}=\conf{\gamma_{i+1},\xi_{i+1},\phi_{i+1},h_{i+1}}$.
By case analysis on $p_{i+1}$,  we first construct an active configuration $E_i=\conf{M',c',\gamma_i',\xi_i',\phi_i', h_{i+1}}$ 
such that $E_i\ired{p_{i+1}} \CC_{i+1}$.

Given an abstract value $A$, let $V_A=A[(\lambda x. (!\fpr{j})x)/f_P^{j}]$, i.e. the function names $f_P^j$ are replaced
with function values $(\lambda x. (!\fpr{j})x)$. Below we write $\phi_{i+1}\setminus X$, $\gamma_{i+1}\setminus X$
and $\xi_{i+1}\setminus X$ to stand for the removal of names in $X$ from the domain of the respective function, 
while preserving values for other elements. 
The table below shows the components of $E_i$ in each case.
\[\begin{array}{l|l|l|l|l|l}
p_{i+1} &M' &c' & \gamma_i' & \xi_i'  &\phi_i'\\
\hline\\[-3mm]
\ansP{c_O^{j'}}{A} &  V_A & c_O^{j'} &\gamma_{i+1}\setminus A & \xi_{i+1} &  \phi_{i+1}\setminus A\\[2mm]
\questP{f_O^{j'}}{A}{c_P^{j''}} & (\lambda x. (!\cpr{j''}) x) [f_O^{j'} V_A]  & \topp{o_1\cdots o_{i+1}}& \gamma_{i+1}\setminus A,c_P^{j''} & \xi_{i+1}\setminus c_P^{j''} & \phi_{i+1}\setminus A,c_P^{j''} \\
\end{array}\]
Note that, in each case, $E_i\ired{p_{i+1}} \CC_{i+1}$. In particular, our definition of $V_A$ (based on $\lambda x. (!\fpr{j})x$)
and the occurrence of  $\lambda x. (!\cpr{j'}) x$ in the second case guarantee that, after the step,
$\gamma_i'$ extends to $\gamma_{i+1}$ in accordance with our description of $\gamma_{i+1}$ at the beginning of the proof. 
Similarly, setting $c'$ to $ \topp{o_1\cdots o_{i+1}}$ in the second case means that $\xi_i'$ will evolve into $\xi_{i+1}$.

As a next step we define another active configuration $D_i  = \tuples{M'',\topp{o_1\cdots o_{i+1}},\gamma_i',\xi_i',\phi_i', h_{i+1}}$,
where $M''$ is specified by the table below, by case analysis on $p_{i+1}$.

Note that $D_i\ired{\tau} E_i$. 
\[\begin{array}{l|ll}
p_{i+1} &M''\\
\hline
\ansP{c_O^{j'}}{A} &  \throwto{V_A}{\cont(\bullet,\circ_\sigma)} &c_O^{j'}=\circ_\sigma\\
\ansP{c_O^{j'}}{A} &  \throwto{V_A}{! \cor{j'}} & c_O^{j'}\not\in\circ\\
\questP{f_O^{j'}}{A}{c_P^{j''}} & (\lambda x.(!\cpr{j''}) x)( \errn V_A) & f_O^{j'}=\errn\\
\questP{f_O^{j'}}{A}{c_P^{j''}} & (\lambda x.(!\cpr{j''}) x)( (!\foor{j'}) V_A)  &  f_O^{j'}\neq\errn
\end{array}\]

Finally, we are ready to  define $\CC_i=\tuples{\gamma_i, \xi_i,\phi_i,h_i}$ by case analysis on $o_{i+1}$.
Recall that $\phi_i$, $\gamma_i$, $\xi_i$, $\dom{h_i}$,  $h_i(\foor{j})$ ($0\le j<\nfo$), $h_i(\cor{j})$ ($0\le j < \nco$)
are covered by the invariants discussed at the beginning of the proof. Thus, it suffices to specify $h_i (\fpr{j})$ and $h_i(\cpr{j})$.

\begin{itemize}
\item Suppose $o_{i+1}= \ansO{c_P^j}{A}$. 
Since $o_{i+1}$ is the only O-move that should be responded to by P:
\begin{itemize}
\item we let $h_i(\fpr{j'})=\lambda x. (!\fpr{j'}) x$ for any $0\le j'<\nfp$, in order to create divergence after any $\questO{f_P^{j'}}{A}{c_O^{j''}}$;
\item we let $h_i(\cpr{j'}) = \lambda x. (!\cpr{j'}) x$ for any $0\le j'<\ncp$ such that $j'\neq j$, in order to create divergence
after $\ansO{c_P^{j'}}{A}$ with $j'\neq j$.
\end{itemize}
To allow a suitable response after $\ansO{c_P^j}{A}$, we set
\[
h_i(\cpr{j}) = \lambda x. \ass{x}{A};\, \mathit{savefun}(A);\,\mathit{setheap}(i+1); \, M''
\]
where the special code fragments are explained below. 
\begin{itemize}
 \item $\mathit{savefun}(A)$ is meant to save all functions from $A$ in the corresponding references. 
 Let $\fnam{A}=\{(\vec{i},w)\,|\, \pi_{\vec{i}}\, A=f_O^w \}$. Then
 $\mathit{savefun}(A)$ is the
sequence of assignments $\foor{w} := \pi_{\vec{i}}\, x$, for all $(\vec{i},w)\in\fnam{A}$.
\item $\mathit{setheap}(i+1)$ is the sequence of assignments 
$\fpr{j'} := h_{i+1}(\fpr{j'})$ ($0\le j'<\nfp$)
and $\cpr{j'}:= h_{i+1}(\cpr{h})$ ($0\le j'<\ncp$).
\end{itemize}
Suppose $c_P^j$ was introduced in $p_{j'}$ then we have $\topp{o_1\cdots o_{j'}} = \topp{o_1\cdots o_{i+1}}$,
i.e. types of the codomains of $!\cpr{j}$ and $!\cpr{j''}$ match, and indeed we can use $M''$ to define $h_i(\cpr{j})$
(note that throw is not causing typing problems).

Then we have $\CC_i\ired{o_{i+1}} C_i$, where $C_i=\tuples{(\lambda x. !\cpr{j} x)[A], \topp{o_1\cdots o_{j'}},\gamma_i', \xi_i', \phi_i',h_i)}$
and $C_i\ired{\tau^\ast} D_i = \tuples{M'',\topp{o_1\cdots o_{i+1}},\gamma_i',\xi_i', ,\phi_i',h_{i+1}}$.
Recall that we have already established $D_i \ired{\tau} E_i \ired{p_i} \CC_{i+1}$, so we are done.

\item Suppose $o_{i+1}= \questO{f_P^j}{A}{c_O^{j'}}$. 
Then we let $h_i(\cpr{j''})=\lambda x.(!\cpr{j''}) x$ ($0\le j''<\ncp$)  to create divergence after any $\ansO{c_P^{j''}}{A}$,
and $h_i(\fpr{j''}) = \lambda x. (!\fpr{j''}) x$ for any $0\le j'' <\nfp$ such that $j''\neq j$, to 
create divergence after any $\questO{f_P^{j''}}{A}{c_O^{j'''}}$ with $j''\neq j$.
Then, to arrange for the right reaction after $o_{i+1}$, we set
\[
h_i(\fpr{j}) = \lambda x. \ass{x}{A};\, \mathit{savefun}(A);\, \callcc(y. \ccr{j'}:= y;\, \mathit{setheap}(i+1);\, M'')
\]
where the special code fragments are specified above.
Note that, similarly, we have
$\CC_i\ired{o_{i+1}} C_i$, where $C_i=\tuples{(\lambda x. !\fpr{j} x)[A], c_O^{j'},\gamma_i, \xi_i, \phi_i,h_i)}$,
$C_i\ired{\tau^{\ast}} D_i = \tuples{M'',\topp{o_1\cdots o_{i+1}},\gamma_i', \xi_i', \phi_i',h_{i+1}}$ and $D_i\ired{\tau} E_i\ired{p_{i+1}} \CC_{i+1}$,
because in this case $c_O^{j'}=\topp{o_1\cdots o_{i+1}}$.
\end{itemize}
The invariance property follows from Remark~\ref{rem:invar}.
\cutout{
\paragraph{i=1}
We define $C_i'$ such that $C_i'\ired{p_i} \CC_i$, as above.
To complete the argument, we need to specify $M$ such that $\Gamma\vdash M:\tau$. 
$M$ is built as in the previous case except that we now need to create all the references that are used in the proof.
Thus, we set
\[
M\equiv \mathtt{let}\,\,\mathit{refs}\,\,\mathtt{in}\,\,{\mathit{assert}(A);\, \mathit{savefun}(A);\, \callcc(y.\, \ccr{0}:= y;\,\mathit{setheap}(i+1); M'')},
\]
where $\pi_k x$ in $\mathit{assert}(A)$ is now replaced with $x_k$ (from $\Gamma$),
and $\mathit{refs}$ stands for a sequence of reference creation statements of the shape
 $\oor{h} = \nuref{V_{\oor{h}}}$ or $\pr{h} = \nuref{V_{\pr{h}}}$ or $\kkr{j}=\nuref{V_{\kkr{j}}}$ or 
 $\ccr{j}=\nuref{\cont{V_{\ccr{j}}[]}}$, one for each reference $\oor{h}, \pr{h}, \oor{j}, \ccr{j}$,
 where the initialising terms $V_{\oor{h}}, V_{\pr{h}}, V_{\kkr{j}}, V_{\ccr{j}}$ are all of the shape $\lambda x^\sigma.\Omega_{\sigma'}$ for suitable $\sigma,\sigma'$ (to match the type).
 Then we have $\CC_0=\conf{\Gamma\vdash M:\tau}\ired{o_1} C_1=\tuples{M,c,\emptyset,\Delta_0', h_\mathit{init}}$
 and $C_1\red C_1'$.}
\end{proof}

\subsection{Proof of Theorem~\ref{hosccomplete}}

\begin{proof}
Suppose $\seq{\Gamma}{M_1 \ciuapperr^\HOSC M_2}$.
Let $\rho$ be a $\Gamma$-configuration, $A_i=\rho(x_i)$,
 $c:\sigma$ and $t\in \TrR{\HOSC}{\cconf{M_1}{\rho_{\vec{A_i}},c}}$.
 Then $t$ is a $(\nu(\rho)\uplus\{c\}, \emptyset)$-trace.
Let $t_1=t \substF{\errn'/\errn,\tern'/\tern}$, where $\errn',\tern'$ are fresh names of the same type as $\errn,\tern$ respectively
(this is done to ensure that $\errn,\tern$ do not occur in $t_1$).
By Lemma~\ref{lem:inv},  because $\perm{t_1}{t}{\nu(\rho)\uplus\{c\}}$, we also have $t_1\in \TrR{\HOSC}{\cconf{M_1}{\rho_{\vec{A_i}},c}}$.
Let $c':\Unit$ be fresh.
Then $t_2= t_1^{\bot}\,\questP{\errn}{()}{c'}$ is an $(\{\errn,\tern \},\nu(\rho)\uplus\{c\})$-trace.
By Lemma~\ref{hoscdef}, there exists a passive configuration 
$\CC_O=\conf{\gamma_O, \{ c\mapsto \tern \}, \nu(\rho)\uplus \{c,\errn,\tern\},h}$
such that $\TrReven{\HOSC}{\CC_O}$ consists of all (even-length prefixes of) traces $t'$ such that $\perm{t'}{t_2}{\nu(\rho)\uplus\{c,\errn,\tern\}}$.
Observe that $\CC_O = \cconf{h, K,\gamma}{\vec{\gamma_i}, c}$, where 
$K=(\gamma_O(c))\subst{\errn}{\err}$, $\gamma(x_i) = (A_i\substF{\gamma_O})\subst{\errn}{\err}$,
 and $\gamma_i=\rest{\gamma_O}{\nu(A_i)}$.
Hence, $t_1\in \TrR{\HOSC}{\cconf{M_1}{\rho_{\vec{A_i}},c}}$ and $t_1^{\bot}\,\questP{\errn}{()}{c'}\in  \cconf{h, K,\gamma}{\vec{\gamma_i}, c}$.
By Lemma~\ref{lem:cor} (right-to-left),  $(K[M_1\substF{\gamma}],h)\opredtererr$.
Because $\seq{\Gamma}{M_1 \ciuapperr^\HOSC M_2}$, $(K[M_2\substF{\gamma}],h)\opredtererr$ follows.
By Lemma~\ref{lem:cor} (left-to-right), there exist $t'',c''$ such that 
$t''\in \TrR{\HOSC}{\cconf{M_2}{\rho_{\vec{A_i}},c}}$ and $(t'')^{\bot}\,\questP{\errn}{()}{c''}\in  \cconf{h, K,\gamma}{\vec{\gamma_i}, c}$.
By the definition of $\CC_O$, we must have ${(t'')^{\bot}\,\questP{\errn}{()}{c''}}$ $\sim_{{\nu(\rho)\uplus\{c,\errn,\tern\}}}$ ${ t_1^{\bot}\,\questP{\errn}{()}{c'}}$, 
so $\perm{t''}{t_1}{\nu(\rho)\uplus\{c,\errn,\tern\}}$.
Because $t''\in \TrR{\HOSC}{\cconf{M_2}{\rho_{\vec{A_i}},c}}$,
we have $t_1\in \TrR{\HOSC}{\cconf{M_2}{\rho_{\vec{A_i}},c}}$ by Lemma~\ref{lem:inv}.
Since $\perm{t_1}{t}{\nu(\rho)\uplus\{c\}}$, it follows that $t\in \TrR{\HOSC}{\cconf{M_2}{\rho_{\vec{A_i}},c}}$, as required.
\end{proof}


\section{Composite Interaction (Proof of Lemma~\ref{lem:cor})}\label{sec:proof}

\begin{definition}
 A \emph{composite configuration} $\CCc$ is a tuple 
 $\conf{M,c,\gamma_P,\gamma_O,\xi,\phi,h_P,h_O}$ with
 $M$ a term, $c$ a continuation name, $\gamma_P,\gamma_O$ two environments, 
 $\phi$ a set of names and $h_P,h_O$ two heaps.
\end{definition}

\begin{definition}
 Taking a continuation function $\xi$, we define a relation $\orderCont{\xi}$ 
 between the continuation names as the graph of $\xi$, i.e. 
 $c \orderCont{\xi} c'$ when $\xi(c) = c'$.
\end{definition}

We write $\finalName$ for the final continuation name, used by Opponent to answer the resulting 
value of the whole interaction.

\begin{definition}
 A \emph{valid composite configuration} $\CCc$ is a tuple 
 $\conf{M,c,\gamma_P,\gamma_O,\xi,\phi,h_P,h_O}$ with:
 \begin{itemize}
  \item $\dom{\gamma_P} \cap \dom{\gamma_O} = \varnothing$ and $\finalName \notin \dom{\gamma_P}
   \cup \dom{\gamma_O}$;
  \item $\dom{\gamma_P} \cup \dom{\gamma_O}  \cup \{\finalName,\errn\} = \phi$;
  \item $\dom{\xi} = (\dom{\gamma_O} \cup \dom{\gamma(P)}) \cap \CNames$;
  \item for all $c \in \dom{\xi}$, if $c \in \dom{\gamma_X}$ then $\xi(c) \in 
  \dom{\gamma_{X^{\bot}}}$, for $X \in \{O,P\}$;
  \item the transitive closure of $\orderCont{\xi}$ is a strict partial order 
   which admit a unique maximal element equal to $\finalName$;
  \item $\gamma_P\cdot\gamma_O$ is well-typed;
  \item $c \in \phi$ with $c:\sigma$ $\vdash M: \sigma$;
  \item $\dom{h_P} \cap \dom{h_O} = \varnothing$.
 \end{itemize}
\end{definition}

The composite LTS, defined on such composite configurations, is given in Figure~\ref{fig:comp-lts-hosc}. 
Up to choice of name, it is deterministic.

\begin{figure}[t]
\[\begin{array}{l|lll}
   (P\tau) & \conf{M,c,\gamma_P,\gamma_O,\xi,\phi,h_P,h_O} & \ired{\tau} & 
    \conf{N,c',\gamma_P,\gamma_O,\xi,\phi,h'_P,h_O} \\
    & \multicolumn{3}{l}{\text{ when } c \in \dom{\gamma_O} \text{ and } 
     (M,c,h_P) \ered (N,c',h'_P)}\\
   (PA) & \conf{V,c,\gamma_P,\gamma_O,\xi,\phi,h_P,h_O} & \ired{\ansP{c}{A}} & 
    \conf{K[A],\xi(c),\gamma_P \cdot \gamma',\gamma_O,\xi,\phi\uplus\dom{\gamma'},h_P,h_O} \\
    & \multicolumn{3}{l}{\text{ when } c:\sigma, 
     \gamma_O(c) = K, \text{ and } (A,\gamma') \in \AVal{V}{\sigma}}\\   
   (PQ) & \conf{K[fV],c,\gamma_P,\gamma_O,\xi,\phi,h_P,h_O} & \ired{\questP{f}{A}{c'}} & 
    \langle V' A,c',\gamma_P \cdot\gamma'\cdot [c' \mapsto K],\gamma_O,\xi \cdot [c' \mapsto c],
    \\ & & & \qquad \phi\uplus\dom{\gamma'}\uplus \{c'\},h_P,h_O\rangle \\
    & \multicolumn{3}{l}{\text{ when } f:\sigma \rarr \sigma', c':\sigma',
      \gamma_O(f) = V' \text{ and }(A,\gamma') \in \AVal{V}{\sigma}}\\   
   (O\tau) & \conf{M,c,\gamma_P,\gamma_O,\xi,\phi,h_P,h_O} & \ired{\tau} & 
    \conf{N,c',\gamma_P,\gamma_O,\xi,\phi,h_P,h'_O} \\
    & \multicolumn{3}{l}{\text{ when } c \in \dom{\gamma_P} \text{ and } 
     (M,c,h_O) \ered (N,c',h'_O)}\\
   (OA) & \conf{V,c,\gamma_P,\gamma_O,\xi,\phi,h_P,h_O} & \ired{\ansO{c}{A}} & 
    \conf{K[A],\xi(c),\gamma_P,\gamma_O \cdot \gamma',\xi,\phi\uplus\dom{\gamma'},h_P,h_O} \\
    & \multicolumn{3}{l}{\text{ when } c:\sigma, 
    \gamma_P(c) = K, \text{ and } (A,\gamma') \in \AVal{V}{\sigma}}\\   
   (OQ) & \conf{K[fV],c,\gamma_P,\gamma_O,\xi,\phi,h_P,h_O} & \ired{\questO{f}{A}{c'}} & 
    \langle V' A,\gamma_P ,\gamma_O\cdot\gamma'\cdot [c' \mapsto K],\xi \cdot [c' \mapsto c],
     \\ & & & \qquad \phi\uplus\dom{\gamma'}\uplus \{c'\},h_P,h_O\rangle \\
    & \multicolumn{3}{l}{\text{ when } f:\sigma \rarr \sigma', c':\sigma',
      \gamma_P(f) = V' \text{ and }(A,\gamma') \in \AVal{V}{\sigma}}\\
  \end{array}\]
 \caption{Composite LTS for HOSC[HOSC]}
 \label{fig:comp-lts-hosc}
\end{figure}
 
\begin{definition}
  \label{def:hosc:compat-conf}
  Two valid $\HOSC$-configurations $\CC_P,\CC_O$ are said to be \emph{compatible} 
  if one of the two is active and the other one is passive, 
  and, without loss of generality, supposing that $\CC_P$ is the active configuration 
  $\conf{M,c,\gamma_P,\xi_P,\phi_P,h_P}$ and $\CC_O$ the passive configuration 
  $\conf{\gamma_P,\xi_O,\phi_O,h_O}$, then $\phi_O = \phi_P \uplus \{\finalName,\errn\}$ and
  the composite configuration $\conf{M,c,\gamma_P,\gamma_O,\xi_P \cdot \xi_O,\phi_O,h_P, h_O}$,
  written $\mergeConf{\CC_P}{\CC_O}$, is valid.
\end{definition}

\begin{lemma}
 Taking $\CCc$ a valid composite configuration and $\CCc'$ a composite configuration s.t.
 $\CCc \iRed{\act} \CCc'$, then $\CCc'$ is valid.
\end{lemma}

%
%
%
%

\begin{lemma}
 \label{lm:hosc:iredcomp-to-ired-1step}
 Taking $\CC_P,\CC_O$ two compatible configurations, for all composite configuration 
 $\CCc'$, if $(\mergeConf{\CC_P}{\CC_O}) \iRed{\act} \CC'$
 then there exists two compatible configurations $\CC'_P,\CC'_O$ s.t.:
 \begin{itemize}
  \item $\CCc' = \mergeConf{\CC'_P}{\CC'_O}$;
  \item $\CC_P \iRed{\act} \CC'_P$ and $\CC_O \iRed{\act^{\bot}} \CC'_O$.
 \end{itemize}
\end{lemma}

\begin{proof}
 Without loss of generality, we suppose that $\CC_P$ is the active configuration and 
 $\CC_O$ the passive one.
 So we write $\CC_P$ as $\conf{M,c,\gamma_P,\phi,h_P}$ and $\CC_O$ as $\conf{\gamma_O,\phi,h_O}$.
 \begin{itemize} 
  \item If $\act$ is a Player Answer $\ansP{c'}{A}$, then there exists $V,h'_P$ s.t.
  \[(\mergeConf{\CC_P}{\CC_O}) \ired{\tau} \conf{V,c',\gamma_P,\gamma_O,\xi,\phi,h'_P, h_O}\] 
  so that $(M,c,h_P) \ered (V,c',h'_P)$. 
  Then there exists $K,c''$ s.t.
  $\gamma_O(c') = K, \xi(c') = c''$ and there exists $\sigma,\gamma'$ s.t. 
  $c' :\sigma$ 
  and $(A,\gamma') \in \AVal{V}{\sigma}$, so that
  $\CCc' = \conf{K[A],c'',\gamma_P\cdot\gamma',\gamma_O,\phi \uplus \dom{\gamma'},h'_P, h_O}$.
  
  We then define $\CC'_P$ as $\conf{\gamma_P\cdot\gamma',\phi \uplus \dom{\gamma'},h'_P}$
  and $\CC'_O$ as $\conf{K[A],c'',\gamma_O,\phi \uplus \dom{\gamma'},h_O}$.
  One easily check that: 
  \begin{itemize}
   \item $\CC'_P,\CC'_O$ are two compatible configurations;
   \item $\CCc' = \mergeConf{\CC'_P}{\CC'_O}$;
   \item $\CC_P \ired{\tau} \conf{V,c',\gamma_P,\phi,h'_P} 
     \ired{\ansP{c'}{A}} \CC'_P$;
   \item $\CC_O \ired{\ansO{c'}{A}} \CC'_O$.
  \end{itemize} 
  
  \item If $\act$ is a Player Question $\questP{f}{A}{c'}$, then there exists $K,V,c'',h'_P$ s.t. 
  \[(\mergeConf{\CC_P}{\CC_O}) \ired{\tau} \conf{K[f \ V],c'',\gamma_P,\gamma_O,\xi,\phi,h'_P, h_O}\]
  so that $(M,c,h_P) \ered (K[f \ V],c'',h'_P)$. 
  Then there exists $V'$ s.t $\gamma_O(f) = V'$, 
  and there exists $\sigma,\sigma',\gamma'$ s.t. $f:\sigma \rarr \sigma'$, 
  and $(A,\gamma') \in \AVal{V}{\sigma}$, so that
  $\CCc' = \conf{V'A,c',\gamma_P\cdot \gamma'\cdot[c' \mapsto K],\gamma_O,
  \xi \cdot [c' \mapsto c''],\phi \uplus \dom{\gamma'} \cdot \{c'\},h'_P, h_O}$.
   
  We then define $\CC'_P$ as $\conf{\gamma_P\cdot\gamma'\cdot[c' \mapsto K],
   \xi \cdot [c' \mapsto c''],\phi \uplus \dom{\gamma'} \uplus \{c'\},h'_P}$
  and $\CC'_O$ as $\conf{V' A,c',\gamma_O,\phi \uplus \dom{\gamma'} \uplus \{c'\},h_O}$.
  One easily check that:
  \begin{itemize}
   \item $\CC'_P,\CC'_O$ are two compatible configurations;
   \item $\CCc' = \mergeConf{\CC'_P}{\CC'_O}$;
   \item $\CC_P \ired{\tau} \conf{K[f \ V],c'',\gamma_P,\phi,h'_P} 
     \ired{\questP{f}{A}{c'}} \CC'_P$;
   \item $\CC_O \ired{\questP{f}{A}{c'}} \CC'_O$.
  \end{itemize} 
 \end{itemize}
\end{proof}

\begin{lemma}
\label{lm:hosc:ired-to-iredcomp-1step}
 Taking $\CC_P,\CC_O$ two compatible configurations, if
 \begin{itemize}
  \item $\CC_P \iRed{\act} \CC'_P$;
  \item $\CC_O \iRed{\act^{\bot}} \CC'_O$;
 \end{itemize}
 then $\CC'_P,\CC'_O$ are two compatible configurations
 and $(\mergeConf{\CC_P}{\CC_O}) \iRed{\act} (\mergeConf{\CC'_P}{\CC'_O})$.
\end{lemma}

\begin{proof}
 Without loss of generality, we suppose that $\CC_P$ is the active configuration and 
 $\CC_O$ the passive one. So we write $\CC_P$ as $\conf{M,c,\gamma_P,\phi,h_P}$ and 
 $\CC_O$ as $\conf{\gamma_O,\phi,h_O}$.
 \begin{itemize}
  \item If $\act$ is a Player Answer $\ansP{c'}{A}$, then there exists $V,h'_P$ s.t. 
  $\CC_P \ired{\tau} \conf{V,c',\gamma_P,\phi,h'_P}$ so that
  $(M,c,h_P) \ered (V,c',h'_P)$. Then:
  \begin{itemize}
   \item there exists $\sigma$ s.t. $c' :\sigma$,
   and $\gamma',$ s.t. $(A,\gamma') \in \AVal{V}{\sigma}$ so that
   $\CC'_P = \conf{\gamma_P\cdot \gamma',\phi \uplus \dom{\gamma'},h'_P}$;
   \item there exists $K,c''$ s.t. $\gamma_O(c') = K, \xi(c') = c''$ 
   and $\CC'_O = \conf{K[A],c'',\gamma_O,\phi \uplus \dom{\gamma'},h_O}$.
  \end{itemize}

 Then one easily checks that $\CC'_P,\CC'_O$ are two compatible configurations, and:
 \[\begin{array}{lll}
  (\mergeConf{\CC_P}{\CC_O}) 
   & \ired{\tau} & 
    \conf{V,c',\gamma_P,\gamma_O,\xi,\phi,h'_P,h_O} \\
   & \ired{\ansP{c'}{A}} & 
    \conf{K[A],c'',\gamma_P\cdot\gamma',\gamma_O,\phi \uplus \dom{\gamma'},h'_P, h_O}
  \end{array}\]
 so that $\conf{K[A],c'',\gamma_P\cdot\gamma',\gamma_O,\phi \uplus \dom{\gamma'},
  h'_P, h_O} = \mergeConf{\CC'_P}{\CC'_O}$.
  \item If $\act$ is a Player Question $\questP{f}{A}{c'}$, there exists $K,V,c'',h'_P$ 
  s.t. $\CC_P \ired{\tau} \conf{K[f \ V],c'',\gamma_P,\phi,h'_P}$
  so that $(M,c,h_P) \ered (K[f \ V],c'',h'_P)$. Then:
  \begin{itemize}
   \item there exists $\sigma,\sigma'$ s.t. $f:\sigma \rarr \sigma'$,
   and $V',\gamma'$, s.t. $\gamma_O(f) = V'$ 
   and $(A,\gamma') \in \AVal{V}{\sigma}$ so that
   $\CC'_P = \conf{\gamma_P\cdot \gamma'\cdot[c' \mapsto K], \xi \cdot [c' \mapsto c''],
   \phi \uplus \dom{\gamma'}\uplus\{c'\},h'_P}$;
   \item there exists $V'$ s.t. $\gamma_O(f) = V'$ and 
   $\CC'_O = \conf{V'A,c',\gamma_O,\phi \uplus \dom{\gamma'} \uplus \{c'\},h_O}$.
  \end{itemize}
 Then one easily checks that $\CC'_P,\CC'_O$ are two compatible configurations,
 and:
 \[\begin{array}{lll}
  (\mergeConf{\CC_P}{\CC_O}) & \ired{\tau} & \conf{K[f\ V],c'',\gamma_P,\gamma_O,\xi,
  \phi,h'_P, h_O} \\
  & \ired{\questP{f}{A}{c'}} & \conf{V'A,c',\gamma_P\cdot \gamma'\cdot[c' \mapsto K],
   \gamma_O,\xi\cdot[c' \mapsto c''],\phi \uplus \dom{\gamma'} \uplus \{c'\},h'_P, h_O}
  \end{array}
 \]
 so that $\conf{K[A],c'',\gamma_P\cdot\gamma'\cdot[c' \mapsto K],\gamma_O,
 \xi \cdot[c' \mapsto c''],\phi \uplus \dom{\gamma'} \uplus \{c'\},h'_P\cdot h_O} = 
 \mergeConf{\CC'_P}{\CC'_O}$.
 \end{itemize}  
\end{proof}

\begin{definition}
 A composite configuration $\CCc$ \emph{terminates} following a trace $\tr$, 
 written $\CCc \ltster{\tr}$, when there exists a \emph{final} composite 
 configuration $\CCc_f = \conf{\unit,\finalName,\gamma_P,\gamma_O,\xi,\phi,h_P,h_O}$ s.t.
 $\CCc \iRed{\tr} \CCc_f$. We often omit the trace $\tr$ and simply write 
 $\CCc \ltster{}$.
\end{definition}

\begin{definition}
 A composite configuration $\CCc$ \emph{errors} following a trace $\tr$, 
 written $\CCc \ltstererr{\tr}$, when there exists a composite 
 configuration $\CCc_f = \conf{K[\err\unit],c,\gamma_P,\gamma_O,\xi,\phi,h_P,h_O}$ s.t.
 $\CCc \iRed{\tr} \CCc_f$. We often omit the trace $\tr$ and simply write 
 $\CCc \ltstererr{}$.
\end{definition}

\begin{lemma}
 \label{lm:hosc:active-passive-length}
 Taking $\CC_P,\CC_O$ two compatible configurations
 if $(\mergeConf{\CC_P}{\CC_O}) \ltster{\tr}$ then:
 \begin{itemize}
  \item if $\CC_P$ is active and $\CC_O$ passive, $\tr$ is even-length;
  \item if $\CC_P$ is passive and $\CC_O$ active, $\tr$ is odd-length.
 \end{itemize}
\end{lemma}

\begin{proof}
 By induction on the length of $\tr$:
 \begin{itemize}
  \item If $\tr = \emptytrace$, then $\mergeConf{\CC_P}{\CC_O}$
  can be written as $\conf{\unit,\finalName,\gamma_P,\gamma_O,\xi,\phi,h_P,h_O}$.
  Writing $\phi_P$ for the name environment component of $\CC_P$,
  and $\phi_O$ for the one of $\CC_O$, then 
  $\phi = \phi_O = \phi_P \uplus \{\finalName\}$.
  So necessarily is the $\CC_O$ active one.
  \item If $\tr = \act \cdot \tr'$, then we conclude using 
  Lemma~\ref{lm:hosc:iredcomp-to-ired-1step} and the induction hypothesis.
 \end{itemize}
\end{proof}

\begin{definition}
 Taking $\CC_P,\CC_O$ two compatible configurations, one write
 $(\CC_P|\CC_O) \semobs{\yy}{\tr}$, with $\yy \in \{\ter,\err\}$, when $\tr \in \Tr{\CC_P}$
 and
 \begin{itemize}
  \item if $\yy = \ter$ then $\tr^{\bot} \cdot \ansP{\tern}{\unit} \in \Tr{\CC_O}$;
  \item if $\yy = \err$ then $\tr^{\bot} \cdot \questP{\errn}{\unit}{c} \in \Tr{\CC_O}$
  for some $c \in \CNames$;
 \end{itemize}

\end{definition}

\begin{lemma}
 \label{lm:hosc:ired-iff-iredcomp}
 Taking $\CC_P,\CC_O$ two compatible configurations and $\tr$ a trace, then
 $(\CC_P|\CC_O) \semobs{\yy}{\tr}$ iff $(\mergeConf{\CC_P}{\CC_O}) \ltsobs{\yy}{\tr}$,
 with $\yy \in \{\ter,\err\}$.
\end{lemma}

\begin{proof}
 We first prove that if $(\CC_P|\CC_O) \semobs{\yy}{\tr}$ 
 then $(\mergeConf{\CC_P}{\CC_O}) \ltsobs{\yy}{\tr}$
 by induction on the length of $\tr$:
 \begin{itemize}
 \item if $\tr$ is empty and $\yy = \ter$, then $\ansP{\finalName}{\unit} \in \Tr{\CC_O}$,
 so there exists $\gamma_O,\phi,h_O$ s.t.
 $\CC_O \ired{\tau} \conf{\unit,\finalName,\gamma_O,\phi,h_O}$.
 Since $\CC_O$ is an active configuration, $\CC_P$ must be a passive configuration,
 that we write as $\conf{\gamma_P,\phi,h_P}$.
 Then $\mergeConf{\CC_P}{\CC_O} = \conf{\unit,\finalName,\gamma_P,\gamma_O,\xi,\phi,
 h_P,h_O}$,
 so that indeed $(\mergeConf{\CC_P}{\CC_O}) \ltster{\emptytrace}$.
 \item if $\tr$ is empty and $\yy = \err$, then $\questP{\err}{\unit}{c} \in \Tr{\CC_O}$,
 so there exists $\gamma_O,\phi,h_O$ s.t.
 $\CC_O \ired{\tau} \conf{K[\err\unit],c,\gamma_O,\phi,h_O}$.
 Since $\CC_O$ is an active configuration, $\CC_P$ must be a passive configuration,
 that we write as $\conf{\gamma_P,\phi,h_P}$.
 Then $\mergeConf{\CC_P}{\CC_O} = \conf{K[\err\unit],c,\gamma_P,\gamma_O,\xi,\phi,
 h_P,h_O}$,
 so that indeed $(\mergeConf{\CC_P}{\CC_O}) \ltster{\emptytrace}$.
 \item if $\tr = \act \cdot \tr'$, then
 there exists two configurations $\CC'_P,\CC'_O$ s.t.:
 \begin{itemize}
 \item $\CC_P \iRed{\act} \CC'_P$;
 \item $\CC_O \iRed{\act^{\bot}} \CC'_O$;
 \item $(\CC'_P|\CC'_O) \semter{\tr'}$.
 \end{itemize}
 From Lemma~\ref{lm:hosc:ired-to-iredcomp-1step}, we get that 
 $\CC'_P,\CC'_O$ are two compatible configurations
 and $(\mergeConf{\CC_P}{\CC_O}) \iRed{\act} (\mergeConf{\CC'_P}{\CC'_O})$.
 Using the induction hypothesis we get that
 $(\mergeConf{\CC'_P}{\CC'_O}) \ltsobs{\yy}{\tr'}$.  
 So $(\mergeConf{\CC_P}{\CC_O}) \ltsobs{\yy}{\tr}$.
 \end{itemize}
 
 \noindent
 We now prove that if $(\mergeConf{\CC_P}{\CC_O}) \ltsobs{\yy}{\tr}$ then 
 $(\CC_P|\CC_O) \semobs{\yy}{\tr}$, by induction on the length of $\tr$:
 \begin{itemize}
  \item if $\tr$ is empty and $\yy = \ter$, then
  $(\mergeConf{\CC_P}{\CC_O}) \ired{\tau} \conf{\unit,\finalName,\gamma_P,\gamma_O,\xi,
  \phi,h_P,h_O}$.
  So $\CC_O \ired{\tau} \conf{\unit,\finalName,\gamma_O,\phi,h_O}$ 
  and $\CC_P = \conf{\gamma_P,\phi,h_P}$.
  Thus $\CC_O \iRed{\ansP{\finalName}{\unit}} \conf{\gamma_O,\phi,h_O}$,
  so $(\CC_P|\CC_O) \semter{\emptytrace}$.
  \item if $\tr$ is empty and $\yy = \err$, then
  $(\mergeConf{\CC_P}{\CC_O}) \ired{\tau} \conf{K[\err\unit],c,\gamma_P,\gamma_O,\xi,
  \phi,h_P,h_O}$.
  So $\CC_O \ired{\tau} \conf{K[\err\unit],c,\gamma_O,\phi,h_O}$ 
  and $\CC_P = \conf{\gamma_P,\phi,h_P}$.
  Thus $\CC_O \iRed{\questP{\err}{\unit}{c}} \conf{\gamma_O,\phi,h_O}$,
  so $(\CC_P|\CC_O) \semter{\emptytrace}$.
  \item if $\tr = \act \cdot \tr'$, then there exists a composite configuration $\CCc'$ s.t.
  $(\mergeConf{\CC_P}{\CC_O}) \iRed{\act} \CCc'$ and $\CCc' \ltster{\tr'}$.
  From Lemma~\ref{lm:hosc:iredcomp-to-ired-1step}, 
  we get the existence of two compatible configurations $\CC'_P,\CC'_O$
  s.t.:
  \begin{itemize}
   \item $\CCc' = \mergeConf{\CC'_P}{\CC'_O}$;
   \item $\CC_P \iRed{\act} \CC'_P$;
   \item $\CC_O \iRed{\act^{\bot}} \CC'_O$.
  \end{itemize}
  From $(\mergeConf{\CC'_P}{\CC'_O}) \ltsobs{\yy}{\tr'}$,
  we get from the induction hypothesis that
  $(\CC'_P|\CC'_O) \semobs{\yy}{\tr'}$.
  So $(\CC_P|\CC_O) \semobs{\yy}{\tr}$.
 \end{itemize} 
\end{proof}

\begin{definition}
 Taking $\gamma,\xi$ a valid environment and $c,c'$ 
 two continuation names s.t. $c \orderCont{\xi}^{*} c'$,
 we define the evaluation context $K_{c,c'}$ as:
 \begin{itemize}
  \item $K_{c,c} \defeq \bullet$
  \item $K_{c,c'} \defeq K_{c'',c'}[K]$, 
  when $\gamma(c) = K$ and $\xi(c) = c''$.
 \end{itemize}
 We write $K_{c}$ for $K_{c,\finalName}$.
%
\end{definition}

\begin{definition}
 To an environment $\gamma$, we associate an idempotent substitution $\delta$ 
 defined as the relation:
 \begin{itemize}
 \item $\delta^{0} \defeq \{(f,V) \sep f \in \dom{\gamma} \land \gamma(f) = V\} \cup 
 \{(c,K) \sep c \in \dom{\gamma} \land \gamma(c) = K\}$
 \item $\delta^{i+1} \defeq \{(f,V\substF{\delta^{i}}) \sep (f,V) \in \delta^{i}\} \cup 
   \{(c,K\substF{\delta^{i}}) \sep (c,K) \in \delta^{i}\}$
 where we write $V\substF{\delta^{i}}$ for the action of 
 the substitution $\delta^{i}$ to $V$
\end{itemize}
then there exists $n \in \N$ s.t. $\delta^{n+1} = \delta^{n}$, 
and $\delta$ is then defined as
$\delta^{n}$.
\end{definition}

One need this iterative construction to get the idempotency result, 
that corresponds to the fact that the support of the values and evaluation contexts 
in the codomain of $\delta$ are empty (i.e. they do not have continuation or functional names anymore). 
This is possible because there is no cycles between names.


\begin{lemma}
\label{lm:quests-ired}
 Taking $\CCc = \conf{K[f \ V],c,\gamma_P,\gamma_O,\xi,\phi,h_P,h_O}$ 
 a valid composite configuration that is going to perform 
 a question, with $f \in \dom{\gamma}$, where $\gamma = \gamma_P \cdot \gamma_O$,
 there exists a functional name $g$, an abstract value $A$,
 a composite configuration $\CCc'$ and a trace $\tr$ formed by questions s.t.:
 \begin{itemize}
  \item $\gamma(g)$ is a $\lambda$-abstraction $\lambda x.M$;
  \item $\delta(f) = \delta(g)$, writing $\delta$ for the idempotent substitution associated to $\gamma$;
  \item $\CCc \ired{\tr} \CCc'$;
  \item $\CCc'$ can be written as 
    $\conf{g \ A,c',\gamma_P \cdot \gamma'_P,\gamma_O \cdot \gamma'_O,\phi \uplus \dom{\gamma'_P},h_P,h_O}$;
  \item $A\substF{\delta'} = V$, with $\delta'$ 
  the idempotent substitution associated to $\gamma'_P \cdot \gamma'_O$;
  \item $K_{c',c}^{\gamma} = \bullet$.
 \end{itemize}
\end{lemma}

\begin{lemma}
\label{lm:anss-ired}
 Let $\CCc = \conf{V,c,\gamma_P,\gamma_O,\xi,\phi,h_P,h_O}$ be
 a valid composite configuration that is going to perform 
 an answer.
 Suppose that there exists $c'$ s.t. $c \orderCont{\gamma}^{*} c'$
 and $K_{c,c'}^{\gamma} = \bullet$.
 Then there exists a composite configuration 
 $\CCc' = \conf{A,c',\gamma_P \cdot \gamma'_P,\gamma_O \cdot \gamma'_O,\phi \uplus \dom{\gamma'_P},h_P,h_O}$
 and a trace $\tr$ formed only by answers s.t. $\CCc \ired{\tr} \CCc'$
 and $A\substF{\delta'} = V$,
 with $\delta'$ the idempotent substitution associated to $\gamma'_P \cdot \gamma'_O$.
\end{lemma}


\begin{definition}
 One define the configuration transformation $\theta$ from valid composite configurations to pair 
 formed by a term and a heap, defined as
 \[\theta: \conf{M,c,\gamma_P,\gamma_O,\xi,\phi,h_P,h_O} \mapsto ((K_{c}^{\gamma}[M])\substF{\delta},(h_P \cdot h_O)\substF{\delta})\]
 writing $\gamma$ for $\gamma_P \cdot\gamma_O$ 
 and $\delta$ for the idempotent substitution associated to $\gamma$.
\end{definition}

\begin{lemma}
 \label{lm:ired-stable-red}
 Taking $\CCc,\CCc'$ two valid composite configuration and $\act$ an action 
 (different of $\tau$) s.t. $\CCc \ired{\act} \CCc'$ then $\theta(\CCc) = \theta(\CCc')$.
\end{lemma}

\begin{proof}
Let us write $\CCc$ as $\conf{M,c,\gamma_P,\gamma_O,\xi,\phi,h_P,h_O}$.
Without loss of generality, we suppose the composite configuration 
$\CCc$ to be $P$-active, i.e. $c \in \dom{\gamma_O}$

We reason by case analysis over $\alpha$:
 \begin{itemize}
  \item If $\alpha = \ansP{c}{A}$, so that $M$ is a value $V$. Then we have:
   \begin{itemize}
    \item $\gamma_O(c) = K$ and $c:\tau$ for some context $K$ and type $\tau$;
    \item $\gamma'_O = \gamma_O$, $\gamma'_P = \gamma_P \cdot \gamma_A$ and 
    $\phi' = \phi \uplus \dom{\gamma_A}$;
    with $(A,\gamma_A) \in \AVal{V}{\tau}$;
    \item $h'_P = h_P$ and $h'_O = h_O$;
    \item $M' = K[A]$.
   \end{itemize} 
    We conclude using these and the fact that:
   \begin{itemize} 
    \item $K_{c}^{\gamma} = K_{c'}^{\gamma}[K]$, where $c' = \xi(c)$;
    \item $A\substF{\gamma_A} = V$;
   \end{itemize}
  that $(K_{c}^{\gamma}[V])\substF{\delta} = (K_{c'}^{\gamma'}[K[A]])\substF{\delta'}$.
  So $\theta(\CCc) = \theta(\CCc')$.
  \item If $\alpha = \questP{f}{A}{c'}$, so that $M$ is a callback $K[f\ V]$ for some context K, value $V$, and functional name $f$. Then  we have:
   \begin{itemize}
    \item $\gamma_O(f) = V'$ and $f:\sigma \rightarrow \sigma'$ 
     for some value $V$ and type $\sigma,\sigma'$; 
    \item $\gamma'_O = \gamma_O$, $\gamma'_P = \gamma_P \cdot \gamma_A \cdot [c' \mapsto K]$,
    $\xi' = \xi \cdot [c' \mapsto c]$ and $\phi' = \phi \uplus \dom{\gamma_A} \cdot \{c'\}$,
    with $(A,\gamma_A) \in \AVal{V}{\sigma}$;
    \item $h'_P = h_P$ and $h'_O = h_O$;
    \item $M' = V' \ A$.
   \end{itemize} 
    We conclude using these and the fact that:
   \begin{itemize} 
    \item $K_{c'}^{\gamma} = K_{c}^{\gamma}[K]$;
    \item $\gamma_O(f) = V'$;
    \item $A\substF{\gamma_A} = V$;
   \end{itemize} 
  that $(K_{c}^{\gamma}[K[f\ V]])\substF{\delta} = (K_{c'}^{\gamma'}[V' A])\substF{\delta'}$.
  So $\theta(\CCc) = \theta(\CCc')$.
 \end{itemize}
\end{proof}

\begin{definition}
 Taking $\CCc,\CCc'$ two composite configuration, we write $\CCc \rightsquigarrow \CCc'$ 
 when there exists a trace $\tr$ of actions (without any $\tau$-actions)
 s.t. $\CCc \ired{\tr \cdot \tau} \CCc'$.
\end{definition}

\begin{lemma}
 \label{thm:hosc:red-iff-iredcomp}
 The configuration transformation $\theta$ is a functional bisimulation between 
 the transition system over composite configurations $(CompConf,\rightsquigarrow)$ 
 and the operational transition system $(\Lambda \times \Heap,\red)$, 
 that is, for all valid composite configuration $\CCc$:
 \begin{itemize}
  \item for all composite configuration $\CCc'$, 
  if $\CCc \rightsquigarrow \CCc'$ then $\theta(\CCc) \red \theta(\CCc')$;
  \item for all pairs $(N,h)$ formed by a term an a heap $h'$,
  if $\theta(\CCc) \red (N,h')$ then
  there exists a valid composite configuration $\CCc'$ s.t. 
  $\CCc \rightsquigarrow \CCc'$ and $(N,h') = \theta(\CCc')$
 \end{itemize}
\end{lemma}

\begin{proof}
We write:
\begin{itemize}
 \item $\CCc$ as $\conf{M,c,\gamma_P,\gamma_O,\xi,\phi,h_P,h_O}$;
 \item $\gamma$ for $\gamma_P \cdot \gamma_O$;
 \item $\delta$ for the idempotent substitution associated to $\gamma_P \cdot \gamma_O$;
 \item $\theta(\CCc)$ as $(K_{c}^{\gamma}[M])\substF{\delta},h)$
 with $h = (h_P \cdot h_O)\substF{\delta}$.
\end{itemize}

We first suppose that $\CCc \rightsquigarrow \CCc'$, i.e. there exists 
a trace $\tr$ of actions (without any $\tau$)
and a composite configurations $\CCc_1$
s.t. $\CCc \ired{\tr} \CCc_1 \ired{\tau} \CCc'$.
From Lemma~\ref{lm:ired-stable-red}, we get that $\theta(\CCc) = \theta(\CCc_1)$.

Without loss of generality, we suppose the composite configuration 
$\CCc_1$ is $P$-active.
We write $\CCc'$ as $\conf{M',c',\gamma'_P,\gamma'_O,\phi',h'_P,h_O}$
and $\CCc_1$ as $\conf{M_1,c_1,\gamma'_P,\gamma'_O,\phi',h_P,h_O}$,
so that we have $(M_1,c_1,h_P) \ered (M',c',h'_P)$.

From Lemma~\ref{thm:ered-mplies-fred}, writing $\delta'$ for the 
idempotent substitution associated to $\gamma'_P \cdot \gamma'_O$,
and $\delta'_C$ for its restriction to the domain of continuation names,
one has that $(K_{c_1}^{\gamma'}[M]\substF{\delta'_C},h_P\substF{\delta'_C}) \red 
(K_{c'}^{\gamma'}[M']\substF{\delta'_C},h'_P\substF{\delta'_C})$.
Extending the heap with $h_O$ and the substitution to $\delta'$, we get that
$(K_{c_1}^{\gamma'}[M]\substF{\delta'},h) \red (K_{c'}^{\gamma'}[M']\substF{\delta'},
 (h'_P \cdot h_O)\substF{\delta'})$, i.e. $\theta(\CCc_1) \red \theta(\CCc')$.\\[1em]


 \noindent
 Now, we suppose that there exists a term $N$ and a heap $h'$ s.t. $\theta(\CCc) \red (N,h')$.
 From Lemma~\ref{lm:decomp-fred}, 
 there is three possible cases for the reduction $\theta(\CCc) \red (N,h')$:
 \begin{itemize}
  \item Either $(M,c,h_P \cdot h_O)$ is reducible.
  Without loss of generality, we suppose the composite configuration 
  $\CCc$ is $P$-active, so that $(M,c,h_P)$ is reducible.
  Then there exists $(M',c',h'_P)$ s.t.:
  \begin{itemize}
   \item $(M,c,h_P) \ered (M',c',h'_P)$;
   \item $N = (K_{c'}^{\gamma'}[M'])\substF{\delta}$;
   \item $h' = (h'_P \cdot h_O)\substF{\delta}$.
  \end{itemize}
  So we take $\CCc'=\conf{M',c',\gamma_P,\gamma_O,\xi,\phi,h'_P,h_O}$ so that $\CCc \ired{\tau} \CCc'$.
  \item Or $M$ is a callback:
  \begin{itemize}
   \item $M = K[f\ V]$ for some context $K$, value $V$, and functional name $f$;
   \item $\delta(f)$ is a $\lambda$-abstraction that we write $\lambda x.P$ (with $x \notin \dom{\delta}$);
   \item $N = (K_{c}^{\gamma}[K[P\subst{x}{V}]])\substF{\delta}$;
   \item $h' = h$;
  \end{itemize}
  
  From Lemma~\ref{lm:quests-ired}, there exists a functional name $g$, an abstract value $A_1$,
  a composite configuration $\CCc_1$ and a trace $\tr$ formed by questions s.t.:
  \begin{itemize}
   \item $\gamma(g)$ is a $\lambda$-abstraction $\lambda x.\hat{P}$;
   \item $\delta(f) = \delta(g)$;
   \item $\CCc \ired{\tr} \CCc_1$;
   \item $\CCc_1$ can be written as $\conf{g \ A_1,c_1,\gamma_P \cdot \gamma_{1,P},
     \gamma_O \cdot \gamma_{1,O},\phi \uplus \dom{\gamma_{1,P}},h_P,h_O}$;
   \item $A_1\substF{\delta_1} = V$, with $\delta_1$
   the idempotent substitution associated to $\gamma_{1,P} \cdot \gamma_{1,O}$;
  \item $K_{c_1,c}^{\gamma_1} = K$.
  \end{itemize}
  
  Without loss of generality, we suppose the composite configuration 
  $\CCc_1$ is $P$-active. Then we have:
  \[\begin{array}{lll}
  \CCc \ired{\tr} \CCc_1 & \ired{\questP{g}{A_2}{c_2}} & 
    \overbrace{\conf{(\lambda x.\hat{P}) \ A_2,c_2,\gamma_{2,P},\gamma_O\cdot\gamma_{1,O},
     \phi_2,h_P,h_O}}^{\CCc_2} \\
                   & \ired{\ \tau\ \ } & 
    \underbrace{\conf{\hat{P}\subst{x}{A_2},c_2,\gamma_{2,P},\gamma_O \cdot \gamma_{1,O},
    \phi_2,h_P,h_O}}_{\CCc'}
    \end{array}\]
  with $\gamma_{2,P} = \gamma_{P} \cdot \gamma_{1,P} \cdot \gamma_{A_2} \cdot [c_2 \mapsto (\bullet,c_1)]$
  and $A_2 \substF{\gamma_{A_2}} = A_1$.
  From Lemma~\ref{lm:ired-stable-red}, we have that $\theta(\CCc) = \theta(\CCc_2)$.
  
  We prove that $(\hat{P}\subst{x}{A_2})\substF{\delta_2} = P\subst{x}{V\substF{\delta}}$ from the fact that:
  \begin{itemize}
   \item $A_2 \substF{\delta_2} = V\substF{\delta}$ since 
    $A_1\substF{\delta_1} = V$ and $A_1 = A_2 \substF{\gamma_{A_2}}$;
   \item $\hat{P}\substF{\delta} = P$ since $\delta(f) = \delta(g)$, $\delta(f) = \lambda x.P$ 
   and $\gamma(g) = \lambda x.\hat{P}$.
  \end{itemize}
  
  Finally, from $K_{c_1,c}^{\gamma_1} = K$ and $\gamma_2(c_2) = (\bullet,c_1)$,
  we get that $K_{c_2}^{\gamma_2} = K_{c}^{\gamma}[K]$.
  So $\theta(\CCc') = (N,h)$. 

  \item Or $M$ is a value $V$ and $K_{c}^{\gamma}$ an evaluation context larger than $\bullet$.
  Then there exists a continuation name $c_1$ s.t.:
  \begin{itemize}
   \item $c \orderCont{\gamma}^{*} c_1$
   \item $K_{c,c_1}^{\gamma} = \bullet$.
   \item $\gamma(c_1) = K$ with $K$ an evaluation context larger than $\bullet$;
  \end{itemize} 
  From Lemma~\ref{lm:anss-ired}, there exists an abstract value $A_1$,
  a composite configuration $\CCc_1$ and a trace $\tr$ formed by answers s.t.:
  \begin{itemize}
   \item $\CCc \ired{\tr} \CCc_1$;
   \item $\CCc_1$ can be written as $\conf{A_1,c_1,\gamma_P \cdot\gamma_{1,P},
    \gamma_O \cdot\gamma_{1,O},\phi\uplus\dom{\gamma_{1,O}},h_P,h_O}$;
   \item $A_1\substF{\delta_1} = V$, with $\delta_1$
   the idempotent substitution associated to $\gamma_{1,P} \cdot \gamma_{1,O}$;
  \end{itemize}
  Without loss of generality, we suppose the composite configuration 
  $\CCc_1$ is $P$-active. Then we have:
  \[\begin{array}{lll}
   \CCc \ired{\tr} \CCc_1 & \ired{\ansP{c_1}{A_2}} & 
    \overbrace{\conf{K[A_2],c_2,\gamma_{2,P},\gamma_{2,O},\phi_2,h_P,h_O}}^{\CCc_2} \\
    \end{array}\]
  with $\xi(c_1) = c_2$,$\gamma_{2,P} = \gamma_{1,P} \cdot \gamma_{A_2}$ and 
  $A_2 \substF{\gamma_{A_2}} = A_2$.
  
  From Lemma~\ref{lm:ired-stable-red}, we have that $\theta(\CCc) = \theta(\CCc_2)$.
  From $K_{c,c_1}^{\gamma} = \bullet$, we get that $K_{c}^{\gamma} = K_{c_1}^{\gamma_1}$, 
  so that $K_{c_2}^{\gamma_2}[K] = K_{c}^{\gamma}$.
  Since $K$ is larger than $\bullet$, $K[A_2]$ cannot be a value,
  so from Lemma~\ref{lm:decomp-ered} we have that:
  \begin{itemize}
   \item either $(K[A_2],c_2,h_P)$ is reducible, 
   and we conclude using a similar reasoning as in the first case, on $\CCc_2$.
   \item or $K[A_2]$ is a callback, and we conclude  using a similar reasoning as in the second case, 
   on $\CCc_2$.
 \end{itemize}  
 \end{itemize}
\end{proof}

\begin{corollary}
 \label{corol:hosc:ter-iff-ter}
 Taking $\CCc$ a valid composite configurations,
 $\CCc \ltster{}$ iff $\theta(\CCc) \opredter$.
\end{corollary}

\begin{proof}
 We write $\CCc$ as $\conf{M,c,\gamma_P,\gamma_O,\xi,\phi,h_P,h_O}$.

 We first prove that if $\CCc \ltster{}$ then $\theta(\CCc) \opredter$.
 From $\CCc \ltster{}$, we get the existence of a sequence of reductions 
 $\CCc \rightsquigarrow^{*} \overbrace{\conf{\unit,\finalName,\gamma_{f,P},\gamma_{f,O},
  \phi_f,h_{f,P},h_{f,O}}}^{\CCc_f}$.
 We reason by induction over the length of this reduction.
 \begin{itemize}
  \item if $\CCc = \CCc_f$, then
  $\theta(\CCc) = (\unit,\_)$ since $M = \unit$ and $c = \finalName$ so that $K_{c}^{\gamma_P \cdot \gamma_O} = \bullet$.
  \item if there exists a composite configuration $\CCc'$
  s.t. $\CCc \rightsquigarrow \CCc' \rightsquigarrow^{*} \CCc_f$,
  then by induction hypothesis $\theta(\CCc') \opredter$,
  and from Theorem~\ref{thm:hosc:red-iff-iredcomp} one has that $\theta(\CCc) \red \theta(\CCc')$,
  so that $\theta(\CCc) \opredter$.
 \end{itemize}
 
  We now prove that if $\theta(\CCc) \opredter$ then $\CCc \ltster{}$.
  From $\theta(\CCc) \opredter$ we get the existence of $(\unit,h)$ s.t.
  $\theta(\CCc) \red^{*} (\unit,h)$.
  We reason by induction over the length of this reduction.
  \begin{itemize}
   \item if the reduction is empty, then $\theta(\CCc) = (\unit,h)$.
    So necessarily $M = \unit$ and $K_{c,\finalName}^{\gamma} = \bullet$.
    Then from Lemma~\ref{lm:anss-ired},
    $\theta(\CCc) \opredter$.
   \item if there exists $(M',h')$ s.t. $\theta(\CCc) \red (M',h') \red^{*} (\unit,h)$,
    then from Theorem~\ref{thm:hosc:red-iff-iredcomp},
    there exists a configuration $\CCc'$ s.t. $\theta(\CCc) \red \theta(\CCc')$
    and $\theta(\CCc') = (M',h')$.
    Then by induction hypothesis, since $\theta(\CCc') \red^{*} (\unit,h)$,
    we get that $\theta(\CCc') \opredter$, so that $\theta(\CCc) \opredter$.
  \end{itemize}
\end{proof}

\begin{corollary}
 \label{corol:hosc:err-iff-err}
 Taking $\CCc$ a valid composite configurations,
 $\CCc \ltstererr{}$ iff $\theta(\CCc) \opredtererr$.
\end{corollary}

Finally, we can prove Lemma~\ref{lem:cor}

\begin{lemma}[Correctness]
Let $\seq{\Gamma}{M:\tau}$ be a cr-free $\HOSC$ term,
let $\Sigma,h,K,\gamma$ be as above,
$(\vec{A_i},\vec{\gamma_i})\in \AVal{\gamma}{\Gamma}$,
and $c:\tau$ ($c\not\in\tern$).
Then
\begin{itemize}
\item 
$(K[M\substF{\gamma}],h)\opredtererr$ iff there exist $t,c'$ such that 
$t\in \TrR{\HOSC}{\cconf{M}{\rho_{\vec{A_i}},c}}$ and $t^\bot\, 
\questP{\errn}{()}{c'} \in \TrR{\HOSC}{\cconf{h, K,\gamma}{\vec{\gamma_i}, c} }$.
\item 
$(K[M\substF{\gamma}],h)\opredter$ iff there exist $t,A,\tau$ such that 
$t\in \TrR{\HOSC}{\cconf{M}{\rho_{\vec{A_i}},c}}$ and $t^\bot\, 
\ansP{\tern_{\tau'}}{A} \in \TrR{\HOSC}{\cconf{h, K,\gamma}{\vec{\gamma_i}, c}}$.
\end{itemize}
Moreover, $t$ must satisfy $\nu(t)\cap (\circ\cup\{\errn\})=\emptyset$.
\end{lemma}

\begin{proof}
 Let $\yy \in \{\ter,\err\}$.
Note that $(K[M\substF{\gamma}],h)\opredobs{\yy}$ iff $\theta(\mergeConf{\cconf{M}{\rho_{\vec{A_i}},c}}
   {\cconf{h, K,\gamma}{\vec{\gamma_i}}}) \opredobs{\yy}$.
 From Corollary~\ref{corol:hosc:ter-iff-ter} and \ref{corol:hosc:err-iff-err}, 
 this is equivalent to the existence of  a trace $\tr$ such that
 $(\mergeConf{\cconf{M}{\rho_{\vec{A_i}},c}}
   {\cconf{h, K,\gamma}{\vec{\gamma_i}}}) \ltsobs{\yy}{\tr}$.
By Lemma~\ref{lm:hosc:ired-iff-iredcomp}, this is the same as
 $(\cconf{M}{\rho_{\vec{A_i}},c}|\cconf{h, K,\gamma}{\vec{\gamma_i}})
 \semobs{\yy}{\tr}$, which implies the Lemma. 
\end{proof}

\section{Additional material for Section~\ref{sec:goschosc} (GOSC[HOSC])}

\subsection{Proof of Lemma~\ref{lem:pvis} (visibility)}
We write $\CC\ired{t}\CC'$ to say that there exists a sequence of transitions from $\CC$ to $\CC'$ such that
the collected labels, including $\tau$ transitions, give a trace $t$. The proof is based on an auxiliary lemma (Lemma~\ref{lem:visaux}),
which generalizes P-visibility to configurations, enabling an inductive proof.

\spnewtheorem*{lemma*}{Lemma}{\bfseries}{\rmfamily}
\begin{lemma*}[Original Statement of Lemma~\ref{lem:pvis}]
Let $\CC_O=\cconf{h, K,\gamma}{\vec{\gamma_i}, c}$, where $h,K,\gamma$ are from $\GOSC$, 
and $(\vec{A_i},\vec{\gamma_i})\in \AVal{\gamma}{\Gamma}$.
All traces in  $\TrReven{\HOSC}{\CC_O}$ are P-visible.
\end{lemma*}


\begin{proof}
Suppose $\CC_O\iRed{a_1\cdots a_{2i+1}} \CC$ and $\CC\ired{\tau^\ast} \CC'\ired{a_{2i+2}} \CC''$.
By Lemma~\ref{lem:visaux}, $\CC'=\conf{M',c',\cdots}$ with $\nu (M',c')\subseteq\pav{a_1\cdots a_{2i+1}}$.
Because the O-names in $a_{2i+2}$ come from $\nu (M',c')$, P-visibility follows.
\end{proof}

\begin{lemma}\label{lem:visaux}
Suppose $\CC_O\ired{a_1\cdots a_k} \CC$.
\begin{enumerate}
\item If $\CC=\conf{\gamma,\xi,\phi,h}$ then, for any $n\in\dom{\gamma}$, if $n$ was introduced in $a_{2i}$ ($0\le i\le k/2$) then 
$\nu(\gamma(n))\subseteq \pav{a_1\cdots a_{2i-1}}$ and if $n\in\CNames$ then $\xi(n)\in\pav{a_1\cdots a_{2i-1}}$
(introduced in $a_0$ is taken to mean $\errn,\tern$ and $\pav{a_1\cdots a_0}$ stands for $\{\errn,\tern\}$).
\item If $\CC=\conf{M,c,\gamma,\xi,\phi,h}$ then $\nu(M,c)\subseteq \pav{a_1\cdots a_k}$ and all of the conditions listed above hold.
\end{enumerate}
\end{lemma}
\begin{proof}
By induction on the number of transitions between $\CC_O$ and $\CC$, including $\tau$-transitions.

The base case is  $\CC_O=\CC$. The Lemma then holds because 
$\nu(\gamma)\subseteq \{\errn\}$, $\xi(c)=\tern$, 
and $\pav{a_1\cdots a_0}=\{\errn,\tern\}$.

Suppose $\CC_O\ired{a_1\cdots a_k}\CC'$ and $\CC_O\ired{t} \CC\ired{x} \CC'$, where $t$ is a trace and $x$ is an action or $x=\tau$.
\begin{itemize}
\item If $x=\tau$ then $\gamma,\xi$ do not change during the transition
 and the reduction does not generate new names by Lemma~\ref{lem:goscinv}.
Hence, the Lemma follows from IH.
\item Suppose $x$ is an O-action, i.e. $x=a_k$.
Then $\CC'=\conf{M',c',\gamma',\xi',\phi',h'}$ and $\CC=\conf{\gamma',\xi',\phi'\setminus A, h'}$.
By IH for $\CC$, all the conditions for $\gamma',\xi'$ hold, so it remains to check $\nu(M',c')$.
\begin{itemize}
\item If $x=\ansO{c''}{A''}$ then $\nu(M',c')=\nu(\gamma'(c'')[A''],\xi'(c''))$. By IH for $\CC,c''$, 
assuming $c''$ was introduced in $a_{2i}$, we get
$\nu(M',c')\subseteq \pav{a_1\cdots a_{2i-1}}\cup \nu(A'') = \pav{a_1\cdots a_k}$. 
\item  If $x=\questO{f}{A''}{c''}$ then $\nu(M',c')=\nu(\gamma'(f)[A''],c'')$. By IH for $\CC,f$, 
assuming $f$ was introduced in $a_{2i}$, we get
$\nu(M',c')\subseteq \pav{a_1\cdots a_{2i-1}}\cup \nu(A'') \cup \{c''\} = \pav{a_1\cdots a_k}$. 
\end{itemize}
\item Suppose $x$ is a P-action, i.e. $x=a_k$.
Then $\CC'=\conf{\gamma',\xi',\phi',h'}$.
\begin{itemize}
\item If $x=\ansP{c''}{A''}$ then $\CC=\conf{V,c'',\gamma'\setminus \nu(A''), \xi',\phi\setminus \nu(A''), h'}$.
By IH, $\gamma'\setminus \nu(A'')$ and $\xi'$ satisfy the Lemma. It suffices to check $\gamma'(n)$ for $n\in\nu(A'')$.
Observe that then $\nu(\gamma'(n))\subseteq \nu(V,c'')$ and, by IH for $\CC$, $\nu(V,c'')\subseteq \pav{a_1\cdots a_{k-1}}$, as required.
\item If $x=\questP{f}{A''}{c''}$ then $\CC=\conf{K[fV], c''',\gamma'\setminus X, \xi'\setminus\{c''\},\phi\setminus X, h'}$,
where $X=\nu(A'')\cup\{c''\}$.
By IH, $\gamma'\setminus X$ and $\xi'\setminus\{c''\}$ satisfy the Lemma. 
It suffices to check $\gamma'(n)$ for $n\in\nu(A'')$, $\gamma'(c'')$ and $\xi'(c'')$.
Observe that then 
$\nu(\gamma'(n))\cup\nu(\gamma'(c''))\cup\{\xi'(c'')\}\subseteq 
\nu(K[fV],c''')$ and, by IH for $\CC$, $\nu(K[fV],c''')\subseteq \pav{a_1\cdots a_{k-1}}$, as required.
\end{itemize}
\end{itemize}
\end{proof}

\subsection{Proof of Theorem~\ref{goscsound}}
    \begin{proof}
Suppose $\trsem{\GOSC}{\seq{\Gamma}{M_1}}\subseteq \trsem{\GOSC}{\seq{\Gamma}{M_2}}$.
Consider  $\Sigma, h,K,\gamma$ (as in the definition of $\ciupre{\GOSC}{err}$) such that $(K[M_1\substF{\gamma}],h)\opredtererr$.
In particular, $h,K,\gamma$ consist of $\GOSC$ syntax.
Suppose $(\vec{A_i},\vec{\gamma_i})\in\AVal{\gamma}{\Gamma}$ and $c:\tau$ ($c\not\in\circ$).
By Lemma~\ref{lem:cor} (left-to-right), there exist $t,c'$ such that
$t\in \TrR{\HOSC}{\cconf{M_1}{\rho_{\vec{A_i}},c}}$ and $t^\bot\, \questP{\errn}{()}{c'} \in \TrR{\HOSC}{\cconf{h, K,\gamma}{\vec{\gamma_i}, c} }$.
By  Lemma~\ref{lem:pvis}, $t^\bot\, \questP{\errn}{()}{c'}$ is P-visible. Thus, $t$ is O-visible and, by Lemma~\ref{lem:ovis} (right-to-left),  
$t\in\trsem{\GOSC}{\cconf{M_1}{\rho_{\vec{A_i}},c}}$.
From $\trsem{\GOSC}{\seq{\Gamma}{M_1}}\subseteq \trsem{\GOSC}{\seq{\Gamma}{M_2}}$, we get $t\in \TrR{\GOSC}{\cconf{M_2}{\rho_{\vec{A_i}},c}}$.
By Lemma~\ref{lem:ovis} (left-to-right), $t\in \TrR{\HOSC}{\cconf{M_2}{\rho_{\vec{A_i}},c}}$.
Because $t\in \TrR{\HOSC}{\cconf{M_2}{\rho_{\vec{A_i}},c}}$ and 
$t^\bot\, \questP{\errn}{()}{c'} \in \TrR{\HOSC}{\cconf{h, K,\gamma}{\vec{\gamma_i}, c} }$, by Lemma~\ref{lem:cor} (right-to-left), we
can conclude $(K[M_2\substF{\gamma}],h)\opredtererr$. Thus, $\seq{\Gamma}{M_1 \ciupre{\GOSC}{err} M_2}$.
  \end{proof}

\subsection{Proof of Lemma~\ref{goscdef}}

Lemma~\ref{hoscdef} follows from the lemma given below for $i=0$. 
Consider $h'=h_0$, $K'=\gamma_0(c)$, $\gamma'=\gamma_0\setminus c$.
We have $\nu(\img{\gamma_0},\img{h_0})\subseteq  \circ\uplus \{\errn\}$.
As names $\circ_\sigma$ can only
occur inside terms of the form $\cont(K',\circ_\sigma)$,
we can conclude that $(h',K',\gamma')= (h_\circ,K_\circ,\gamma_\circ)$, 
where $h,K,\gamma$ are from $\GOSC$.

\begin{lemma}\label{hoscaux}
Suppose $\phi\uplus\{\errn\}\subseteq\FNames$, $c\in \CNames$ and
$t=o_1 p_1\cdots o_n p_n$ is a  P-visible
$(\circ\uplus \{\errn\},\phi\uplus\{c\})$-trace starting with an O-action.
Given $0\le i\le n$, let $t_i = o_{i+1}p_{i+1} \cdots o_n p_n$.
There exist  passive configurations $\CC_i$ such that  $\Treven{\CC_i}$ consists 
of even-length prefixes of $o_{i+1}p_{i+1} \cdots o_n p_n$
(along with their renamings via permutations on $\Names$ that fix $\phi_i$).
Moreover, $\CC_i=\conf{\gamma_i, \xi_i, \phi_i, h_i}$ ($0\le i\le n$), where
\begin{itemize}
\item $\dom{\gamma_i}$ consists of $\phi\cup\{c\}$ and all names introduced by P in $o_1 p_1\cdots o_i p_i$;
\item $\img{\gamma_i}$ contains GOSC syntax;
\item $\nu({\gamma_i}(x)) \subseteq\pav{o_1 p_1\cdots o_i}$ if $x$ has been introduced in $p_i$ ($\phi\uplus\{c\}$ are deemed to
have been introduced in $p_0$ and we assume $\pav{o_1\cdots o_0} = \circ\uplus \{\errn\}$);
\item  for all $d\in\dom{\gamma_i}\cap\CNames$,
 if $d:\sigma_d$ and $d$ was introduced in $p_j$ then $\seq{}{\gamma_i(d):\sigma_d\rarr \sigma_j}$,
where $\topp{o_1\cdots o_j}:\sigma_j$;

\item $\dom{\xi_i}$ consists of $c$ and all continuation names introduced by P in $o_1 p_1\cdots o_i p_i$;

\item for all $d\in\dom{\xi_i}$, $\xi_i(d)=\topp{o_1\cdots o_j}$ if $d$ was introduced in $p_j$ (we regard $c$ as being introduced
in $p_0$ and define $\topp{o_1\cdots o_0}=\circ_{\tau'}$);

\item ${\phi_i}$ consists of $\circ\uplus \{\errn\}\uplus \phi\uplus\{c\}$ and all names introduced in $o_1 p_1\cdots o_i p_i$;

\item for all $0\le i\le n$, $h_i=\{\tick \mapsto i \}$, where $\tick:\reftype\,{\Int}$.
\end{itemize}
\end{lemma}

\begin{proof}
Note that the heap will consist of a single reference only, which will correspond to counting steps in the translation.
At every step of the translation, the value of the reference will be used to schedule the right actions and disable others.

The above description already specifies $\phi_i$, $\dom{\gamma_i}$, $\xi_i$ and $h_i$. To complete the definition of $\CC_i$,
it remains to specify the environments $\gamma_i$. Recall that, 
we need to define $\gamma_0(x)$ for $x\in \phi\cup\{c_P\}$ and, in other cases,
$\gamma_j(x)$ ($x\in\Names$) will be defined for all $j \ge i$ if $x$ was introduced by P in $p_i$. 
Recall also that once $\gamma_j(x)$ is defined, it never changes. Hence if $x$ was introduced by in $p_i$,
we will only specify $\gamma_i(x)$ on the understanding that $\gamma_{i'}(x)=\gamma_i(x)$ for all $i' > i$.

We define $\gamma_i(x)$ by induction using the reverse order of name introduction in $t$, i.e. when defining $\gamma_i(x)$
we will refer to $\gamma_{i'}(y)$, where $y$ is introduced in a later move in $t$. In particular, the names $\phi\cup\{c\}$ are deemed to
be introduced first. Once $\gamma_i(x)$ is defined, we will argue that $\nu(\gamma_i(x)) \subseteq\pav{o_1 p_1\cdots o_i}$.

\begin{itemize}
\item Suppose $f:\sigma_f\rarr\tau_f$ is a function name introduced by P in action $p_i$ ($1\le i\le n$) 
or $f\in\phi$, in which case we let $i=0$.
Consider all subsequent occurrences of $f$ in $t$: suppose $\ind{f}= \{ i < u \le n  \,\,|\,\, o_u = \questO{f}{A_u}{c_u} \}$,
i.e. $\ind{f}$ contains all the time points when it is necessary to respond to $\questO{f}{A'}{c'}$.
Then we let 
\[
\gamma_{i}(f)= \lambda x.\,  (\tick := !\tick+1); \ifte{(!\tick\in \mathcal{I}_{f})}{(\ass{x}{A_{!\tick}};M_{!\tick})}{\Omega},
\]
where $(\ass{x}{A_{!\tick}};M_{!\tick})$ is shorthand for code that performs case distinction on $!\tick$ and
directs reduction to $(\ass{x}{A_u};M_{u})$ for $u=!\tick\in \mathcal{I}_{f}$.
The term $\ass{x}{A_u}$ has been defined  earlier, so we specify $M_u$ ($u\in \ind{f}$), aiming to have $\seq{x:\sigma_f}{M_u:\tau_f}$
in each case.  $M_u$ will depend on the shape of $p_u$.
Note that if $\ind{f}=\emptyset$, i.e. $f$ is not used in $t$, then the construction degenerates to 
$\gamma_{i}(f)= \lambda x.\,  (\tick := !\tick+1); \Omega$.

\begin{description}
\item[$p_u=\ansP{c_u'}{A_u'}$] As $u>i$, $\gamma_u$ is already defined for all names in $A_u'$.
Let  $V={A_u'}\{\gamma_{u}\}$. Recall that $o_u = \questO{f}{A_u}{c_u}$.
We let 
\[
M_u = \callcc{(\,\,y.\,\, (\throwto{V}{\cont{(\bullet,c_u')}})\,\,[y/\cont{(\bullet,c_u)}]\,\,[\pi x/A_u]\,\,)}.
\]
\begin{itemize}
\item $[y/\cont{(\bullet,c_u)}]$  is meant to mimic the reversal of the reduction rule for $\callcc$: because after $o_u$
the continuation name in the active configuration will be $c_u$,  the then current continuation will be ${\cont{(\bullet,c_u)}}$.
Since all continuation names $c'$ are only ever used via the term $\cont{(\bullet,c')}$, the substitution $[y/\cont{(\bullet,c_u)}]$
will remove all occurrences of $c_u$ from $V$.
\item The substitution $[\pi x/A_u]$ has been defined before the first definability proof.
\end{itemize}
Note that, because of ${\mathrm{throw}}$, $M_u$ can  indeed be given type $\tau_f$.
Overall, the shape of $M_u$ guarantees the desired progression ($o_u p_u$) at time $u$
(the configuration will reduce to $(V,c_u', \cdots)$, to be followed by  $p_u=\ansP{c_u'}{A_u'}$).

Because we can assume $\nu(\gamma_u(x)) \subseteq  \pav{o_1\cdots o_u}$ for any $x$ introduced in $p_u$ (IH),
we have $\nu(V)=\nu {({A_u'}\{\gamma_{u}\})}\subseteq \pav{o_1\cdots o_u}$.
As all names introduced in $o_u$ will be substituted for, we have $\nu(M_u) \subseteq \pav{o_1\cdots o_i}\cup\{c_u'\}$.
However, by P-visibility, we have $c_u'\in\pav{o_1\cdots o_u}$, so either $c_u'=c_u$ or $c_u'\in\pav{o_1\cdots o_i}$.
Either way, we can conclude $\nu(M_u) \subseteq \pav{o_1\cdots o_i}$, i.e. $\nu(\gamma_i(f)) \subseteq  \pav{o_1\cdots o_i}$.

\item[$p_u=\questP{f'}{A_u'}{c_u'}$] 
As in the previous case, by IH, $\gamma_{u}$ is already defined for all names in $A_u'$ and $c_u'$.
Let $V={A_u'}\{\gamma_{u}\}$ and $K= \gamma_{u}(c_u')[\bullet]:\sigma_{c_u'}\rarr\sigma_j$, where $\topp{o_1\cdots o_u}:\sigma_j$ (IH).
Note that $\topp{o_1\cdots o_u}=c_u$ in this case, i.e. $\tau_f=\sigma_j$.
We let 
\[
M_u = K[f' V]\,[\pi x/A_u].
\]
The shape of $M_u$ then guarantees the right progression in the $u$th step $o_u p_u$ 
(after $o_u$ the LTS will reach a configuration of the form $(\gamma_{u}(c_u') [f'  V], \topp{o_1\cdots o_u},\cdots)$, 
from which $p_u=\questP{f'}{A_u'}{c_u'}$ can follow).

Because $\nu(V),\nu(\gamma_{u}(c_u'))\subseteq \pav{o_1\cdots o_u}$ and all names introduced in $o_u$ 
are substituted for above,
we have $\nu(M_u) \subseteq \pav{o_1\cdots o_i}\cup\{f'\}$. 
By P-visibility, $f'\in\pav{o_1\cdots o_i}$, so we can conclude that  
 $\nu(M_u) \subseteq \pav{o_1\cdots o_i}$, i.e. $\nu(\gamma_i(f)) \subseteq  \pav{o_1\cdots o_i}$.
\end{description}

\item 
Suppose now that $d:\sigma_d$  is a continuation name introduced by P in action $p_i$ ($1\le i\le n$), 
or $d=c$, in which case we let $i=0$.
Let us consider all subsequent occurrences of $d$ in $t$:
suppose $\ind{d}= \{ i < u \le n  \,\,|\,\, o_u = \ansO{d}{A_u} \}$.
Then we let 
\[
\gamma_{i}(d)= (\lambda x.\,  (\tick := !\tick+1); \ifte{(!\tick\in \mathcal{I}_{d})}{(\ass{x}{A_{!\tick}};M_{!\tick})}{\Omega})[\bullet]
\]
where the terms $M_u$ ($u\in\ind{c}$) are the same as in the previous case, though this time we aim for
$\seq{x:\tau_d}{M_u:\tau_j}$, where $d:\tau_d$ and $\topp{o_1\cdots o_u}: \tau_j$ (recall that $\xi_{i}(d) = \topp{o_1\cdots o_u}$).
As argued above, in the second case $M_u$ will have the required type and in the first case it can be forced thanks to throw.

Similarly, we can conclude that  $\nu(\gamma_{i}(d)) \subseteq \pav{o_1\cdots o_i}$.
\end{itemize}
This completes the definition of configurations. 
They evolve as required by construction, because the definition of $\gamma_i$ is compatible
with the evolution of the $\GOSC[\HOSC]$ LTS:
at each stage, the value of the clock $\tick$ is incremented and the corresponding term $M_u$ is selected.

It is is easy to check that the syntax used in the construction belongs to $\GOSC$ only.
\end{proof}

\subsection{Proof of Theorem~\ref{gosccomplete}}

\begin{proof}
We follow the same path as in the proof of Theorem~\ref{hosccomplete} except that, in this case, we will have 
$t,t_1\in \TrR{\GOSC}{\cconf{M_1}{\rho_{\vec{A_i}},c}}$.
Consequently, we can conclude that $t_2 =t_1^{\bot}\,\questP{\errn}{()}{c'}$ is P-visible and invoke Lemma~\ref{goscdef} (instead of Lemma~\ref{hoscdef})
to obtain $C_O$ that corresponds to $h,K,\gamma$ from $\GOSC$. 
Because $k,K,\gamma$ are in $\GOSC$, we can then appeal to the assumption $\seq{\Gamma}{M_1 \ciupre{\GOSC}{err} M_2}$
and complete the proof like for Theorem~\ref{hosccomplete}.
\end{proof}


\section{Additional material for Section~\ref{sec:hoshosc} (HOS[HOSC])}

\subsection{Proof of Lemma~\ref{lem:pbracket}}

To enable a proof by induction we generalize the Lemma as follows.
\begin{lemma}
Consider $\CC_O=\cconf{h, K,\gamma}{\vec{\gamma_i}, c}$, where $h,K,\gamma$ are from $\HOS$ 
and $(\vec{A_i},\vec{\gamma_i})\in \AVal{\gamma}{\Gamma}$. Let $t\in\TrR{\HOSC}{\CC_O}$ and suppose $\CC_O\ired{t'}\CC$.
\begin{itemize}
\item If $t'$ is of odd length then $\CC=\conf{M,c',\cdots}$ and $c'=\topp{t'}$.
\item If $t'$ is of even length and $t' = t \questP{f}{A}{c'}$ then $\CC=\conf{\cdots,\xi,\cdots}$ and $\xi(c')=\topp{t}$.
\item If $t'$ is of even length and $t'=t \ansP{c'}{A}$ then $c'=\topp{t}$.
\end{itemize}
\end{lemma}
\begin{proof}
By induction on the number of transitions in $\CC_O\ired{t}\CC$. In the base case (no transitions) the Lemma holds vacuously.

Note that the Lemma is preserved by silent transitions ($t$ is of odd length then) by Lemma~\ref{lem:hosinv}.

Suppose $\CC_O\ired{t} \CC\ired{a}\CC'$.
\begin{itemize}
\item The even-length cases follow immediately the odd-length case due to the shape of LTS rules.
\item Suppose $t' = t a$ is of odd length.
\begin{itemize}
\item If $a=\questO{f}{A}{c''}$ then  $\topp{t'} = c''$ and $c'=c''$, so the Lemma holds.
\item If $a=\ansO{c''}{A}$ then $c'=\xi(c'')$. 
\begin{itemize}
\item If $c''=c$ then $c'=\tern$ and indeed $\topp{t'}=\tern$.
\item Otherwise $\topp{t'} = \topp{t''}$, where $c''$ is introduced by an action (question) after $t''$. Then, by IH, $\xi(c'')=\topp{t''}$.
Because $\topp{t''}=\topp{t'}$, we get $c'=\topp{t'}$, as required.
\end{itemize}
\end{itemize}
\end{itemize}
\end{proof}

\subsection{Proof of Theorem~\ref{hossound}}

      \begin{proof}
Suppose $\trsem{\HOS}{\seq{\Gamma}{M_1}}\subseteq \trsem{\HOS}{\seq{\Gamma}{M_2}}$.
Consider  $\Sigma, h,K,\gamma$ (as in the definition of $\ciuapperr^\HOS$) such that $(K[M_1\substF{\gamma}],h)\opredtererr$.
In particular, $h,K,\gamma$ consist of $\HOS$ syntax.
Suppose $(\vec{A_i},\vec{\gamma_i})\in\AVal{\gamma}{\Gamma}$ and $c:\sigma$ ($c\neq\tern$).
By Lemma~\ref{lem:cor} (left-to-right), there exist $t,c'$ such that
$t\in \TrR{\HOSC}{\cconf{M_1}{\rho_{\vec{A_i}},c}}$ and $t^\bot\, \questP{\errn}{()}{c'} \in \TrR{\HOSC}{\cconf{h, K,\gamma}{\vec{\gamma_i}, c} }$.
By  Lemma~\ref{lem:pbracket}, $t^\bot\, \questP{\errn}{()}{c'}$ is P-bracketed. Thus, $t$ is O-bracketed and, by Lemma~\ref{lem:obracket} (right-to-left),  
$t\in\trsem{\HOS}{\cconf{M_1}{\rho_{\vec{A_i}},c}}$.
From $\trsem{\HOS}{\seq{\Gamma}{M_1}}\subseteq \trsem{\HOS}{\seq{\Gamma}{M_2}}$, we get $t\in \TrR{\HOS}{\cconf{M_2}{\rho_{\vec{A_i}},c}}$.
By Lemma~\ref{lem:obracket} (left-to-right), $t\in \TrR{\HOSC}{\cconf{M_2}{\rho_{\vec{A_i}},c}}$.
Because $t\in \TrR{\HOSC}{\cconf{M_2}{\rho_{\vec{A_i}},c}}$ and 
$t^\bot\, \questP{\errn}{()}{c'} \in \TrR{\HOSC}{\cconf{h, K,\gamma}{\vec{\gamma_i}, c} }$, by Lemma~\ref{lem:cor} (right-to-left), we
can conclude $(K[M_2\substF{\gamma}],h)\opredtererr$. Thus, $\seq{\Gamma}{M_1 \ciuapperr^\HOS M_2}$.
  \end{proof}

\subsection{Proof of Lemma~\ref{lem:hosdef}}

\begin{proof}
We take advantage of the definability result for $\HOSC$ (Lemma~\ref{hoscaux}) and argue that, for P-bracketed traces, continuation-related syntax can be eliminated.
This will follow from the careful integration of  $\topp{}$ in the construction.

Indeed, the only place where ``throw" is needed in the construction is to transition from configuration $D_i$ to $E_i$.
The second component (current continuation) in $D_i$ is equal to $\topp{o_1\cdots o_{i+1}}$, whereas
the second component in $E_i$ in this case is $c_O^{j'}$. For a P-bracketed trace, the two continuation names will be the same (Definition~\ref{def:pbra}).
Consequently, the use of ``throw" in this case is trivial: it will have the form 
$(\throwto{V_A}{\cont{(\bullet,c)}}, c,\cdots)$, where $c=c_O^{j'}$, because the continuation in $\cor{j'}$ is $\cont{(\bullet,c_O^{j'})}$ by one of our invariants.
This use of ``throw" can be replaced simply by $(V_A,c,\dots)$, i.e. occurrences of ``throw" can be eliminated.

Next, one observes that references to continuations ($\cor{j}$) are redundant as well, because they are only used in 
connection with ``throw", and we already know that ``throw" is redundant.

Finally, ``callcc" is redundant, because the only purpose of invoking it was to record continuations in a reference, 
and we know from the previous point that such references will not be needed.

Overall this yields a construction that involves  only $\HOS$ syntax.
\end{proof}

\subsection{Proof of Theorem~\ref{hoscomplete}}

  \begin{proof}
We follow the same path as in the proof of Theorem~\ref{hosccomplete} except that, in this case, we have 
$t,t_1\in \TrR{\HOS}{\cconf{M_1}{\rho_{\vec{A_i}},c}}$.
Consequently, we can conclude that $t_2 =t_1^{\bot}\,\questP{\errn}{()}{c'}$ is P-bracketed and invoke Lemma~\ref{hosdef} (instead of Lemma~\ref{hoscdef})
to obtain $C_O$ that corresponds to $h,K,\gamma$ from $\HOS$. 
Because $k,K,\gamma$ are in $\HOS$, we can appeal to the assumption $\seq{\Gamma}{M_1 \ciuapperr^\HOS M_2}$
and complete the proof like for Theorem~\ref{hosccomplete}.
\end{proof}


\section{Additional material for Section~\ref{sec:goshosc} (GOS[HOSC])}
\label{apx:gos}

\subsection{GOS[HOSC] LTS}
\begin{figure}
\[
\begin{array}{l|l@{}ll}
 (P\tau) & \conf{M,c,\gamma,\xi,\phi,h,\ViewF}  &\quad  \ired{\ \tau \ }
 & \conf{N,c',\gamma,\xi,\phi,h',\ViewF} \\
 & \multicolumn{3}{l}{\text{ when } (M,c,h) \ered (N,c',h')}\\
 (PA) & \conf{V,c,\gamma,\xi,\phi,h,\ViewF} 
 & \ired{\ansP{c}{A}} 
 & \conf{\gamma \cdot \gamma', \xi,
     \phi\uplus\nu(A),h, \ViewF, \ViewF(c) \uplus \nu(A),c'} \\
           & \multicolumn{3}{l}{\text{ when } c:\sigma,\,  (A,\gamma') \in \AVal{V}{\sigma},\,\xi(c)=c' } \\  
 (PQ) & \conf{K[fV],c,\gamma,\xi,\phi,h,\ViewF} 
 & \ired{\questP{f}{A}{c'}} 
 & \conf{  \gamma \cdot\gamma'\cdot [c' \mapsto K],
    \xi\cdot[c' \mapsto c],\phi\uplus\phi',h, \ViewF, \ViewF(f) \uplus \phi' ,c'} \\
 & \multicolumn{3}{l}{\text{ when } f:\sigma \rarr \sigma', \,(A,\gamma') \in \AVal{V}{\sigma}, \, c':\sigma' \text{ and } \phi'=\nu(A)\uplus\{c'\}}\\
 (OA) & \conf{\gamma,\xi,\phi,h,\ViewF,\View,c''} 
 & \ired{\ansO{c}{A}} 
 & \conf{K[A],c',\gamma,\xi,
    \phi \uplus \nu(A),h, \ViewF\cdot[\nu(A)\mapsto \View]} \\
 & \multicolumn{3}{l}{\text{ when } c \in \View,\, c=c'',\, c:\sigma,\, A:\sigma,\, \gamma(c) = K,\, \xi(c)=c'}\\
 (OQ) & \conf{\gamma,\xi,\phi,h, \ViewF,\View,c''} 
 & \ired{\questO{f}{A}{c}}
 & \conf{V A,c,\gamma,\xi\cdot[c\mapsto c''],   \phi \uplus \phi',h,\ViewF \cdot [\phi' \mapsto \View]} \\
 & \multicolumn{3}{l}{\text{ when } f \in \View,\, f:\sigma \rarr \sigma',\, A:\sigma,\, c:\sigma',\,
  \gamma(f) = V   \text{ and } \phi' = \nu(A)\uplus \{c\}}\\
\cutout{
 (IOQ) & \conf{\Gamma \vdash M:\tau} 
 & \ired{\questOinit{\vec{A}}{c}} 
 & \conf{M\{\delta\},c,\emptymap,\emptymap             ,[c \mapsto \varnothing],
    \phi \uplus \{c\},\emptymap} \\
 & \multicolumn{3}{l}{\text{ when } (\vec{A},\delta,\phi) \in \sem{\Gamma}\}}}\\
  \multicolumn{4}{l}{\text{Given $N\subseteq\Names$, $[N\mapsto\View]$ stands for the map $[n\mapsto\View\,|\, n\in N]$.}}
  \end{array}
 \]

\caption{$\GOS[\HOSC]$ LTS}\label{fig:gos}
\end{figure}

Recall that, given a $\Gamma$-assignment $\rho$, term $\seq{\Gamma}{M:\tau}$ and  $c\in\CNames_\tau$, 
the active configuration  $\cconf{M}{\rho,c}$ was defined by
 $\cconf{M}{\rho,c} = \conf{M\{\rho\}, c, \emptyset, \emptyset, \nu(\rho)\cup\{c\}, \emptyset}$.
 We need to upgrade it to the LTS by initializing the new components:
 $\cconf{M,\mathit{vis,bra}}{\rho,c}= \conf{M\{\rho\}, c, \emptyset, [c\mapsto \bot], \nu(\rho)\cup\{c\}, \emptyset,\emptyset}$.
 \begin{definition}
 The \boldemph{$\GOS[\HOSC]$ trace semantics} of a cr-free HOSC term $\seq{\Gamma}{M:\tau}$ is defined to be
 $\trsem{\GOS}{\seq{\Gamma}{M:\tau}} = \{  ((\rho,c) ,t)\,|\, \textrm{$\rho$ is a $\Gamma$-assignment},\,c:\tau,\,t\in 
 \TrR{\GOSC}{\cconf{M,\mathit{vis,bra}}{\rho,c}} \}$.
   \end{definition}
 By construction and from the $\GOSC$ and $\HOS$ sections, it follows that
\begin{lemma}\label{lem:ovis}
$t\in \TrR{\GOS}{\cconf{M,\mathit{vis,bra}}{\rho,c}}$ iff $t\in  \TrR{\HOSC}{\cconf{M}{\rho,c}}$ and $t$ is O-visible and O-bracketed.
\end{lemma}

 \begin{lemma}[Definability]\label{gosdef}
Suppose $\phi\uplus\{\errn\}\subseteq\FNames$
and $t$ is an even-length P-bracketed and P-visible $( \{\circ_{\tau'},\errn\},\phi\uplus\{c\})$-trace starting with an O-action.
There exists a passive configuration $\CC$ such that the even-length traces
 $\TrR{\HOSC}{\CC}$ are exactly the even-length prefixes of $t$
(along with all renamings that preserve types and $\phi\uplus \{c,\circ_{\tau'},\errn\}$).
Moreover, $\CC=\conf{\gamma\cdot [c\mapsto K], \{ c\mapsto \tern_{\tau'} \}, \phi\uplus \{c,\circ_{\tau'},\errn\},h}$,
where 
$h,K,\gamma$ are built from $\GOS$ syntax.
\end{lemma}

\begin{proof}
Follows from the argument for $\GOSC$. We first observe that throw is needed before answer actions to adjust the continuation
from $\topo{o_1\cdots o_i}$. With P-bracketing there is no need for such adjustments. Consequently, we do not need $\callcc$,
which was used to generate continuations to be used in future adjustments.
\end{proof}

\end{document}